\documentclass[preprint]{elsarticle}

\journal{arXiv}

\usepackage{float}
\floatstyle{boxed}
\restylefloat{table}
\restylefloat{figure}
\usepackage[section,above,below]{placeins}
\usepackage{microtype}
\usepackage{amsmath,amsfonts,amssymb}
\allowdisplaybreaks
\usepackage{stmaryrd}
\SetSymbolFont{stmry}{bold}{U}{stmry}{m}{n}
\usepackage[braket]{qcircuit}
\usepackage{proof}
\usepackage{tikz}
\usepackage{multicol,multirow}
\usepackage{graphicx}
\usepackage{array,longtable}
\usepackage{wrapfig}

\newcolumntype{C}{>{$}c<{$}}  % automatic math mode, centered
\newcolumntype{R}{>{$}r<{$}}  % automatic math mode, centered
\newcolumntype{L}{>{$}l<{$}}  % automatic math mode, centered

\newtheorem{theorem}{Theorem}[section]
\newtheorem{lemma}[theorem]{Lemma}
\newtheorem{corollary}[theorem]{Corollary}
\newdefinition{definition}[theorem]{Definition}
\newproof{proof}{Proof}

\newcommand\Span{\ensuremath{\mathsf{Span}}}
\newcommand\interp[1]{\llparenthesis #1\rrparenthesis}
\newcommand\size[1]{||#1||}
\newcommand\Red[1]{\mathsf{Red}(#1)}
\newcommand\SN{\mathsf{SN}}
\newcommand\lpl[1]{|#1|}
\newcommand\titulo[3][\scriptsize]{\rotatebox[origin=c]{90}{\parbox[t]{#2}{\centering #1{#3}}}}

\newcommand\recap[3]{\noindent {\bf #1 \ref{#2}.} \emph{#3}}
\newcommand\xrecap[4]{\noindent {\bf #1 \ref{#3} (#2).} \emph{#4}}
\newcommand\braket[2]{\langle{#1}|{#2}\rangle}
\newcommand\lra[1][1]{\longrightarrow_{\left(#1\right)}}
\newcommand\lrap{\lra[p]}
\newcommand\B{\ensuremath{\mathbb B}}
\newcommand\Bs{\ensuremath{\mathfrak B}}
\newcommand\gB{\ensuremath{\Psi}}
\newcommand\ite[3]{{#1}?{#2}\mathord{\cdot}{#3}}
\newcommand\pair[2]{({#1}+{#2})}
\newcommand\npair[2]{({#1}-{#2})}
\newcommand\bqtypes{\ensuremath{\mathbf B}}
\newcommand\qtypes{\ensuremath{\mathbf Q}}
\newcommand\types{\ensuremath{\mathbf T}}
\newcommand\basis{\ensuremath{\mathcal B}}
\newcommand\values{\ensuremath{\mathcal V}}
\newcommand\s[1]{\ensuremath{\mathsf{#1}}}
\newcommand\head{\text{\sl head}}
\newcommand\tail{\text{\sl tail}}

\newcommand\red[2][1]{\overset{\scriptscriptstyle\smash{#2}\vphantom{x}}{\lra[#1]}\ }
\newcommand\z[1][A]{\vec 0_{S(#1)}}
\newcommand\nullvec[1]{\z[#1]}
\newcommand\may[1][\alpha]{[{#1}.]}
\newcommand\rbetab{(\s{\beta_b})}
\newcommand\rbetan{(\s{\beta_n})}
\newcommand\riftrue{(\s{if_{1}})}
\newcommand\riffalse{(\s{if_{0}})}
\newcommand\rlinr{(\s{lin^+_r})}
\newcommand\rlinscalr{(\s{lin^\alpha_r})}
\newcommand\rlinzr{(\s{lin^0_r})}
\newcommand\rlinl{(\s{lin^+_l})}
\newcommand\rlinscall{(\s{lin^\alpha_l})}
\newcommand\rlinzl{(\s{lin^0_l})}
\newcommand\rneut{(\s{neutral})}
\newcommand\runit{(\s{unit})}
\newcommand\rzeros{(\s{zero_\alpha})}
\newcommand\rzero{(\s{zero})}
\newcommand\rzeroS{(\s{zero_S})}
\newcommand\rprod{(\s{prod})}
\newcommand\rdists{(\s{\alpha dist})}
\newcommand\rdistcasum{(\s{dist^+_\Uparrow})}
\newcommand\rdistcascal{(\s{dist^\alpha_\Uparrow})}
\newcommand\rcaneutl{(\s{neut^\Uparrow_\ell})}
\newcommand\rcaneutr{(\s{neut^\Uparrow_r})}
\newcommand\rfact{(\s{fact})}
\newcommand\rfacto{(\s{fact^1})}
\newcommand\rfactt{(\s{fact^2})}
\newcommand\rproj{(\s{proj})}
\newcommand\rhead{(\s{head})}
\newcommand\rtail{(\s{tail})}
\newcommand\rdistzr{(\s{dist^0_r})}
\newcommand\rdistzl{(\s{dist^0_l})}
\newcommand\rdistscalr{(\s{dist^\alpha_r})}
\newcommand\rdistscall{(\s{dist^\alpha_l})}
\newcommand\rdistsumr{(\s{dist^+_r})}
\newcommand\rdistsuml{(\s{dist^+_l})}
\newcommand\rcomm{(\s{comm})}
\newcommand\rassoc{(\s{assoc})}
\newcommand\tax{\textsl{Ax}}

\newcommand\tif{\textsl{If}}
\newcommand\gen{\mathcal S} % S for Span
\newcommand\den[1]{\llbracket #1 \rrbracket}
\newcommand\vect[2]{\left(\begin{smallmatrix} #1 \\ #2\end{smallmatrix}\right)}

\newcommand{\app}[2]{#1 #2}
\newcommand{\reducesto}{\rightarrow}
\newcommand{\proj}[2]{\pi_{#1} #2}
\newcommand{\vrbl}[2]{#1^{#2}}
\newcommand{\abstr}[2]{\lambda#1\ldotp#2}

\begin{document}

\begin{frontmatter}

  \title{Two linearities for quantum computing\\ in the lambda
    calculus\tnoteref{previous,projects}} \tnotetext[previous]{This paper
    extends a paper presented at TPNC'17 by the first two authors
    \cite{DiazcaroDowekTPNC17} and the third author \emph{Licenciatura}'s thesis
    \cite{RinaldiUNR18}.} \tnotetext[projects]{Partially supported by the ECOS
    Project A17C03 QuCa, PICT-PRH 2015-1208, and the Laboratoire International
    Associ\'e SINFIN.}

  \author[CONICET,UNQ]{Alejandro D\'iaz-Caro}
  \ead{adiazcaro@icc.fcen.uba.ar}

  \author[ENS]{Gilles Dowek} \ead{gilles.dowek@ens-paris-saclay.fr}

  \author[UNR]{Juan Pablo Rinaldi} \ead{juampi.rinaldi@gmail.com}

  \address[CONICET]{CONICET-UBA, ICC,
    Pabell\'on 1, Ciudad Universitaria, Buenos Aires, Argentina.}
  \address[UNQ]{Universidad Nacional de Quilmes, R.~S\'aenz Pe\~na
    352, Bernal, BA, Argentina}
  \address[ENS]{Inria, LSV, ENS Paris-Saclay, 61, avenue du Pr\'esident Wilson,
    Cachan, France}
  \address[UNR]{Universidad Nacional de Rosario, Pellegrini 250, Rosario, SF, Argentina}

\begin{abstract}
  We propose a way to unify two approaches of non-cloning in quantum
  lambda-calculi: logical and algebraic linearities. The first approach is to forbid duplicating variables, while
  the second is to consider all lambda-terms as algebraic-linear functions. We
  illustrate this idea by defining a quantum extension of first-order
  simply-typed lambda-calculus, where the type is linear on superposition, while
  allows cloning base vectors. In addition, we provide an interpretation of the
  calculus where superposed types are interpreted as vector spaces and
  non-superposed types as their basis.
\end{abstract}

\begin{keyword}
  quantum computing, lambda-calculus, algebraic linearity, linear logic,
  measurement
\end{keyword}

\end{frontmatter}
\section{Introduction}

Extending $\lambda$-calculus into a programming language for quantum computing
requires to add, besides the usual abstraction and application symbols, a sum
and a product to build linear combinations of terms, and a tensor-product like
symbol to include datatypes composed of several qubits. Yet, mixing all the
constructs in a naive way leads to a too powerful calculus where non linear,
that is non physical functions, can be defined. For instance, in $\lambda$-calculus, applying the term $\lambda x.(x \otimes x)$, that
expresses a non-linear function for some convenient definition of $\otimes$, to
a term $u$ yields the term $(\lambda x.(x \otimes x))u$, that reduces to $u
\otimes u$. But ``cloning'' this vector $u$ is forbidden in quantum computing.
Various quantum $\lambda$-calculi address this problem in different ways.

One way is to forbid the construction of the term $\lambda x.(x \otimes x)$
using a typing system inspired from linear logic
\cite{GirardTCS87,AbramskyTCS93}, leading to logic-linear calculi
\cite{AltenkirchGrattageLICS05,SelingerValironSTQC09,GreenLeFanulumsdaineRossSelingerValironPLDI13,PaganiSelingerValironPOPL14,ZorziMSCS16}.
Another is to define the operational semantics in such a way that every $\lambda$-term represents a linear function. The term
$\lambda x.(x \otimes x)$, for instance, expresses the linear function that maps
$\ket{0}$ to $\ket{0} \otimes \ket{0}$ and $\ket{1}$ to $\ket{1} \otimes
\ket{1}$\footnote{Where $\ket x$ is the Dirac notation for vectors, with $\ket
  0=\bigl( \begin{smallmatrix}1\\0\end{smallmatrix}\bigr)\in\mathbb C^2$ and
  $\ket 1=\bigl( \begin{smallmatrix}0\\1\end{smallmatrix}\bigr)\in\mathbb C^2$,
  so $\{\ket 0,\ket 1\}$ is an orthonormal basis of $\mathbb C^2$, called here
  the ``computational basis''.}. This leads to restrict beta-reduction to the
case where $u$ is a base vector (in the computational basis) and to add the
linearity rule $f\pair uv \longrightarrow\pair{fu}{fv}$, leading to
algebraic-linear calculi \cite{ArrighiDowekRTA08, ArrighiDiazcaroLMCS12,
  DiazcaroPetitWoLLIC12,
  ArrighiDiazcaroValironIC17,AssafDiazcaroPerdrixTassonValironLMCS14}.

Each solution has its advantages and drawbacks. For example, let $\ite tuv$ be
the conditional statement on $\ket 0$ and $\ket 1$. Interpreting $\lambda$-terms
as algebraic-linear functions permits to reduce the term $(\lambda x.\ite x
{\ket 0} {\ket 1}) \pair{\alpha.\ket 0} {\beta.\ket 1}$ to
$\pair{\alpha.(\lambda x.\ite x {\ket 0} {\ket 1}) {\ket 0}} {\beta.(\lambda
  x~\ite x {\ket 0} {\ket 1}) {\ket 1}}$ then to $\pair {\alpha.\ket 1}
{\beta.\ket 0}$, instead of reducing it to the term $\ite {\pair{\alpha.\ket 0}
  {\beta.\ket 1}} {\ket 0} {\ket 1}$ that would be blocked. This explains that
this linearity rule, that is systematic in the algebraic-linear languages cited
above, is also present for the condition in \cite{AltenkirchGrattageLICS05} (the
so-called $\mathbf{if}^\circ$ operator).

However, interpreting all $\lambda$-terms as linear functions forbids to extend
the calculus with non-linear operators, such as measurement. For instance, the
term $(\lambda x.\pi x)\pair{\ket{0}}{\ket{1}}$, where $\pi$ represents a
measurement in the computational basis, would reduce to $\pair{(\lambda x.\pi x)
  \ket{0}}{(\lambda x.\pi x) \ket{1}}$, while it should reduce to $\ket{0}$ with
probability $\frac{1}{2}$ and to $\ket{1}$ with probability $\frac{1}{2}$.

In this paper, we propose a way to unify the two approaches, distinguishing
duplicable and non-duplicable data by their type, like in the logic-linear
calculi; and interpreting $\lambda$-terms as linear functions, like in the
algebraic-linear calculi, when they expect duplicable data. We illustrate this
idea with an example of such a calculus.

In this calculus, a qubit has type $\B$ when it is in the computational basis,
hence duplicable (a non-linear term in the sense of linear logic), and
$S(\B)$\footnote{$S$ for \emph{superposition} and also for the \emph{Span} of
  $\B$.} when it is a superposition, hence non-duplicable (a linear term in the
sense of linear logic). Hence, we can distinguish a basis term, from a term in
the span of such a basis. We could also state that the term $\ket{0}
\otimes\pair{\ket{0}}{\ket{1}}$ has type $\B\otimes S(\B)$. However, giving this
type to this term and the type $S(\B\otimes\B)$ to the term $\pair{\ket{0}
  \otimes \ket{0}}{\ket{0} \otimes \ket{1}}$ would jeopardize the subject
reduction property as, using the bilinearity of the tensor product, the former
should develop to the latter. This dilemma is not specific to quantum computing
as computing is often a non-reversible process where some information is lost.
For instance, if we express, in its type, that the term $(X - 1)(X - 2)$ is a
product of two polynomials, developing it to $X^2 - 3X + 2$ does not preserve
this type. Therefore, instead of a bilinear tensor product, we will use $n$-ary
Cartesian products, so the term $\ket 0\times\pair{\ket 0}{\ket 1}$ has type
$\B\times S(\B)$, and to move from this type to $S(\B\times\B)$ we use an
explicit cast. Notice that, if $\B$ is a set of vectors and $S(A)$ is the span
of the set $A$, then $S(\B\times\B)$ is isomorphic to $S(\B)\otimes S(\B)$.
Hence, the term $\ket{0}\times\pair{\ket{0}}{\ket{1}}$ has type $\B\times S(\B)$
and it cannot be reduced. But the term $\Uparrow
\ket{0}\times\pair{\ket{0}}{\ket{1}}$, where $\Uparrow$ is used as a mark to allow casting, has type $S(\B\times\B)$ and can be developed
to $\pair{\ket{0}\times\ket{0}}{\ket{0}\times\ket{1}}$.

This language permits expressing quantum algorithms with a very precise
information about the nature of the data processed by these algorithms.

\paragraph*{Outline of the paper}

Section~\ref{sec:QC} introduces some basic notations and concepts of quantum computing.
In Section~\ref{sec:calculus} we introduce the calculus, without product. In
Section~\ref{sec:tensor} we extend the language with a $n$-ary Cartesian product
for multiple-qubits systems. In Section~\ref{sec:SR} we state and prove the
Subject Reduction property. In Section~\ref{sec:SN} we state and prove the
Strong Normalization property. Then, in Section~\ref{sec:denSem} we provide a
straightforward interpretation of the calculus considering base types as sets of
vectors, and types $S(\cdot)$ as vector spaces. In Section~\ref{sec:examples} we
express two non-trivial example in our calculus: the Deutsch algorithm and the
Teleportation algorithm, demonstrating the expressivity of the proposed
language. Finally, in Section~\ref{sec:conclusion}, we conclude. There are also
two appendices, \ref{ap:Deutsch} and \ref{ap:telep}, which have more details of the
examples given in Section~\ref{sec:examples}.

\section{Basics notions of quantum computing}\label{sec:QC}
This section does not intend to introduce a full description of quantum
computing, the interested reader can find actual introductions to this area in
many textbooks, e.g.~\cite{NielsenChuang00,Jaeger07}. This section only intends
to introduce some basic notations and concepts.

In quantum computation, data is expressed by normalised vectors in Hilbert
spaces. For our purpose, this means that the vector spaces are defined over
complex numbers and come with a norm and a notion of orthogonality. The smallest
space usually considered is the space of {\em qubits}. This space is the
two-dimensional vector space $\mathbb{C}^2$, and it comes with a chosen
orthonormal basis denoted by $\{\ket0, \ket1\}$. A qubit (or quantum bit) is a
normalised vector $\alpha\ket0 + \beta\ket1$, where $|\alpha|^2 +|\beta|^2=1$.
To denote an unknown qubit $\psi$ it is common to write $\ket{\psi}$. A
two-qubits vector is a normalised vector in $\mathbb{C}^2\otimes\mathbb{C}^2$,
that is, a normalised vector generated by the orthonormal basis
$\{\ket{00},\ket{01},\ket{10},\ket{11}\}$, where $\ket{xy}$ stands for
$\ket{x}\otimes\ket{y}$. In the same way, a $n$-qubits vector is a normalised
vector in $(\mathbb{C}^2)^n$ (or $\mathbb{C}^N$ with $N=2^n$). Also common is
the notation $\bra{\psi}$ for the transposed, conjugate of $\ket{\psi}$, e.g.~if
$\ket{\psi}=[\alpha_1,\alpha_2,\dots,\alpha_n]^T$, then
$\bra{\psi}=[\alpha_1^*,\alpha^*_2,\dots,\alpha_n^*]$ where for any
$\alpha\in\mathbb{C}$, $\alpha^*$ denotes its conjugate.

The operators on qubits that are considered in this paper are the {\em quantum
gates}, that is, unitary operators. An {\em unitary operator} is an invertible linear
function preserving the norm and the orthogonality of vectors. The {\em adjoint}
of a given operator $U$ is denoted by $U^\dagger$, and the unitary condition
imposes that $U^\dagger U=Id$. These functions are linear, and so it is enough
to describe their action on the base vectors. Another way to describe these
functions would be by matrices, and then the adjoint is just its conjugate
transpose. A set of universal quantum gates is the set $cnot$,
$R_{\frac{\pi}{4}}$ and $had$, which can be defined as follows:
\begin{description}
\item[The $cnot$ gate.] The {\em controlled-not} is a two-qubits gate which only changes the second qubit if the first one is $\ket{1}$:
  \begin{align*}
    cnot\ket{0x} &=\ket{0x}\\
    cnot\ket{1x} &=\ket 1\otimes not\ket x
  \end{align*}
  where $not\ket{0}=\ket{1}$ and $not\ket{1}=\ket{0}$.
\item[The $R_{\frac{\pi}{4}}$ gate.] The $R_{\frac{\pi}{4}}$ gate is a
  single-qubit gate that modifies the {\em phase} of the qubit:
  \begin{align*}
    R_{\frac\pi 4}\ket 0 & = \ket 0\\
    R_{\frac\pi 4}\ket 1 & = e^{i\frac\pi 4}\ket 1
  \end{align*}
  where $\frac{\pi}{4}$ is the phase shift.
\item[The $H$ gate.] The {\em Hadamard} gate is a single-qubit gate which
  produces a 45 degree rotation of the basis:
  \begin{align*}
    H\ket 0 &=\frac 1{\sqrt 2}\ket 0+\frac 1{\sqrt 2}\ket 1\\
    H\ket 1 &=\frac 1{\sqrt 2}\ket 0-\frac 1{\sqrt 2}\ket 1
  \end{align*}
\end{description}

To make these gates act in higher-dimension qubits, they can be put together
with the bilinear symbol $\otimes$. For example, to make the Hadamard gate act
only in the first qubit of a two-qubits register, it can be taken to $H\otimes
Id$, and to apply a Hadamard gate to both qubits, just $H\otimes H$.

An important restriction, which has to be taken into account if a calculus
pretends to encode quantum computing, is the so called {\em no-cloning
  theorem}~\cite{WoottersZurekNATURE82}:

\begin{theorem}[No cloning]\label{thm:nocloningINT}
  There is no linear operator such that, given any qubit
  $\ket{\phi}\in\mathbb{C}^N$, can clone it. That is, it does not exist any
  unitary operator $U$ and fixed $\ket{\psi}\in\mathbb{C}^N$ such that
  $U\ket{\psi\phi}=\ket{\phi\phi}$.
\end{theorem}
\begin{proof}
  Assume there exists such an operator $U$, so given any $\ket{\varphi}$ and
  $\ket{\phi}$ one has $U\ket{\psi\varphi}=\ket{\varphi\varphi}$ and also
  $U\ket{\psi\phi}=\ket{\phi\phi}$. Then
  \begin{equation}\label{eq:proofNCINT}
    \braket{U\varphi\psi}{U\psi\phi}=\braket{\varphi\varphi}{\phi\phi}
  \end{equation}
  where $\bra{U\varphi\psi}$ is the conjugate transpose of $U\ket{\psi\varphi}$.
  However, notice that the left side of equation \eqref{eq:proofNCINT} can be
  rewritten as
  \[
    \bra{\varphi\psi}{U}^\dagger
    U\ket{\psi\phi}=\braket{\varphi\psi}{\psi\phi}=\braket{\varphi}{\phi}
  \]
  While the right side of equation \eqref{eq:proofNCINT} can be rewritten as
  \[
    \braket{\varphi}{\phi}\braket{\varphi}{\phi}={\braket{\varphi}{\phi}}^2
  \]

  So $\braket{\varphi}{\phi}={\braket{\varphi}{\phi}}^2$, which implies either
  $\braket{\varphi}{\phi}=0$ or $\braket{\varphi}{\phi}=1$, none of which can be
  assumed in the general case, since $\ket{\varphi}$ and $\ket{\phi}$ were
  picked as random qubits. \qed
\end{proof}

The implication of this theorem in the design choices of a calculus is that it
must be forbidden to allow functions duplicating arbitrary arguments. However,
notice that this does not forbid cloning some specific qubit states. Indeed, for
example the qubits $\ket{0}$ and $\ket{1}$ can be cloned without much effort by
using the $cnot$ gate: $cnot\ket{00}=\ket{00}$ and $cnot\ket{10}=\ket{11}$. In
this sense, the imposed restriction is not a resources-aware restriction {\em
  \`a la} linear logic~\cite{GirardTCS87}. It is a restriction that forbids us
from creating a `universal cloning machine', but still allows us to clone any
given known term.

Another operation considered on qubits is the measurement. A projector is an
operator of the form $\ket\phi\bra\phi$. For example, in the canonical base
$\{\ket 0,\ket 1\}$ of $\mathbb C^2$, $P_0=\ket 0\bra 0$ is a projector and
$P_1=\ket 1\bra 1$ is another projector, with respect to such a base. Indeed,
\begin{align*}
  P_0(\alpha\ket 0+\beta\ket 1)
  &=\alpha P_0\ket 0+\beta P_0\ket 1\\
  &=\alpha\ket 0\braket 00 +\beta\ket 0\braket 01\\
  &=\alpha\ket 0
  \\
  P_1(\alpha\ket 0+\beta\ket 1)
  &=\alpha P_1\ket 0+\beta P_1\ket 1\\
  &=\alpha\ket 1\braket 10 +\beta\ket 1\braket 11\\
  &=\beta\ket 1
\end{align*}
With these projectors we can define the measurement operators $M_0$ and $M_1$ as
\[
  M_i\ket\psi = \frac{P_i\ket\psi}{\sqrt{\bra\psi P_i\ket\psi}}
\]
For example,
\begin{align*}
  M_0(\alpha\ket 0+\beta\ket 1)
  &=\frac{P_0(\alpha\ket 0+\beta\ket 1)}{\sqrt{(\alpha^*\bra 0+\beta^*\bra 1)P_0(\alpha\ket 0+\beta\ket 1)}}\\
  &=\frac{\alpha\ket 0}{\sqrt{(\alpha^*\bra 0+\beta^*\bra 1)(\alpha\ket 0)}}\\
  &=\frac{\alpha\ket 0}{\sqrt{|\alpha|^2\braket 00+\beta^*\alpha\braket 10}}\\
  &=\frac{\alpha\ket 0}{\sqrt{|\alpha|^2}}
    =\frac{\alpha}{|\alpha|}\ket 0~\equiv\!{\text{\footnotemark}}~\ket 0
\end{align*}
\footnotetext{The scalar $\frac\alpha{|\alpha|}$ is known as a {\em phase} and
  can be ignored, so only $\ket 0$ remains.}

The quantum measurement is defined in terms of sets of measurements operators.
For example, in the canonical base $\{\ket 0,\ket 1\}$, the set $\{M_0,M_1\}$ is
a quantum measurement. When it acts on a qubit $\ket\phi$, it will apply the
operator $M_i$, with probability $\bra\psi P_i\ket\psi$.
\section{No cloning, superpositions and measurement}\label{sec:calculus}
\label{sec:fstGram}\label{sec:fstTS}\label{sec:fstTRS}

The grammar of types is defined in Table~\ref{tab:typesNoX} and the grammar of
terms in Table~\ref{tab:termsNoX}, where $\alpha\in\mathbb C$.

\begin{table}
  \centering
  \begin{align*}
    \gB &:= \B~|~S(\gB) & \textrm{Qubit types ($\qtypes$)}\\
    A & := \gB~|~\gB\Rightarrow A~|~S(A) & \textrm{Types (\types)}
  \end{align*}
  \caption{First grammar of types, without product.}
  \label{tab:typesNoX}
\end{table}

\begin{table}
  \centering
  \begin{align*}
    b  & := x~|~\lambda x^{\gB}.t~|~\ket 0~|~\ket 1 & \textrm{Base terms (\basis)}\\
    v  & := b~|~\pair vv~|~\z~|~\alpha.v & \textrm{Values (\values)}\\
    t  & := v~|~tt~|~\pair tt~|~\pi t~|~\ite{}tt~|~\alpha.t & \textrm{Terms ($\Lambda$)}
  \end{align*}
  \caption{First grammar of terms, without product.}
  \label{tab:termsNoX}
\end{table}

Terms are variables, abstractions, applications, two constants for base qubits
($\ket 0$ and $\ket 1$), linear combinations of terms (built with addition and
product by a scalar, addition being commutative and associative), a family of
constants for the null vectors, one for each type of the form $S(A)$, ($\z$),
and an if-then-else construction ($\ite{}tt$) deciding on base vectors. We use
the notation $\ite trs$ for $(\ite{}rs)t$, so making $\ite{}rs$ a function, which applied to a term $t$ produces ``if $t$ then $r$ else $s$'', when $t$ is a basis term. We also include a symbol $\pi$ for
measurement in the computational basis.

The grammar is split into base terms (non-superposed values), general values,
and general terms. Types are also split into qubit types and general types.

The set of free variables of a term $t$ is defined as usual in the
$\lambda$-calculus and denoted by $FV(t)$. We use $\may t$ as a notation to
refer indistinctly to $\alpha.t$ and to $t$. We use $-t$ as a shorthand notation
for $-1.t$, and $(t-r)$ as a shorthand notation for $\pair t{(-r)}$. The term
$(t-t)$ will have type $S(A)$, and reduce to $\z$, which is not a base term.

An important property of this calculus is that types $S(\cdot)$ are linear
types. Indeed, those correspond to superpositions, and so no duplication is
allowed on them. Instead, at this tensor-free stage, a type without an
$S(\cdot)$ on head position is a non-linear type, such as $\B$, which correspond
to base terms, i.e.~terms that can be cloned. A non-linear function is allowed
to be applied to a linear argument, for example, $\lambda x^{\B}.(fxx)$ can be
applied to $\pair{\frac 1{\sqrt 2}.\ket 0}{\frac 1{\sqrt 2}.\ket 1}$, however,
it distributes in the following way: $(\lambda x^{\B}.(f xx))~\pair{\frac
  1{\sqrt 2}.\ket 0}{\frac 1{\sqrt 2}.\ket 1} \longrightarrow \pair{\frac
  1{\sqrt 2}.(\lambda x^{\B}.(f xx))\ket 0}{\frac 1{\sqrt 2}.(\lambda x^{\B}.(f
  xx))\ket 1} \longrightarrow \pair{\frac 1{\sqrt 2}.(f\ket 0\ket 0)}{\frac
  1{\sqrt 2}.(f\ket 1\ket 1)} $.

Hence, the beta reduction occurs only when the type of the argument is the same
as the type expected by the abstraction. Thus, the rewrite system depends on
types. For this reason, we describe first the type system, and only then the
rewrite system.

A type $A$ will be interpreted as a set of vectors and $S(A)$ as the vector
space generated by the span of such a set (cf.~Section~\ref{sec:denSem}). Hence,
we naturally have $A\subseteq S(A)$ and $S(S(A))=S(A)$. Therefore, we also
define a subtyping relation on types (cf.~Table~\ref{tab:SubtypingNoX}). The
type system is given in Table~\ref{tab:TS}, where contexts $\Gamma$ and $\Delta$
have a disjoint support.

\begin{table}
  \centering
  \[
    \infer{A \preceq A}{}
    \qquad
    \infer{A \preceq C}{A \preceq B & B \preceq C}
  \]
  \[
    \infer{A\preceq S(A)}{} \qquad \infer{S(S(A))\preceq S(A)}{} \qquad
    \infer{\gB\Rightarrow A\preceq\gB\Rightarrow B}{A\preceq B} \qquad
    \infer{S(A)\preceq S(B)}{A\preceq B}
  \]
  \caption{First subtyping relation, without product.}
  \label{tab:SubtypingNoX}
\end{table}

\begin{table}
  \[\def\arraystretch{1.5}
    \begin{array}{c}
      \infer[\tax]
      {x:\gB\vdash x:\gB}
      {}
      \quad
      \infer[\tax_{\vec 0}]
      {\vdash \z:S(A)}
      {}
      \quad
      \infer[\tax_{\ket 0}]
      {\vdash\ket 0:\B}
      {}
      \quad
      \infer[\tax_{\ket 1}]
      {\vdash\ket 1:\B}
      {}
      \\
      \infer[S_I^\alpha]
      {\Gamma\vdash \alpha.t:S(A)}
      {\Gamma\vdash t:A}
      \qquad
      \infer[S_I^+]
      {\Gamma,\Delta\vdash\pair tu:S(A)}
      {
      \Gamma\vdash t:A
      &
        \Delta\vdash u:A
        }
        \qquad
        \infer[S_E]
        {\Gamma\vdash\pi t:\B}
        {\Gamma\vdash t:S(\B)}
      \\
      \infer[\preceq]
      {\Gamma\vdash t:B}
      {
      \Gamma\vdash t:A
      \
      {\scriptstyle (A\preceq B)}
      }
      \qquad
      \infer[\tif]
      {\Gamma\vdash\ite{}tr:\B\Rightarrow A}{\Gamma\vdash t:A & \Gamma\vdash r:A}
                                                                \qquad
                                                                \infer[\Rightarrow_I]
                                                                {\Gamma\vdash\lambda x^{\gB}.t:\gB\Rightarrow A}
                                                                {\Gamma,x:\gB\vdash t:A}
      \\
      \infer[\Rightarrow_E]
      {\Gamma,\Delta\vdash tu:A}
      {
      \Gamma\vdash t:\gB\Rightarrow A
      &
        \Delta\vdash u:\gB
        }
        \qquad
        \infer[\Rightarrow_{ES}]
        {\Gamma,\Delta\vdash tu:S(A)}
        {
        \Gamma\vdash t:S(\gB\Rightarrow A)
                                                              &
                                                                \Delta\vdash u:S(\gB)
                                                                }
      \\
      \ \qquad
      \infer[W]
      {\Gamma,x:\B\vdash t:A}
      {\Gamma\vdash t:A}
      \qquad
      \infer[C]
      {\Gamma,x:\B\vdash (x/y)t:A}
      {\Gamma,x:\B,y:\B\vdash t:A}
      \qquad
    \end{array}
  \]
  \caption{First type system, without product.}
  \label{tab:TS}
\end{table}

Remarks: Rule $\tax$ allows typing variables only with qubit types. Hence, the
system is first-order and only qubits can be passed as arguments (more when the
rewrite system is presented). Rule $\tax_{\vec 0}$ types the null vector as a
non-base term, because the null vector cannot belong to the base of any vector
space. Rules $\tax_{\ket 0}$ and $\tax_{\ket 1}$ type the base qubits with the
base type $\B$.

Thanks to rule $\preceq$ the term $\ket 0$ has type $\B$ and also the more
general type $S(\B)$. Note that $\npair{\pair{\ket 0}{\ket 0}}{\ket 0}$ has type
$S(\B)$ and reduces to $\ket 0$ which has the same type $S(\B)$. Reducing this
term to $\ket 0$ of type $\B$ would not preserve its type. Moreover, this type
would contain information impossible to compute, because the value $\ket 0$ is
not the result of a measurement, but of an interference.

Rule $S_I^\alpha$ states that a term multiplied by a scalar is not a base term.
Even if the scalar is just a phase, we must type the term with an $S(\cdot)$
type, because our measurement operator will remove any scalars, so having the
scalar means that it has not been measured yet. Rule $S_I^+$ is the analog for
sums to the previous rule. Rule $S_E$ is the elimination of the superposition,
which is achieved by measuring (using the $\pi$ operator).

Notice that $\ite{}tr$ is typed as a non-linear function by rule $\tif$, and so,
the if-then-else linearly distributes over superpositions, e.g.
\begin{align*}
  \ite{(\alpha.\ket 0+\beta.\ket 1)}tr &= (\ite{}tr)(\alpha.\ket 0+\beta.\ket 1)\\
                                       &\longrightarrow^* \alpha.(\ite{}tr)\ket 0+\beta.(\ite{}tr)\ket 1\\
                                       &=\alpha.\ite{\ket 0}tr+\beta.\ite{\ket 1}tr\\
                                       &\longrightarrow^*\alpha.r+\beta.t
\end{align*}

Rule $\Rightarrow_{ES}$ is the elimination for superpositions, corresponding to
the linear distribution. Notice that the type of the argument is a superposition
of the argument expected by the abstraction ($S(\gB)$ vs.~$\gB$). Also, the
abstraction is allowed to be a superposition. If, for example, we want to apply
the sum of functions $\pair fg$ to the base argument $\ket 0$, we would obtain
the superposition $\pair{f\ket 0}{g\ket 0}$. The typing is as follows:

\[
  \infer[\Rightarrow_{ES}] {\vdash\pair fg\ket 0:S(A)} { \infer[S_I^+]
    {\vdash\pair fg:S(\B\Rightarrow A)} {\vdash f:\B\Rightarrow A & \vdash
      g:\B\Rightarrow A} & \infer[\preceq] {\vdash\ket 0:S(\B)}
    {\infer[\tax_{\ket 0}] {\vdash\ket 0:\B}{}} }
\]
which reduces to
\[
  \infer[S_I^+] {\vdash\pair{f\ket 0}{g\ket 0}:S(A)} { \infer[\Rightarrow_E]
    {\vdash f\ket 0:A} { \vdash f:\B\Rightarrow A & \infer[\tax_{\ket 0}]
      {\vdash\ket 0:\B}{} } & \infer[\Rightarrow_E] {\vdash g\ket 0:A} {\vdash
      g:\B\Rightarrow A & \infer[\tax_{\ket 0}] {\vdash\ket 0:\B}{}} }
\]

Similarly, a linear function ($\vdash f:\B\Rightarrow A$) applied to a
superposition $\pair{\ket 0}{\ket 1}$ reduces to a superposition $\pair{f\ket
  0}{f\ket 1}$. The typing is as follows:
\[
  \infer[\Rightarrow_{ES}] {\vdash f\pair{\ket 0}{\ket 1}:S(A)} {
    \infer[\preceq] {\vdash f:S(\B\Rightarrow A)} {\vdash f:\B\Rightarrow A} &
    \infer[S_I^+] {\vdash\pair{\ket 0}{\ket 1}:S(\B)} { \infer[\tax_{\ket
        0}]{\vdash\ket 0:\B}{} & \infer[\tax_{\ket 1}]{\vdash\ket 1:\B}{} } }
\]
which reduces to
\[
  \infer[S_I^+] {\vdash\pair{f\ket 0}{f\ket 1}:S(A)} { \infer[\Rightarrow_E]
    {\vdash f\ket 0:A} { \vdash f:\B\Rightarrow A & \infer[\tax_{\ket 0}]
      {\vdash\ket 0:\B}{} } & \infer[\Rightarrow_E] {\vdash f\ket 1:A} { \vdash
      f:\B\Rightarrow A & \infer[\tax_{\ket 1}] {\vdash\ket 1:\B}{} } }
\]

Finally, Rules $W$ and $C$ correspond to weakening and contraction on variables
with base types. The rationale is that base terms can be cloned.

The null vectors $\z$ need to be interpreted as the null vector of the vector
space $S(A)$. Therefore, since the vector space $S(S(A))$ is the same as $S(A)$,
their null vectors should coincide. Then, we define a function $\mathsf{min}(A)$
which gives us the smallest type in terms of the amounts of $S$ it includes, that generates the vector space, so the null vector can be
taken from such a space.
\begin{align*}
  \mathsf{min}(\B) &= \B\\
  \mathsf{min}(\Psi\Rightarrow A) &=\Psi\Rightarrow\mathsf{min}(A)\\
  \mathsf{min}(S(A)) &=\mathsf{min}(A)
\end{align*}
Therefore, we will identify, through reduction, the term $\z$ with $\z[\mathsf{min}(A)]$.

The rewrite system is given in Table~\ref{tab:RSNoX}.
\begin{table}
  \vspace{-3mm}
  \[
    \begin{array}{c|r@{\ }lr}
      \multirow{2}{*}{\titulo{8.4mm}{Beta rules}}
      & \textrm{If $b$ has type $\B$ and $b\in\basis$, then } (\lambda x^\B.t)b &\lra (b/x)t & \rbetab\\
      & \textrm{If $u$ has type $S(\gB)$, then } (\lambda x^{S(\gB)}.t)u &\lra (u/x)t & \rbetan\\\hline
      \multirow{2}{*}{\titulo{8.4mm}{If rules}}
      &\ite{\ket 1}tr &\lra t &\riftrue\\
      &\ite{\ket 0}tr &\lra r &\riffalse\\\hline
      \multirow{6}{*}{\titulo{25.2mm}{Linear distribution rules}}
      &\textrm{If $t$ has type $\B\Rightarrow A$, then } t\pair uv &\lra \pair{tu}{tv} & \rlinr\\
      &\textrm{If $t$ has type $\B\Rightarrow A$, then } t(\alpha.u) &\lra\alpha.tu & \rlinscalr\\
      &\textrm{If $t$ has type $\B\Rightarrow A$, then } t\z[\B] &\lra\z[\mathsf{min}(A)] &\rlinzr\\
      &\pair tuv &\lra\pair{tv}{uv} & \rlinl\\
      &(\alpha.t)u &\lra\alpha.tu &\rlinscall\\
      &\z[\B\Rightarrow A]t &\lra\z[\mathsf{min}(A)] &\rlinzl\\\hline
      \multirow{10}{*}{\titulo{37.8mm}{Vector space axioms rules}}
      &\pair\z t &\lra t &\rneut\\
      &1.t &\lra t &\runit\\
      &\textrm{If $t$ has type $A$, then }0.t &\lra\z[\mathsf{min}(A)] &\rzeros\\
      &\alpha.\z &\lra\z[\mathsf{min}(A)] &\rzero\\
      &\alpha.(\beta.t) &\lra (\alpha\beta).t &\rprod\\
      &\alpha.\pair tu &\lra\pair{\alpha.t}{\alpha.u} &\rdists\\
      &\pair{\alpha.t}{\beta.t} &\lra(\alpha+\beta).t &\rfact\\
      &\pair{\alpha.t}t &\lra (\alpha+1).t &\rfacto\\
      &\pair tt &\lra 2.t &\rfactt\\
      &\textrm{If $A\neq\mathsf{min}(A)$, then }\z &\lra\z[\mathsf{min}(A)] & \rzeroS \\\hline
      \multirow{2}{*}{\titulo{8.4mm}{=}}
      &\pair tr &=_{AC}\pair rt & \rcomm\\
      &\pair{\pair tr}s &=_{AC} \pair t{\pair rs} & \rassoc\\\hline
      \multirow{2}{*}{\titulo{15mm}{Projection rule}}
      &\pi(\sum\limits_{i=1}^n
        \may[\alpha_i]b_i)&\lrap
                            b_k & \rproj\\
      &\multicolumn{3}{c}{\parbox{11cm}{where
        $p=\frac{|\alpha_k|^2}{\sum_{i=1}^n|\alpha_i|^2}$;
        $\forall i, b_i=\ket 0$ or $b_i=\ket 1$;
        $\sum_{i=1}^n\alpha_i.b_i$ is a normal term (hence $1\leq n\leq 2$);
        and if an $\alpha_k$ is absent, $|\alpha_k|^2=1$, and $1\leq k\leq n$.
        }}\\\hline
      \multirow{2}{*}{\titulo{12.4mm}{Contextual rules}}
      &\multicolumn{3}{c}{
        \infer{tv\lrap uv}{t\lrap u}
        \qquad
        \infer{(\lambda x^\B.v)t\lrap(\lambda x^\B.v)u}{t\lrap u}
        \qquad
        \infer{\ite trs\lrap\ite urs}{t\lrap u}
        }\\
      &\multicolumn{3}{c}{
        \infer{\pair tv\lrap\pair uv}{t\lrap u}
        \qquad
        \infer{\alpha.t\lrap\alpha.u}{t\lrap u}
        \qquad
        \infer{\pi t\lrap\pi u}{t\lrap u}
        }
    \end{array}
  \]
  \centering
  All the terms are considered to be closed (i.e. reduction is weak).
  \caption{First rewrite system, without product.}
  \label{tab:RSNoX}
\end{table}
The relation $\lrap$ is a probabilistic relation where $p$ is the probability of
occurrence. Every rewrite rule has a probability $1$ of occurrence, except for
the projection rule $\rproj$. The rewrite system depends on the typing, in
particular an abstraction can either expect a base term as argument (that is, a
non-linear term) or a superposition, which has to be treated linearly. However,
an abstraction expecting a non-linear argument can be given a superposition
(which is linear), and it is typable, only that the reduction distributes before
beta-reduction.

There are two beta rules. Rule $\rbetab$ acts only when the argument is a base
term, and the type expected by the abstraction is a base type. Hence, rule
$\rbetab$ is ``call-by-base'' (base terms coincides with values of
$\lambda$-calculus, while values on this calculus also includes superpositions
of base terms and the null vector). Instead, $\rbetan$ is the usual call-by-name
beta rule. They are distinguished by the type of the argument. Rule \rbetab\
acts on non-linear functions while \rbetan\ is for linear functions. The test on
the type of the argument is due to the type system that allows an argument with
a type not matching with the type expected by the abstraction (in such a case,
one of the linear distribution rules applies).

Since there are two beta reductions, the contextual rule admitting reducing the
argument on an application is valid only when the abstraction expects an
argument of type $\B$. If the argument is typed with a base type, then it
reduces to a term that can be cloned, and we must reduce it first to ensure that
we are cloning a term that can be cloned indeed. For example, a measure over a
superposition (e.g. $\pi\pair{\ket 0}{\ket 1})$ has a base type $\B$, but it
cannot be cloned until it is reduced. Indeed, $(\lambda
x^{\B}.(fxx))(\pi\pair{\ket 0}{\ket 1})$ can reduce either to $f\ket 0\ket 0$ or
$f\ket 1\ket 1$, but never to $f\ket 0\ket 1$ or $f\ket 1\ket 0$, which would be
possible only if the measure happens after the cloning machine. A more physical
way to state it is that cloning after measurement is not a problem, since we
already know the state to be cloned: It would be enough to prepare a second
system in the same state.

The group If-then-else contains the tests over the base qubits $\ket 0$ and
$\ket 1$.

The first three of the linear distribution rules (those marked with subindex
$r$), are the rules that are used when a non-linear abstraction is applied to a
linear argument (that is, when an abstraction expecting a base term is given a
superposition). In these cases the beta reductions cannot be used since the side
conditions on types are not met. Hence, these distributivity rules apply
instead.

For example, let us give more details in the reduction sequence on the example
given at the beginning of this Section.
\begin{align*}
  &(\lambda x^{\B}.(fxx))\pair{\frac 1{\sqrt 2}.\ket
    0}{\frac 1{\sqrt 2}.\ket 1}\\
  &\red\rlinr
    \pair{(\lambda x^{\B}.(fxx))\frac 1{\sqrt 2}.\ket 0}{(\lambda
    x:\B~(fxx))\frac 1{\sqrt 2}.\ket 1}\\
  &\red{\rlinscalr^2}
    \pair{\frac 1{\sqrt 2}.(\lambda x^{\B}.(fxx))\ket 0}{\frac 1{\sqrt 2}.(\lambda x^{\B}.(fxx))\ket 1}\\
  &\red{\rbetab^2}
    \pair{\frac 1{\sqrt 2}.f\ket 0\ket 0}{\frac 1{\sqrt 2}.f\ket 1\ket 1}
\end{align*}

Notice that in Rule \rlinzr, the term needs to be reduced to
$\z[\mathsf{min}(A)]$. Indeed, if we just reduce $t\z[\B]$ to $\z$, there is a
problem of subject reduction: $t$ having type $\B\Rightarrow A$ do not implies
it has no other type, for example, $\B\Rightarrow B$, and so, reducing to $\z$
would break subject reduction since $\z$ does not have necesarilly type $S(B)$.
On the contrary, we can prove (cf.~Lemmas~\ref{lem:minlqA} and \ref{lem:min}) that if $t$ has types $\B\Rightarrow A$ and
$\B\Rightarrow B$, then $\mathsf{min}(A)\preceq B$, and so subject reduction is preserved.

The remaining rules in this group deal with a superposition of functions. For
example, rule \rlinl\ is the sum of functions: A superposition is a sum,
therefore, if an argument is given to a sum of functions, it needs to be given
to each function in the sum. We use a weak reduction strategy (i.e.~reduction
occurs only on closed terms), hence the argument $v$ on this rule is closed,
otherwise, it could not be typed. For example \(
x:S(\B),t:\B\Rightarrow\B,u:\B\Rightarrow\B\vdash\pair tux:S(\B) \) is
derivable, but \(
x:S(\B),t:\B\Rightarrow\B,u:\B\Rightarrow\B\vdash\pair{tx}{ux}:S(\B) \) is not.

The vector space axioms rules are the directed axioms of vector spaces
\cite{ArrighiDowekRTA08,AssafDiazcaroPerdrixTassonValironLMCS14}. The rule
$\rzeroS$ ensures that each vector space have only one null vector. The Modulo
AC rules are not proper rewrite rules, but express that we consider the symbol
$+$ to be associative and commutative, and hence our rewrite system is {\em
  rewrite modulo AC}~\cite{PetersonStickelJACM81}. As a consequence, the
parenthesis are not needed and we may use the notation $\sum_{i=1}^nt_i$.
(for example, in rule \rproj).

Rule \rproj\ is the projection over weighted associative pairs, that
is, the projection over a generalization of multisets where the multiplicities
are given by complex numbers. This reduction rule is the only one with a
probability different from $1$, and it is given by the square of the modulus of
the weights\footnote{We speak about weights and not amplitudes, since the vector
  may not have norm $1$. The projection rule normalizes the vector while
  reducing.}, implementing this way the quantum measurement over the
computational basis.

Remark, to conclude, that this calculus can represent only pure states, and not
mixed states. For example, let $\s Z$ be an encoding for the quantum $Z$ gate
(cf.~Section~\ref{sec:examples}), $\ket +=\frac 1{\sqrt 2}.\pair{\ket 0}{\ket
  1}$, and $\ket -=\frac 1{\sqrt 2}.(\ket 0-\ket 1)$. The terms \( (\lambda
x:S(\B)~(\lambda y^{\B}.\ite y{(\s Z x)}x)(\pi\ket +)) \ket + \) and \( (\lambda
x:S(\B)~\pi(x)) \ket +\) may be considered equivalent if taking into account the
density matrix representation of mixed states. Indeed, the first reduces either
to $\ket +$ or $\ket -$, with probability $\frac 12$ each, while the second
reduces to $\ket 0$ or to $\ket 1$, with probability $\frac 12$ each. The sets
of pure states $\{(\frac 12,\ket +),(\frac 12,\ket -)\}$ and $\{(\frac 12,\ket
0),(\frac 12,\ket 1)\}$ have both density matrix $\frac I2$, and hence are
indistinguishable. However, once the result of the measure is known, the pure
states can be distinguished. A different approach, using density matrices, can
be seen in~\cite{DiazcaroAPLAS17}, however such a calculus has a linear type
system, and no algebraic reduction occurs.

\section{Multi-qubit systems: Tensor products}\label{sec:tensor}
\label{sec:TRS}

One postulate of quantum mechanics determines how to compose several quantum
systems. This way, the Hilbert space of a multi-qubit system is the tensor
product between single-qubit Hilbert spaces. If $\ket\psi\in\mathcal H_1$ and
$\ket\phi\in\mathcal H_2$ represent the states of two quantum systems, the state
of the full system composed by those two is $\ket\psi\otimes\ket\phi\in\mathcal
H_1\otimes\mathcal H_2$. In particular, if we chose bases $\mathcal B_1$ and
$\mathcal B_2$ of $\mathcal H_1$ and $\mathcal H_2$ respectively, we can write
$\ket\psi=\sum_{i\in\mathcal B_1}\alpha_i\ket i$ and
$\ket\phi=\sum_{j\in\mathcal B_2}\beta_j\ket j$, and so
$\ket\psi\otimes\ket\phi=\sum_{i\in\mathcal B_1}\sum_{j\in\mathcal
  B_2}\alpha_i\beta_j\ket i\otimes\ket j=\sum_{i\in\mathcal
  B_1}\sum_{j\in\mathcal B_2}\alpha_i\beta_j\ket{ij}$. The last equality can be
seen as a matter of notation, but also it is clear that $\ket i\otimes\ket
j\simeq\ket i\times\ket j$, and so $\mathcal H_1\otimes\mathcal H_2\simeq
\Span(\mathcal B_1\times\mathcal B_2)$. Therefore, since we already introduced a
symbol for the span of a type, and basis types, we only need to introduce an associative
Cartesian product to our calculus in order to recover the tensor product. For
example, the term $\ket 0\times\pair{\frac 1{\sqrt 2}\ket 0}{\frac 1{\sqrt
    2}\ket 1})$ have type $\B\times S(\B)$, while $\pair{\frac 1{\sqrt 2}\ket
  0\times\ket 0}{\frac 1{\sqrt 2}\ket 0\times\ket 1}$ have type $S(\B\times\B)$.
Therefore, the distributivity of linear combinations over tensor products is not
trivially tracked in the type system, and so an explicit cast between types is
also added: The term $\ket 0\times\pair{\frac 1{\sqrt 2}\ket 0}{\frac 1{\sqrt
    2}\ket 1}$ does not rewrite to $\pair{\frac 1{\sqrt 2}\ket 0\times\ket
  0}{\frac 1{\sqrt 2}\ket 0\times\ket 1}$, but the term $\Uparrow_\ell\ket 0\times\pair{\frac 1{\sqrt 2}\ket 0}{\frac 1{\sqrt
    2}\ket 1}$ does, where $\Uparrow_\ell$ casts the type $S(\B\times S(\B))$ into the type $S(\B\times\B)$.

The grammar of types is given in Table~\ref{tab:types}, where the Cartesian
product is added to each level. The new level ``base qubit types'' (\bqtypes) is needed
since the abstractions with variables in \bqtypes\ are the non-linear ones.
We will use the following notation: $\B^n=\B\times\dots\times\B$ ($n$-times).
Hence, $\bqtypes = \{\B^n\mid n>0\}$. Also, we may use the notation $A\times S(B^0)=A$.
\begin{table}
  \centering
  \begin{align*}
    \Bs &:= \B\mid\Bs\times\Bs & \textrm{Base qubit types (\bqtypes)}\\
    \gB &:=\Bs\mid S(\gB)\mid\gB\times\gB & \textrm{Qubit types (\qtypes)}\\
    A & := \gB\mid\gB\Rightarrow A\mid S(A)\mid A\times A & \textrm{Types (\types)}
  \end{align*}
  \caption{Grammar of types.}
  \label{tab:types}
\end{table}

The grammar of terms is given in Table~\ref{tab:terms}.
\begin{table}
  \centering
  \begin{align*}
    b  & := x\mid \lambda x^{\gB}.t\mid \ket 0\mid \ket 1\mid b\times b & \textrm{Base terms (\basis)}
    \\
    v  & := b\mid \pair vv\mid \z\mid \alpha.v\mid v\times v & \textrm{Values (\values)}
    \\
    t  & := v\mid tt\mid \pair tt\mid \pi_j t\mid \ite{}tt\mid \alpha.t & \textrm{Terms ($\Lambda$)}\\
    &\hspace{8.5mm}\mid t\times t\mid \head~t\mid \tail~t\mid \Uparrow_r t\mid\Uparrow_\ell t
  \end{align*}
  \caption{Grammar of terms.}
  \label{tab:terms}
\end{table}

Each level in the term grammar (base terms, values and general terms) is
extended with the tensor of the terms in such a level. The primitives $\head$
and $\tail$ are added to the general terms. The projector $\pi$ is generalized
to $\pi_j$, where the subindex $j$ stands for the number of qubits to be
measured, which are those in the first $j$ positions.
Notice that it is always possible to do a swap between qubits and so place the
qubits to be measured at the beginning. For instance, $\lambda x^{\B \times
  \B}.\tail~x\times\head~x$.
Finally, an explicit type cast of a term $t$ ($\Uparrow_rt$ and $\Uparrow_\ell t$) is included in the
general terms. We may use just $\Uparrow$ to refer to any of $\Uparrow_r$ or
$\Uparrow_\ell$. As the product is associative, we also may use the notation $\prod_{i=1}^n t_i$ and $\prod_{i=1}^n
A_i$ for associative Cartesian products.

The subtyping relation is also updated to include Cartesian products, and it is
given in Table~\ref{tab:Subtyping}.
\begin{table}[!h]
  \centering
  \[
    \infer{A \preceq A}{}
    \qquad
    \infer{A \preceq C}{A \preceq B & B \preceq C}
  \]
  \[
    \infer{A\preceq S(A)}{} \qquad \infer{S(S(A))\preceq S(A)}{} \qquad
    \infer{\gB\Rightarrow A\preceq\gB\Rightarrow B}{A\preceq B}
  \]
  \[
    \infer{S(A)\preceq S(B)}{A\preceq B}
    \qquad
    \infer{A\times C\preceq B\times C}{A\preceq B}
    \qquad
    \infer{C\times A\preceq C\times B}{A\preceq B}
  \]
  \caption{Subtyping relation.}
  \label{tab:Subtyping}
\end{table}

The updated type system, given in Table~\ref{tab:UTS}, includes all the typing
rules given in the previous section, plus the rules for tensor, for cast, and an
updated rule $S_E$.

\begin{table}
  \centering
  \[
    \def\arraystretch{1.5}
    \begin{array}{c}
      \infer[\tax]
      {x:\gB\vdash x:\gB}
      {}
      \quad
      \infer[\tax_{\vec 0}]
      {\vdash \z:S(A)}
      {}
      \quad
      \infer[\tax_{\ket 0}]
      {\vdash\ket 0:\B}
      {}
      \quad
      \infer[\tax_{\ket 1}]
      {\vdash\ket 1:\B}
      {}
      \\
      \infer[S_I^\alpha]
      {\Gamma\vdash \alpha.t:S(A)}
      {\Gamma\vdash t:A}
      \quad
      \infer[S_I^+]
      {\Gamma,\Delta\vdash\pair tu:S(A)}
      {
      \Gamma\vdash t:A
      &
        \Delta\vdash u:A
        }
        \quad
        \infer[S_E]{\Gamma\vdash\pi_j t:\B^j\times S(\B^{n-j})}
        {\Gamma\vdash t:S(\B^n)}
      \\
      \infer[\preceq]
      {\Gamma\vdash t:B}
      {
      \Gamma\vdash t:A
      \
      {\scriptstyle (A\preceq B)}
      }
      \quad
      \infer[\tif]
      {\Gamma\vdash\ite{}tr:\B\Rightarrow A}
      {\Gamma\vdash t:A & \Gamma\vdash r:A}
      \quad
      \infer[\Rightarrow_I]
      {\Gamma\vdash\lambda x^{\gB}.t:\gB\Rightarrow A}
      {\Gamma,x:\gB\vdash t:A}
      \\
      \infer[\Rightarrow_E]
      {\Gamma,\Delta\vdash tu:A}
      {
      \Gamma\vdash t:\gB\Rightarrow A
      &
        \Delta\vdash u:\gB
        }
        \quad
        \infer[\Rightarrow_{ES}]
        {\Gamma,\Delta\vdash tu:S(A)}
        {
        \Gamma\vdash t:S(\gB\Rightarrow A)
      &
        \Delta\vdash u:S(\gB)
        }
      \\
      \ \quad
      \infer[W]
      {\Gamma,x:\B^n\vdash t:A}
      {\Gamma\vdash t:A}
      \quad
      \infer[C]
      {\Gamma,x:\B^n\vdash (x/y)t:A}
      {\Gamma,x:\B^n,y:\B^n\vdash t:A}
      \quad\
      \\
      \ \quad
      \infer[\times_I]
      {\Gamma,\Delta\vdash t\times r:A\times B}
      {\Gamma\vdash t:A & \Delta\vdash r:B}
         \\~\quad
        \infer[\times_{Er}\ {\scriptstyle (n>1)}]
        {\Gamma\vdash \head~t:\B}
        {\Gamma\vdash t:\B^n}
        \qquad
        \infer[\times_{El}\ {\scriptstyle (n>1)}]
        {\Gamma\vdash \tail~t:\B^{n-1}}
        {\Gamma\vdash t:\B^n}
        \quad\
      \\
      \ \quad
      \infer[\Uparrow_r]
      {\Gamma\vdash \Uparrow_r t:S(A\times B)}
      {\Gamma\vdash t:S(S(A)\times B)}
      \quad
      \infer[\Uparrow_\ell]
      {\Gamma\vdash \Uparrow_\ell t:S(A\times B)}
      {\Gamma\vdash t:S(A\times S(B))}
      \quad\
    \end{array}
  \]
  \caption{Type system.}
  \label{tab:UTS}
\end{table}

Rules $\tax$, $\tax_{\vec 0}$, $\tax_{\ket 0}$, $\tax_{\ket 1}$, $\preceq$,
$S_I^\alpha$, $S_I^+$, $\tif$, $\Rightarrow_I$, $\Rightarrow_E$ and
$\Rightarrow_{ES}$ remain unchanged. Rule $S_E$ types the generalized
projection: we force the term to be measured to be typed with a type of the form
$S(\B^n)$, and then, after measuring the first $j$ qubits, the new type becomes
$\B^j\times S(\B^{n-j})$, that is, we remove the superposition mark $S(\cdot)$
from the first $j$ types in the tensor product. Rules $W$ and $C$ are updated
only to act on types $\B^n$ instead of just $\B$.

Rules $\times_I$, $\times_{E_r}$ and $\times_{E_l}$ are the standard
introduction and eliminations for lists, however, the elimination is only
allowed on terms with type $\B^n$ (basis qubits). Rules $\Uparrow_r$ and
$\Uparrow_\ell$ type the castings. We only need to allow to cast a superposed type
into a superposed tensor product, thanks to the subtyping relation. Indeed, for
example, to cast $t\times\pair rs$ from type $A\times S(B)$ to type $S(A\times
B)$, we can use the subtyping first to assign the type $S(A\times S(B))$ to
$t\times\pair rs$.

To update the rewrite system, we need to update the function $\mathsf{min}$ to
include products, as follows.
\begin{align*}
  \mathsf{min}(\B) &= \B\\
  \mathsf{min}(\Psi\Rightarrow A) &=\Psi\Rightarrow\mathsf{min}(A)\\
  \mathsf{min}(A\times B) &=\mathsf{min}(A)\times\mathsf{min}(B)\\
  \mathsf{min}(S(A)) &=\mathsf{min}(A)
\end{align*}

The updated rewrite system is given in Table~\ref{tab:URS}. It includes all the
rules from Table~\ref{tab:RSNoX} plus the rules for lists: $\rhead$ and $\rtail$
and the typing casts rules, which normalize superpositions to sums of base
terms, while update the types.

\begin{table}\centering
  \scalebox{.94}{
    \hspace{-3mm}\parbox{\textwidth}{\vspace{-4mm}
      \[
        \begin{array}{c|r@{\ }lr}
          \multirow{2}{*}{\titulo{8.4mm}{Beta}}
          & \textrm{If $b$ has type $\B^n$ and $b\in\basis$, } (\lambda x^{\B^n}.t)b &\lra (b/x)t & \rbetab\\
          & \textrm{If $u$ has type $S(\gB)$, } (\lambda x^{S(\gB)}.t)u &\lra (u/x)t & \rbetan\\\hline
          \multirow{2}{*}{\titulo{8.4mm}{If}}
          &\ite{\ket 1}tr &\lra t &\riftrue\\
          &\ite{\ket 0}tr &\lra r &\riffalse\\\hline
          \multirow{6}{*}{\titulo{25.2mm}{Linear distribution}}
          &\textrm{If $t$ has type $\B^n\Rightarrow A$, } t\pair uv &\lra \pair{tu}{tv} & \rlinr\\
          &\textrm{If $t$ has type $\B^n\Rightarrow A$, } t(\alpha.u) &\lra\alpha.tu & \rlinscalr\\
          &\textrm{If $t$ has type $\B^n\Rightarrow A$, } t\z[\B^n] &\lra\z[\mathsf{min}(A)] &\rlinzr\\
          &\pair tuv &\lra\pair{tv}{uv} & \rlinl\\
          &(\alpha.t)u &\lra\alpha.tu &\rlinscall\\
          &\z[\B^n\Rightarrow A]t &\lra\z[\mathsf{min}(A)] &\rlinzl\\\hline
          \multirow{10}{*}{\titulo{37.8mm}{Vector space axioms}}
          &\pair\z t &\lra t &\rneut\\
          &1.t &\lra t &\runit\\
          &\textrm{If $t$ has type $A$, }0.t &\lra\z[\mathsf{min}(A)] &\rzeros\\
          &\alpha.\z &\lra\z[\mathsf{min}(A)] &\rzero\\
          &\alpha.(\beta.t) &\lra (\alpha\beta).t &\rprod\\
          &\alpha.\pair tu &\lra\pair{\alpha.t}{\alpha.u} &\rdists\\
          &\pair{\alpha.t}{\beta.t} &\lra(\alpha+\beta).t &\rfact\\
          &\pair{\alpha.t}t &\lra (\alpha+1).t &\rfacto\\
          &\pair tt &\lra 2.t &\rfactt\\
          &\textrm{If $A\neq\mathsf{min}(A)$, then }\z &\lra\z[\mathsf{min}(A)] & \rzeroS \\\hline
          \multirow{2}{*}{\titulo{8.4mm}{=}}
          &\pair tr &=_{AC}\pair rt & \rcomm\\
          &\pair{\pair tr}s &=_{AC} \pair t{\pair rs} & \rassoc\\\hline
          \multirow{2}{*}{\titulo{8.4mm}{Lists}}
          &\textrm{If $h\neq u\times v$ and $h\in\basis$, }\head\ h\times t &\lra h
                                                                                                  & \rhead\\
          &\textrm{If $h\neq u\times v$ and $h\in\basis$, }\tail\ h\times t
                                                                                     &\lra t & \rtail\\\hline
          \multirow{10}{*}{\titulo{52mm}{\mbox{\hspace{1.3cm}Typing casts}}}
          & \Uparrow_r \pair rs\times u &\lra\pair{\Uparrow_r r\times u}{\Uparrow_r s\times u} &\rdistsumr\\
          &\Uparrow_\ell u\times\pair rs &\lra\pair{\Uparrow_\ell u\times r}{\Uparrow_\ell u\times s} &\rdistsuml\\
          &\Uparrow_r (\alpha.r)\times u &\lra \alpha.\Uparrow_r r\times u &\rdistscalr\\
          &\Uparrow_\ell u\times(\alpha.r) &\lra \alpha.\Uparrow_r u\times r &\rdistscall\\
          &\textrm{If $u$ has type $B$, }\Uparrow_r \z\times u &\lra\z[\mathsf{min}(A\times B)] &\rdistzr\\
          &\textrm{If $u$ has type $A$, }\Uparrow_\ell u\times\z[B] &\lra\z[\mathsf{min}(A\times B)] &\rdistzl\\
          &\Uparrow\pair tu&\lra\pair{\Uparrow t}{\Uparrow u} &\rdistcasum\\
          &\Uparrow(\alpha.t)&\lra\alpha.\Uparrow t &\rdistcascal\\
          &\textrm{If $u\in\basis$, }\Uparrow_r u\times v&\lra u\times v &\rcaneutr\\
          &\textrm{If $v\in\basis$, }\Uparrow_\ell u\times v&\lra u\times v &\rcaneutl \\\hline
          \multirow{2}{*}{\titulo{12.5mm}{Projection}}
          &\multicolumn{2}{c}{
            \pi_j(\sum\limits_{i=1}^n
            \may[\alpha_i]\prod\limits_{h=1}^mb_{hi})\lrap
            (\prod\limits_{h=1}^jb_{hk})\times\sum\limits_{i\in P}\left(\frac{\alpha_i}{\sqrt{\sum\limits_{r\in P}|\alpha_r|^2}}\right)\prod\limits_{h=j+1}^mb_{hi}} & \rproj\\
          &\multicolumn{3}{c}{\parbox{11cm}{\footnotesize where
            $k\leq n$;
            $P\subseteq\mathbb N^{\leq n}$ s.t.~$\forall i\in P$, $\forall h\leq j$, $b_{hi}= b_{hk}$;
            $p=\sum\limits_{i\in P}\frac{|\alpha_i|^2}{\sum_{r=1}^n|\alpha_r|^2}$;
            $\forall i, b_i=\ket 0$ or $b_i=\ket 1$;
            $\sum_{i=1}^n\may[\alpha_i]\prod_{h=1}^mb_{hi}$ is a normal term; and
            if an $\alpha_k$ is absent, $|\alpha_k|^2=1$.
            }}\\\hline
          \multirow{2}{*}{\titulo{24mm}{Contextual rules}}
          &\multicolumn{3}{l}{\textrm{ If $t\lrap u$, then}}\\
          &\multicolumn{3}{c}{
            \begin{array}{c@{}cc}
              tv \lrap uv &
                            (\lambda x^{\B^n}.v)t\lrap(\lambda x^{\B^n}.v)u &
                                                                              \pair tv\lrap\pair uv \\
              \alpha.t\lrap\alpha.u&
                                     \pi_j t\lrap\pi_j u&
                                                          t\times v\lrap u\times v \\
              v\times t\lrap v\times u&
                                        \Uparrow_r t\lrap\Uparrow_r u&
                                                                       \Uparrow_\ell t\lrap\Uparrow_\ell u \\
              \head\ t\lrap\head\ u&
                                     \tail\ t\lrap\tail\ u&
                                                            \ite trs\lrap\ite urs
            \end{array}
          }
        \end{array}
      \]
  \centering
  All the terms are considered to be closed (i.e. reduction is weak).
    }
  }
  \caption{Rewrite system.}
  \label{tab:URS}
\end{table}

The rule $\rproj$ has been updated to account for multiple qubits systems. It
normalizes (as in norm $1$) the scalars on the obtained term. The call-by-base
beta rule \rbetab, and the contextual rule admitting reducing the argument on an
application for the call-by-base abstraction are updated to allow for
abstractions expecting arguments of type $\B^n$ instead of just $\B$ (that is, any
base qubit type).

The first six rules in the group typing casts---\rdistsumr, \rdistscalr, and
\rdistzr, and their analogous \rdistsuml, \rdistscall, and \rdistzl---deal with
the distributivity of sums, scalar product and null vector respectively. If we
ignore the type cast $\Uparrow$ on each rule, these rules are just
distributivity rules. For example, rule \rdistsumr\ acts on the term $\pair
rs\times u$, distributing the sum with respect to the tensor product, producing
$\pair{r\times u}{s\times u}$ (distribution to the right). However, the term
$\pair rs\times u$ may have type $S(A)\times B$, $S(A)\times S(B)$ or
$S(A\times B)$, while, among those, the term $\pair{r\times u}{s\times u}$
can only have type $S(A\times B)$. Hence, we cannot reduce the first term to
the second without losing subject reduction. Instead, we need to cast the term
explicitly to the valid type in order to reduce.

The next two rules, \rdistcasum\ and \rdistcascal, distribute the cast over sums
and scalars. For example $\Uparrow_r \pair{(\alpha.\ket 1)\times\ket 0} {(\beta.\ket 0)\times\ket 1}$ reduces by rule $\rdistcasum$
to $\pair { \Uparrow_r (\alpha.\ket 1)\times\ket 0 }{\Uparrow_r (\beta.\ket 0)\times\ket 1 }$, and hence,
the distributivity rule can act.
The last two rules in the group, \rcaneutr\ and \rcaneutl, remove the cast when
it is not needed anymore. For example
\begin{align*}
  \Uparrow_r (\alpha.\beta.\ket 0)\times\ket 1 & \red{\rdistscalr} \alpha.\Uparrow_r (\beta. \ket 0)\times\ket 1\\
                                               &\red{\rdistscalr} \alpha. \beta. \Uparrow_r \ket 0\times\ket 1\\
                                               &\red{\rcaneutr} \alpha.\beta.\ket 0\times\ket 1
\end{align*}

The measurement rule \rproj\ is updated to measure the first $j$ qubits. Hence,
a $n$-qubits in normal form (that is, a sum of products of qubits with or without
a scalar in front), for example, the term
\[
  2.(\ket 0\times\ket 1\times\ket 1)+ \ket 0\times\ket 1\times\ket 0 + 3.(\ket 1\times\ket 1\times\ket 1)
\]
can be measured and will produce a $n$-qubits where the
first $j$ qubits are the same and the remaining are untouched, with its scalars
changed to have norm $1$. In this 3-qubits example, measuring the first two can
produce either
\[\ket 0\times\ket 1\times\pair{\frac 2{\sqrt{5}}.\ket 1}{\frac 1{\sqrt{5}}.\ket
    0}\]
or
\[\ket 1\times\ket 1\times(1.\ket 1)\]
The
probability of producing the first is
$\tfrac{|2|^2}{(|2|^2+|1|^2+|3|^2)}+\tfrac{|1|^2}{(|2|^2+|1|^2+|3|^2)} =
\tfrac{5}{14}$ and the probability of producing the second is
$\tfrac{|3|^2}{(|2|^2+|1|^2+|3|^2)}=\tfrac{9}{14}$.

Remark, to conclude, that since the calculus presented in this paper is
call-by-base for the functions expecting a non-linear argument, it avoids a
well-known problem in others $\lambda$-calculi with a linear logic type system
including modalities. To illustrate this problem, consider the following typing
judgment:
\[
  y:S(\B) \vdash (\lambda x^{\B}.x\times x) (\pi y):S(\B)\times S(\B) \
\]
If we allow to $\beta$-reduce this term, we would obtain $(\pi
y)\times(\pi y)$ which is not typable in the context $y:S(\B)$. A standard
solution to this problem is illustrated in~\cite{Barber96}, where the terms that
can be cloned are distinguished by a mark, and used in a $\textsl{let}$
construction, while non-clonable terms are used in $\lambda$ abstractions. Since
this term will not beta reduce in our calculus, but project first, the problem
is not present neither in our case.

\section{Subject reduction}\label{sec:SR}
Thanks to the explicit casts, the resulting system has the Subject Reduction
property (Theorem~\ref{thm:SR}), that is, the typing is preserved by
weak-reduction (i.e.~reduction on closed terms). The proof of this theorem is
not trivial, specially due to the complexity of the system itself.

The two main lemmas in the proof, the generation lemma (Lemma~\ref{lem:generation}) and the substitution
lemma (Lemma~\ref{lem:substitution}), are stated below, together with a few paradigmatic cases of the proof.

We denote by $|\Gamma|$ to the multiset of types in $\Gamma$. For example,
\[
  |x:\B,y:\B,z:S(\B)|=\{\B,\B,S(\B)\}
\]

\begin{lemma}
  \label{lem:lqBn}
  If $A\preceq\B^n$, then $A=\B^n$
\end{lemma}
\begin{proof}
  By rule inspection.
\end{proof}

\begin{lemma}\label{lem:noBasisPrec}
  If $S(A)\preceq B$, then there exists $C$ such that $B=S(C)$
\end{lemma}
\begin{proof}
  Straightforward induction on the definition of $\preceq$.
\end{proof}

\begin{lemma}\label{lem:SAxBmE}
  If $S^n(A\times B)\preceq C$, then there exist $m,D,E$ such that
  $C=S^m(D\times E)$ with $A\preceq D$ and $B\preceq E$.
\end{lemma}
\begin{proof}
  By induction on the derivation of $S^n(A\times B)\preceq C$.
\end{proof}
\begin{lemma}\label{lem:minlqA}
          For any type $A$, we have $\mathsf{min}(A)\preceq A$.
        \end{lemma}
        \begin{proof}
          By cases over $A$.
        \end{proof}
        \begin{lemma}\label{lem:minS}
          If $A\preceq B$, then $\mathsf{min}(A)=\mathsf{min}(B)$.
        \end{lemma}
        \begin{proof}
          By induction on the derivation of $A\preceq B$.
        \end{proof}

        \begin{lemma}\label{lem:min}
          If $\Gamma\vdash t:A$ and $\Gamma\vdash t:B$, then $\mathsf{min}(A)=\mathsf{min}(B)$.
        \end{lemma}
        \begin{proof}
          Let $\pi$ be the derivation tree of $\Gamma\vdash t:A$ and $\pi'$ the
          derivation tree of $\Gamma\vdash t:B$, we proceed by induction on
          $|\pi|+|\pi'|$, where $|\cdot|$ is the size of the derivation tree.

          We only give three case.
          
          \begin{itemize}
          \item If $t=\pair{r}{s}$ and both derivations end with rule
            $S^+_I$, then $A=S(A')$, $B=S(B')$, 
            $\Gamma_1\vdash r:A'$,
            $\Gamma_2\vdash s:A'$,
            $\Gamma_1\vdash r:B'$,
            and $\Gamma_2\vdash s:B'$ where $\Gamma_1$ is defined only on $FV(r)$ and $\Gamma_2$
            on $FV(s)$.
            By the induction hypothesis, $\mathsf{min}(A')=\mathsf{min}(B')$,
            hence, $\mathsf{min}(A)=\mathsf{min}(B)$.
          \item If $t=rs$, $\pi$ ends with $\Rightarrow_E$ and $\pi'$ with
            $\Rightarrow_{ES}$, then 
            $\Gamma_1\vdash r:\Psi\Rightarrow A$, $\Gamma_2\vdash s:\Psi$,
            $\Gamma_1\vdash r:S(\Psi'\Rightarrow A')$, and $\Gamma_2\vdash
            s:S(\Psi')$, with $B=S(A')$.
            By the induction hypothesis,
            $\Psi\Rightarrow\mathsf{min}(A)=\Psi'\Rightarrow\mathsf{min}(A')$,
            hence $\mathsf{min}(A)=\mathsf{min}(A')=\mathsf{min}(B)$.
          \item If $\pi$ ends with rule $\preceq$, then $\Gamma\vdash t:A'$ with
            $A'\preceq A$. By the induction hypothesis, we have
            $\mathsf{min}(A')=\mathsf{min}(B)$, and by Lemma~\ref{lem:minS},
            $\mathsf{min}(A')=\mathsf{min}(A)$.
            Hence, $\mathsf{min}(A)=\mathsf{min}(B)$.
            \qed
          \end{itemize}
        \end{proof}
\begin{lemma}
  [Generation lemmas]~
  \label{lem:generation}
  \begin{itemize}
  \item If $\Gamma\vdash x:A$, then $x:\gB\in\Gamma$,
    $|\Gamma|\setminus\{\gB\}\subseteq\bqtypes$, and $\gB\preceq A$.
  \item If $\Gamma\vdash \z[B]:A$, then $S(B)\preceq A$ and
    $|\Gamma|\subseteq\bqtypes$.
  \item If $\Gamma\vdash\ket 0:A$, then $\B\preceq A$ and
    $|\Gamma|\subseteq\bqtypes$.
  \item If $\Gamma\vdash\ket 1:A$, then $\B\preceq A$ and
    $|\Gamma|\subseteq\bqtypes$.
  \item If $\Gamma\vdash \alpha.t:A$, then $\Gamma'\vdash t:B$, with
    $\Gamma'\subseteq\Gamma$, $|\Gamma\setminus\Gamma'|\subseteq\bqtypes$ and
    $S(B)\preceq A$.
  \item If $\Gamma\vdash\pair tu:A$, then $\Gamma_1\vdash t:B$ and
    $\Gamma_2\vdash u:B$, with $S(B)\preceq A$ and
    $\Gamma_1\cup\Gamma_2\subseteq\Gamma$,
    $|\Gamma\setminus(\Gamma_1\cup\Gamma_2)|\subseteq\bqtypes$.
  \item If $\Gamma\vdash\pi_j t:A$, then $\Gamma'\vdash t:S(\B^n)$, with
    $\Gamma'\subseteq\Gamma$, $|\Gamma\setminus\Gamma'|\subseteq\bqtypes$ and
    $\B^j\times S(\B^{n-j})\preceq A$.
  \item If $\Gamma\vdash\ite{}{t}{r}:A$, then $\Gamma\vdash t:B$, $\Gamma\vdash
    r:B$, with $\B\Rightarrow B\preceq A$ and $|\Gamma|\subseteq\bqtypes$.
    Moreover, the derivation trees of $\Gamma\vdash t:B$ and $\Gamma\vdash r:B$
    are strictly smaller than the derivation tree of $\Gamma\vdash\ite{}tr:A$.
  \item If $\Gamma\vdash\lambda x^{\gB}.t:A$, then $\Gamma'',x:\gB\vdash t:B$,
    with $\Gamma'\subseteq\Gamma$, $\gB\Rightarrow B\preceq A$ and
    $|\Gamma\setminus\Gamma'|\subseteq\bqtypes$. Moreover, the derivation tree of
    $\Gamma'',x:\gB\vdash t:B$ is strictly smaller than the derivation tree of
    $\Gamma\vdash\lambda x^{\gB}.t:A$.
  \item If $\Gamma\vdash tu:A$, then one of the following possibilities happens:
    \begin{itemize}
    \item $\Gamma_1\vdash t:\gB\Rightarrow B$ and $\Gamma_2\vdash u:\gB$, with
      $B\preceq A$, or
    \item $\Gamma_1\vdash t:S(\gB\Rightarrow B)$ and $\Gamma_2\vdash u:S(\gB)$,
      with $S(B)\preceq A$.
    \end{itemize}
    In both cases, $\Gamma_1\cup\Gamma_2\subseteq\Gamma$ and
    $|\Gamma|\setminus|\Gamma_1\cup\Gamma_2|\subseteq\bqtypes$.
  \item If $\Gamma\vdash t\times u:A$, then $\Gamma_1\vdash t:B$ and
    $\Gamma_2\vdash u:C$, with $\Gamma_1\cup\Gamma_2\subseteq\Gamma$,
    $|\Gamma\setminus(\Gamma_1\cup\Gamma_2)|\subseteq\bqtypes$ and $B\times
    C\preceq A$.
  \item If $\Gamma\vdash \head~t:A$, then $\Gamma'\vdash t:\B^n$, with
    $\Gamma'\subseteq\Gamma$, $|\Gamma\setminus\Gamma'|\subseteq\bqtypes$ and
    $\B\preceq A$.
  \item If $\Gamma\vdash \tail~t:A$, then $\Gamma'\vdash t:\B^n$, with
    $\Gamma'\subseteq\Gamma$, $|\Gamma\setminus\Gamma'|\subseteq\bqtypes$ and
    $\B^{n-1}\preceq A$.
  \item If $\Gamma\vdash\Uparrow_r t:A$, then $\Gamma'\vdash t:S(S(B)\times C)$, with
    $\Gamma'\subseteq\Gamma$, $|\Gamma\setminus\Gamma'|\subseteq\bqtypes$ and
    $S(B\times C)\preceq A$.
  \item If $\Gamma\vdash\Uparrow_\ell t:A$, then $\Gamma'\vdash t:S(B\times S(C))$, with
    $\Gamma'\subseteq\Gamma$, $|\Gamma\setminus\Gamma'|\subseteq\bqtypes$ and
    $S(B\times C)\preceq A$.
  \end{itemize}
\end{lemma}
\begin{proof}
  First notice that if $\Gamma\vdash t:A$ is derivable, then $\Delta\vdash t:B$
  is derivable, with $\Gamma\subseteq\Delta$ and
  $|\Delta\setminus\Gamma|\subseteq\bqtypes$ (because of rule $W$) and $A\preceq
  B$, (because of rule $\preceq$). Notice also that those are the only typing
  rules changing the sequent without changing the term on the sequent. Rules
  $\Rightarrow_E$ and $\Rightarrow_{ES}$ and are straightforward to check. All
  the other rules are syntax directed: one
  rule for each term. Therefore, the lemma is proven by a straightforward rule
  by rule analysis.

  With an analogous reasoning, the condition on the derivation trees stated in cases
  $\Gamma\vdash\ite{}tr:A$ and $\Gamma\vdash\lambda x^\Psi.t:A$ are also straightforward.
  \qed
\end{proof}

\begin{corollary}[Simplification]\label{cor:simplification}~
  \begin{enumerate}
  \item If $\vdash\pair tu:A$, then $\vdash t:A$ and $\vdash u:A$.
  \item If $\vdash\pair tu:A$, then $A=S(B)$.
  \item If $\vdash \alpha.t:A$, then $\vdash t:A$.
  \item If $\vdash\alpha.t:A$, then $\vdash\beta.t:A$.
  \item If $\vdash \alpha.t:A$, then $A=S(B)$.
  \end{enumerate}
\end{corollary}
\begin{proof}
  Cf.~\ref{ap:SR}.\qed
\end{proof}

\begin{corollary}
  \label{cor:basisTermNoSup}
  If $b\in\basis$ and $\vdash b:S(A)$, then $\vdash b:A$.
\end{corollary}
\begin{proof}
  By cases analysis on $b$. Cf.~\ref{ap:SR}. \qed
\end{proof}

\begin{lemma}
  \label{lem:linearContext}
  If $\Gamma\vdash t:A$ and $FV(t)=\emptyset$, then $|\Gamma|\subseteq\bqtypes$.
\end{lemma}
\begin{proof}
  If $FV(t)=\emptyset$ then $\vdash t:A$. If $\Gamma\neq\emptyset$, the only way
  to derive $\Gamma\vdash t:A$ is by using rule $W$ to form $\Gamma$, hence
  $|\Gamma|\subseteq\bqtypes$.
  \qed
\end{proof}

\begin{lemma}[Substitution lemma]\label{lem:substitution}
  Let $FV(u)=\emptyset$, then if $\Gamma,x:\gB\vdash t:A$, $\Delta\vdash u:\gB$,
  where if $\gB=\B^n$ then $u\in\basis$, we have $\Gamma,\Delta\vdash (u/x)t:A$.
\end{lemma}
\begin{proof}
  Notice that due to Lemma~\ref{lem:linearContext}, $|\Delta|\subseteq\bqtypes$,
  hence, it suffices to consider $\Delta=\emptyset$. By structural
  induction on $t$. Cf.~\ref{ap:SR}.
  \qed
\end{proof}

Since the strategy is weak, subject reduction is proven for closed terms.

\begin{theorem}
  [Subject reduction on closed terms]
  \label{thm:SR}
  For any closed terms $t$ and $u$ and type $A$, if $t\lra[p] u$ and $\vdash
  t:A$, then $\vdash u:A$.
\end{theorem}
\begin{proof}
  By induction on the rewrite relation. Cf.~\ref{ap:SR}.
  \qed
\end{proof}

\section{Strong normalization}\label{sec:SN}
In this section we adapt Tait's proof of strong normalization of the simply
typed lambda calculus to show the same property in our calculus.

Let $|t|$ be the size of the longest reduction sequence started in $t$ and
$\SN=\{t\mid |t|<\infty\}$. Also, let $t$ of type $A$, then $\Red t=\{r:A\mid
t\lrap r\}$. Notice that Theorem~\ref{thm:SR} proves the Subject Reduction only
for closed terms, that is why the definition of $\Red t$ requires a condition on types.

\begin{definition}
  We define the following measure $\size t$ on terms:
  \[
    \begin{array}{r@{~=~}l@{\qquad}r@{~=~}l}
      \size{x} & 0 &
      \size{tu} &  (3 \size{t} + 2)(3 \size{u} + 2) \\
      \size{\z} &  0 &
                        \size{t \times u} &  \size{t} + \size{u} \\
      \size{\ket{0}} &  0 &
                             \size{\head\ {t}} &  \size{t} + 1 \\
      \size{\ket{1}} &  0 &
                             \size{\tail\ {t}} &  \size{t} + 1 \\
      \size{\lambda x^\Psi.t} &  \size{t} &
                                             \size{\pi_{j}{t}} &  \size{t} \\
      \size{\pair tr} &  \size{t} + \size{r} + 2 &
                                                    \size{\ite{}{t}{r}} &  \size{t} + \size{r} \\
      \size{\alpha . t} &  2 \size{t} + 1 &
                                                       \size{\Uparrow t} &  \size{t}
    \end{array}
  \]
\end{definition}

\begin{lemma}\label{lem:size}
  If \( t \lra r \) by any of the rules in the groups linear distribution,
  vector space axioms or lists, or their contextual closure, then $\size r \geq \size t$. Moreover, $\size r
  = \size t$ if and only if the rule is $\rzeroS$.
\end{lemma}
\begin{proof}
  By induction and rule by rule analysis. Cf.~\ref{ap:SN}.
  \qed
\end{proof}

\begin{lemma}\label{lem:ri_in_snset_implies_sum_ri_in_snset}
  If for every \( i \in \{1, \ldots, n\} \) we have \( r_i \in \SN \),
  then \( \sum_{i = 1}^{n} \may r_i \in \SN \).
\end{lemma}
\begin{proof}
  By induction on the lexicographic order of \( (\sum_{i = 1}^{n} |r_i|,
  \size{\sum_{i = 1}^{n} \may r_i}) \). Cf.~\ref{ap:SN}. \qed
\end{proof}

\begin{lemma}\label{lem:t_implies_proj_t}
  If \( t \in \SN \), then \( \pi_{j}{t} \in \SN \).
\end{lemma}

\begin{proof}
  By induction on $|t|$. Cf.~\ref{ap:SN}. \qed
\end{proof}

From now on, $\sum_{i=1}^0 t_i=\z$ where $A$ can be determined by the context.

As usual, we associate to each type $A$ a set of strongly normalising terms
$\interp A$. However, since reduction depends on types, these sets must be
sets of typed terms, otherwise we would need to consider ill-typed reductions,
which would make the proof more complex.
\begin{definition}
  For each type $A$ we define a set of strongly normalising terms as follows:
  \begin{align*}
    \interp{\B} & = \{ t : S(\B) \mid t \in \SN \} \\
    \interp{A \times B} & = \{  t : S(S(A) \times S(B)) \mid t \in \SN \}\\
    \interp{\Psi \Rightarrow A} & = \{  t : S(\Psi \Rightarrow A) \mid \forall r \in \interp{\Psi}, t r \in \interp{A} \} \\
    \interp{S(A)} & = \{  t : S(A) \mid t \in \SN \}
  \end{align*}
\end{definition}

We define a set of neutral terms (Definition~\ref{def:neut}), in order to prove
that for every type, its interpretation have the so-called CR3 property
(Lemma~\ref{lem:cr}), that is, the closure by anti-reduction of neutral terms.
Such a property will be useful to prove the adequacy lemma
(Lemma~\ref{lem:adequacy}).
\begin{definition}\label{def:neut}
  The set of neutral terms (\( \mathcal N \)) is defined by the following
  grammar:
  \[
    n:= tt \mid\head\ t\mid\tail\ t
  \]
  where $t$ is any term produced by the grammar from Table~\ref{tab:terms}.
\end{definition}

\begin{lemma}\label{lem:cr}
  For all \( A \), the following properties hold:
  \begin{description}
    \item[(CR1)] If \( t \in \interp{A} \), then \( t \in \SN \).
    \item[(CR2)] If \( t \in \interp{A} \), then \( \Red{t} \subseteq \interp{A} \).
    \item[(CR3)] If \( t \in \mathcal{N} \), $t$ has the same type as all the
      terms in $\interp A$, and \( \Red{t} \subseteq \interp{A} \) then \( t \in \interp{A} \).
    \item[(HAB)] For all \( x^A \), \( x \in \interp{A} \).
    \item[(LIN1)] If \( t \in \interp{A} \) and \( r \in \interp{A} \), then \( t + r \in \interp{A} \).
    \item[(LIN2)] If \( t \in \interp{A} \) then \( \alpha . t \in \interp{A} \).
    \item[(NULL)] \( \nullvec{A} \in \interp{A} \)
  \end{description}
\end{lemma}
\begin{proof}
  By induction over \( A \). Cf.~\ref{ap:SN}. \qed
\end{proof}

\begin{lemma}\label{lem:a_subset_b}
  If \( A \preceq B \) then \( \interp{A} \subseteq \interp{B} \).
\end{lemma}
\begin{proof}
  By induction on the relation $\preceq$. Cf.~\ref{ap:SN}. \qed
\end{proof}

Let $\theta$ be a substitution of variables by terms. We write
$\theta\vDash\Gamma$ if for every $x:A\in\Gamma$, $\theta(x)\in\interp A$.

\begin{lemma}[Adequacy]\label{lem:adequacy}
  If \( \Gamma \vdash t:A \) and \( \theta \vDash \Gamma \) then \( \theta(t)
  \in \interp{A} \).
\end{lemma}
\begin{proof}
  By induction in the derivation of $\Gamma\vdash t:A$. Cf.~\ref{ap:SN}. \qed
\end{proof}

\begin{theorem}
  [Strong normalization]\label{thm:SN}
  If $\Gamma\vdash t:A$ then $t\in\SN$.
\end{theorem}
\begin{proof}
  By Lemma~\ref{lem:adequacy}, if $\theta\vDash\Gamma$, then
  $\theta(t)\in\interp A$. By Lemma~\ref{lem:cr} (CR1), $\interp A\subseteq\SN$.
  Finally, by Lemma~\ref{lem:cr} (HAB), $\mathsf{Id}\vDash\Gamma$, hence $t\in\SN$.
  \qed
\end{proof}

\section{Interpretation}\label{sec:denSem}

We consider vector spaces equipped with a canonical base, and subsets of such
spaces.

Let $E$ and $F$ be two vector spaces with canonical bases $B = \{\vec{b}_i~|~i
\in I\}$ and $C = \{\vec{c}_j~|~j \in J\}$. The tensor product $E \otimes F$ of
$E$ and $F$ is the vector space of canonical base $\{\vec{b}_i \times
\vec{c}_j~|~i \in I \mbox{ and } j \in J\}$, where $\vec{b}_i \times \vec{c}_j$
is the ordered pair formed with the vector $\vec{b}_i$ and the vector
$\vec{c}_j$.
The operation $\otimes$ is extended to the vectors of $E$ and $F$ by making
pairs bilinear:
$(\sum_i \alpha_i . \vec{b}_i) \otimes (\sum_j \beta_j . \vec{c}_j) = \sum_{ij}
\alpha_i . \beta_j . (\vec{b}_i \times \vec{c}_j)$.

Let $E$ and $F$ be two vector spaces equipped with bases $B$ and $C$, and $S$
and $T$ be two subsets of $E$ and $F$ respectively, we define the set $S \times
T$, subset of the vector space $E \otimes F$, as follows: \( S \times T =
\{\vec{u} \times \vec{v}~| \vec{u} \in S, \vec{v} \in T\} \).
Remark that $E \times F$ differs from $E \otimes F$. For instance, if $E$ and
$F$ are ${\mathbb C}^2$ equipped with the base $\{\vec{i}, \vec{j}\}$, then $E
\times F$ contains $\vec{i} \times \vec{i}$ and $\vec j \times \vec j$ but not
$\vec i \times \vec i + \vec j \times \vec j$, that is not a Cartesian product of
two vectors of ${\mathbb C}^2$.
Let $E$ be a vector space equipped with a base $B$, and $S$ a subset of $E$. We
write $\gen(S)$ for the vector space over $\mathbb C$ generated by the span of
$S$, that is, containing all the linear combinations of elements of $S$.
Hence, if $E$ and $F$ are two vector spaces of bases $B$ and $C$ then $E \otimes
F = \gen(B) \otimes \gen(C) = \gen( B\times C)$.

Let $S$ and $T$ be two sets. We write $S \to T$ for the vector space of formal
linear combination of functions from $S$ to $T$. The set
$S \Rightarrow T$ the set of
the functions from $S$ to $T$ is a subset---and even a basis---of this vector
space.
Note that if $S$ and $T$ are two sets , then $S \to T = \gen (S \Rightarrow T)$.

To each type we associate the subset of some vector space
\begin{align*}
  \den{\B} &= \{\vect 10, \vect 01\}, \mbox{a subset of ${\mathbb C}^2$}\\
  \den{\gB\Rightarrow A}   &= \den{\gB}\Rightarrow\den A\\
  \den{A\times B} &= \den A\times \den B\\
  \den{S(A)} &= \gen\den A
\end{align*}
Remark that
\(\den{S(\B\times\B)}=\gen(\den\B\times\den\B)\simeq\den\B\otimes\den\B=\mathbb C^2\otimes\mathbb
C^2=\mathbb
C^4\).

If $\Gamma = x_1:\gB_1, ..., x_n:\gB_n$ is a context, then a $\Gamma$-valuation
is a function mapping each $x_i$ to $\den{\gB_i}$.

We now would associate to each term $t$ of type $A$ an element $\den t$ of $\den
A$. But as our calculus is probabilistic, due to the presence of a measurement
operator, we must associate to each term a set of elements of $\den A$.

Let $t$ be a term of type $A$ in $\Gamma$ and $\phi$ a $\Gamma$-valuation. We
define the interpretation of $t$, $\den t_\phi$ as follows.

\begin{align*}
  \den{x}_\phi &= \phi x\\
  \den{\lambda x^{\gB}.t}_{\phi} & = \{f~|~\forall a \in \den \gB, f a \in \den t_{\phi, x \mapsto\den{\gB}}\}\\
  \den{\ket{0}}_\phi &= \{\vect 10\}\\
  \den{\ket{1}}_\phi &= \{\vect 01\}\\
  \den{t \times u}_\phi &= \den{t}_\phi \times \den{u}_\phi\\
  \den{\pair tu}_\phi &= \{a + b~|~a \in \den{t}_\phi\mbox{ and }b \in \den{u}_{\phi}\}\\
  \den{\alpha.t}_\phi &= \{\alpha.a~|~a \in \den{t}_\phi\}\\
  \den{\z[B]}_\phi &= \{\vec 0\}, \mbox{the null vector of the vector space $\den{S(B)}$}\\
  \den{tu}_\phi &= \left\{\begin{array}{l} \left\{\sum\limits_{i\in
                            I}\alpha_i . g_i(a)~|~\sum\limits_{i\in I} \alpha_i .
                            g_i\in\den t_\phi, a\in\den u_\phi\right\}
                            {\textrm{If }}\Gamma\vdash t:\gB\Rightarrow A\\
                            \left\{\sum\limits_{i\in I}\sum\limits_{j\in J}\alpha_i . \beta_j . g_i(c_j)~|~\sum\limits_{i\in I} \alpha_i . g_i\in\den t_\phi, \sum\limits_{j\in J}\beta_j . c_j\in\den u_\phi\right\} \\
                            \multicolumn{1}{r}{\textrm{If }
                            \Gamma\vdash t:S(\gB\Rightarrow A)}
                          \end{array}
  \right.
  \\
  \den{\pi_j t}_\phi &=
                       \{
                       \prod_{h=1}^j b_{hk}\!\times\!\sum\limits_{i\in P}\!\left(\!\tfrac{\alpha_i}{\sqrt{\sum\limits_{r\in P}|\alpha_r|^2}}\!\right)\prod_{h=j+1}^mb_{hi}~|~\forall i\in P,\forall h, b_{hi}\!=\!b_{hk}
                       \}\\
               &\qquad\qquad\textrm{where }\den{t}_\phi = \{\sum\limits_{i=1}^n\prod_{h=1}^m b_{hi}\}
                 \textrm{ with }b_{hi}=\vect 01\textrm{ or }\vect 10
  \\
  \den{\ite{}{t}{r}}_\phi &=\{f~|~\forall a\in\den\B, fa=\left\{\begin{array}{ll}
                                                                  \den t_\phi & \textrm{If }a=\vect 01\\
                                                                  \den r_\phi & \textrm{If }a=\vect 10
                                                                \end{array}\right.\}\\
  \den{\head\ t}_\phi &=\{a_1~|~\prod_{i=1}^n a_i\in\den t_\phi, a_1\in\den{\B}\}\\
  \den{\tail\ t}_\phi &=\{\prod_{i=2}^na_i~|~\prod_{i=1}^na_i\in\den t_\phi,a_1\in\den{\B}\}\\
  \den{\Uparrow t}_\phi &= \den t_\phi
\end{align*}

\begin{lemma}
  \label{lem:inc}
  If $A\preceq B$, then $\den A\subseteq\den B$.
\end{lemma}
\begin{proof}
  By induction on the relation $\preceq$. Cf.~\ref{ap:denSem}. \qed
\end{proof}

\begin{lemma}
  \label{lem:subsDen}
  If $\Gamma\vdash t:A$ and $\phi,x\mapsto S,y\mapsto S$ is a
  $\Gamma$-valuation, then $\den t_{\phi,x\mapsto S,y\mapsto
    S}=\den{(x/y)t}_{\phi,x\mapsto S}$.
\end{lemma}
\begin{proof}
  By induction on $t$. Cf.~\ref{ap:denSem}. \qed
\end{proof}

\begin{theorem}\label{thm:soundness}
  If $\Gamma\vdash t:A$, and $\phi$ is a $\Gamma$-valuation. Then $\den
  t_\phi\subseteq\den A$.
\end{theorem}
\begin{proof}
  By induction on the typing derivation. Cf.~\ref{ap:denSem}. \qed
\end{proof}

\begin{theorem}\label{thm:denRed}
  If $\Gamma\vdash t:A$, $\phi$ is a $\Gamma$-valuation, and $t\lra[p_i] r_i$,
  with $\sum_ip_i=1$, then $\den t_\phi=\bigcup_i\den{r_i}_\phi$.
\end{theorem}
\begin{proof}
  We proceed by induction on the rewrite relation.
  \begin{description}
  \item[\rbetab\ and \rbetan] Let $\vdash(\lambda x^{\gB}.t)u:A$, with $\vdash
    u:\gB$, where, if $\gB=\B^n$, then $u\in\basis$. Then by
    Lemma~\ref{lem:generation}, one of the following possibilities happens:
    \begin{enumerate}
    \item $\vdash\lambda x^{\gB}.t:\gB'\Rightarrow B$ and $\vdash u:\gB'$, with
      $B\preceq A$. Thus, $\den{(\lambda x^{\gB}.t)u}_\phi=\{f(a)\mid
      a\in\den u_\phi\}\subseteq\den t_{\phi,x\mapsto\den{\gB}} $.
    \item $\vdash\lambda x^{\gB}.t:S(\gB'\Rightarrow B)$ and $\vdash u:S(\gB')$,
      with $S(B)\preceq A$. Thus, $\den{(\lambda x^{\gB}.t)u}_\phi=\{\sum_{j\in
        J}\beta_j . f(c_j)\mid \sum_{j\in J}\beta_j .
      c_j\in\den{u}_\phi\}\subseteq\den t_{\phi,x\mapsto\den{\gB}}$.
    \end{enumerate}
    In any case, by Lemma~\ref{lem:subsDen}, $\den
    t_{\phi,x\mapsto\den{\gB}}=\den{(u/x)t}_\phi$.
  \item[Other cases] All the remaining cases are straightforward by the
    algebraic nature of the interpretation.
    \qed
  \end{description}
\end{proof}

\section{Examples}\label{sec:examples}
In this section we show that our language is expressive enough to express the
Deutsch algorithm (Section~\ref{ex:deut}) and the Teleportation algorithm.
(Section~\ref{ex:telep}).

\subsection{Deutsch algorithm}\label{ex:deut}
The Deutsch algorithm tests whether the binary function $f$ implemented by the
oracle $U_f$ is constant ($f(0)=f(1)$) or balanced ($f(0)\neq f(1)$). The
algorithm is as follows: it starts with a qubit in state $\ket 0$ and another in
state $\ket 1$, and applies Hadamard gates to both. Then it applies the $U_f$
operator, followed by a Hadamard and a measurement to the first qubit. When the
function is constant, the first qubit ends in $\ket 0$, when it is balanced, it
ends in $\ket 1$.

The Hadamard gate ($H$) produces $\tfrac 1{\sqrt 2}.\pair{\ket 0}{\ket 1}$ when
applied to $\ket 0$ and $\tfrac 1{\sqrt 2}.\npair{\ket 0}{\ket 1}$ when applied
to $\ket 1$. Hence, it can be implemented with the if-then-else construction:
$\s H =\lambda x^{\B}.\frac 1{\sqrt 2}.\pair{\ket 0}{(\ite x{(-\ket 1)}{\ket
    1})}$. Notice that the abstracted variable has a base type
(i.e.~non-linear). Hence, if $H$ is applied to a superposition, say
$\pair{\alpha.\ket 0}{\beta.\ket 1}$, it reduces, as expected, in the following
way:
\[
  \s H\pair{\alpha.\ket 0}{\beta.\ket 1} \red\rlinr \pair{\s H\alpha.\ket 0}{\s
    H\beta.\ket 1} \red{\rlinscalr^2} \pair{\alpha.\s H\ket 0}{\beta.\s H\ket 1}
\]
and then is applied to the base terms. We define $\s H_1$ as the function taking
a two-qubits system and applying $\s H$ to the first. $\s H_1 =\lambda
x^{\B\times\B}.((\s H~(\head~x))\times(\tail~x))$. Similarly, $\s
H_{\textsl{both}}$ applies $\s H$ to both qubits.
\[\s H_{\textsl{both}} =
  \lambda x^{\B\times\B}.((\s H~(\head~x))\times(\s H~(\tail~x)))\] The gate
$U_f$ is called {\em oracle}, and it is defined by $U_f\ket{xy} = \ket{x,y\oplus
  f(x)}$ where $\oplus$ is the addition modulo $2$. In order to implement it, we
need the $\mathit{not}$ gate, which can be implemented similarly to the Hadamard
gate:
\[
  \s{not}=\lambda x^{\B}.(\ite x{\ket 0}{\ket 1})
\]
Then, the $U_f$ gate is implemented by:
\[
  \s U_f = \lambda x^{\B\times\B}.
  ((\head~x)\times(\ite{(\tail~x)}{(\s{not}~(f~(\head~x)))}{(f~(\head~x))}))
\]
where $f$ is a given term of type $\B\Rightarrow\B$.

Finally, the Deutsch algorithm combines all the previous definitions:
\[\s{Deutsch}_f = \pi_1~(\Uparrow_r \s H_1~(\s
  U_f~\Uparrow_\ell\Uparrow_r \s H_{\textsl{both}}~(\ket 0\times\ket 1)))\]

The casts after the Hadamards are needed to fully develop the qubits and then be
able to use it as an argument of a non-linear abstraction (i.e.~an abstraction
expecting for base terms and so linear-distributing over superpositions). The
$\s{Deutsch}_f$ term is typed, as expected, by $\vdash \s{Deutsch}_f:\B\times
S(\B)$.

This term, on the identity function, reduces as follows:
\[
  \s{Deutsch}_{id} \lra^* \pi_1\npair{\frac 1{\sqrt 2}.\ket{1} \times
    \ket{0}}{\frac 1{\sqrt 2}.\ket{1} \times \ket{1}} \red{\rproj} \ket
  1\times\npair{\frac 1{\sqrt 2}.\ket{0}}{\frac 1{\sqrt 2}.\ket{1}}
\]

The trace on this reduction and the type derivation are given in~\ref{ap:Deutsch}.

\subsection{Teleportation algorithm}\label{ex:telep}

The circuit for this algorithm is given in Figure~\ref{fig:telep}.
\begin{figure}[t]
  \begin{center}
    \begin{tikzpicture}
      \node[anchor=west] at (0,0) {$ \Qcircuit @C=1em @R=1em {
          & \ctrl{1} \qw & \gate{H} \qw & \meter & \controlo \cw \cwx[1] \\
          & \targ \qw    & \qw          & \meter & \controlo \cw \cwx \\
          & \qw & \qw & \qw & \gate{Z^{b_1}not^{b_2}} \cwx & \qw } $ };
      \node[anchor=east] at (0,.9) {$\ket\psi$}; \node[anchor=west] at (5,-.85)
      {$\ket\psi$}; \node[anchor=east] at (0,-.4) {$\beta_{00}$}; \draw[dashed]
      (.3,-.3) -- (4.2,-.3) -- (4.2,1.2) -- (.3,1.2) -- (.3,-.3); \node at
      (4.7,1.1) {Alice}; \draw[dashed] (.3,-.5) -- (4.7,-.5) -- (4.7,-1.2) --
      (.3,-1.2) -- (.3,-.5); \node at (4.8,-.3) {Bob};
    \end{tikzpicture}
  \end{center}
  \caption{Teleportation circuit}
  \label{fig:telep}
\end{figure}

The $\mathit{cnot}$ gate, which applies $\mathit{not}$ to the second qubit only
when the first qubit is $\ket 1$, can be implemented with an if-then-else
construction as follows:
\[\s{cnot} = \lambda
  x^{\B\times\B}.((\head~x)\times(\ite{(\head~x)}{(\s{not}~(\tail~x))}{(\tail~x)}))\]

We define $\s H^3_1$ to apply $H$ to the first qubit of a three-qubit system.
\[\s H^3_1 =
  \lambda x^{\B\times\B\times\B}. ((\s H~(\head~x))\times(\tail~x))\]

Remark that the only difference with $\s H_1$ is the type of the abstracted
variable. In addition, we need to apply $cnot$ to the two first qubits, so we
define $\s{cnot}^3_{12}$ as
\[
  \s{cnot}^3_{12} = \lambda x^{\B\times\B\times\B}.
  ((\s{cnot}~(\head~x\times(\head~\tail~x)))\times (\tail~\tail~x))
\]

The $Z$ gate returns $\ket 0$ when it receives $\ket 0$, and $-\ket 1$ when it
receives $\ket 1$. Hence, it can be implemented by:
\[\s Z = \lambda x^{\B}.(\ite x{(-\ket 1)}{\ket 0})\]

The Bob side of the algorithm will apply $Z$ and/or $\mathit{not}$ according to
the bits it receives from Alice. Hence, for any $\vdash \s U:\B\Rightarrow
S(\B)$ or $\vdash \s U:\B\Rightarrow\B$, we define $\s U^{(b)}$ to be the
function which depending on the value of a base qubit $b$ applies the $\s U$
gate or not:
\[
  \s U^{(b)} =(\lambda x^{\B}.\lambda y^{\B}.(\ite x{\s Uy}y))\ b
\]

Alice and Bob parts of the algorithm are defined separately:
\[
  \s{Alice} = \lambda x:S(\B)\times S(\B\times
  \B)(\pi_2(\Uparrow_r\s
  H^3_1~(\s{cnot}^3_{12}\Uparrow_\ell\Uparrow_r x)))
\]

Notice that before passing to $\s{cnot}_{12}^3$ the parameter of type
$S(\B)\times S(\B\times\B)$, we need to fully develop the term using the two
casts, and again, after the Hadamard gate. Bob side is implemented by
\[ \s{Bob} = \lambda x^{\B\times\B\times\B}.(\s Z^{(\head~x)}(\s{not}^{(\head\
    \tail\ x)}~(\tail~\tail~x)))\]

The teleportation is applied to an arbitrary qubit and to the following Bell
state
\[\beta_{00} = \pair{\frac 1{\sqrt 2}.\ket{0} \times \ket{0}}{\frac
    1{\sqrt 2}.\ket{1} \times \ket{1}}\] and it is defined by:
\[\s{Teleportation} = \lambda q^{S(\B)}.(\s{Bob}(\Uparrow_\ell \s{Alice}~(q\times\beta_{00})))\]

This term is typed, as expected, by: $\vdash\s{Teleportation}:S(\B)\Rightarrow
S(\B)$ and applying the teleportation to any superposition $\pair{\alpha.\ket
  0}{\beta.\ket 1}$ will reduce, as expected, to $\pair{\alpha.\ket
  0}{\beta.\ket 1}$. The trace on this reduction and the type derivation are
given in~\ref{ap:telep}.

\section{Conclusion}\label{sec:conclusion}

In this paper we have proposed a way to unify logic-linear and algebraic-linear
quantum $\lambda$-calculi, by interpreting $\lambda$-terms as linear functions
when they expect duplicable data and as non-linear ones when they do not, and
illustrated this idea with the definition of a calculus.

This calculus is first-order in the sense that variables do not have functional
types. In a higher-order version we should expect abstractions to be clonable.
But, allowing cloning abstractions allows cloning superpositions, by hiding them
inside. For example, $\lambda x^{\B\Rightarrow\B}.\pair{\frac 1{\sqrt 2}.\ket
  0}{\frac 1{\sqrt 2}.\ket 1}$. It has been argued
\cite{ArrighiDowekRTA08,ArrighiDiazcaroValironIC17} that what is cloned is not
the superposition but a function that creates the superposition, because we had
no way there to create such an abstraction from an arbitrary superposition. The
situation is different in the calculus presented in this paper as the term
$\lambda x^{S(\B)}.\lambda y^{\B}.x$ precisely takes any term $t$ of type
$S(\B)$ and returns the term $\lambda y^{\B}.t$. So, a cloning machine could be
constructed by encapsulating any superposition $t$ under a lambda, which
transform it into a basis term, so a clonable term. Extending this calculus to
the higher-order will require characterizing precisely the abstractions that can
be taken as arguments, not allowing to duplicate functions creating
superpositions.

\subsubsection*{Acknowledgements}
We would like to thank Eduardo Bonelli, Octavio Malherbe, Luca Paolini, Simona
Ronchi della Rocca and Luca Roversi for interesting comments and suggestions in
an early draft of this work.

\bibliographystyle{elsarticle-num} \bibliography{biblio}

\newpage
\appendix

\section{Detailed proofs of Section~\ref{sec:SR} (Subject reduction)}\label{ap:SR}
\xrecap{Corollary}{Simplification}{cor:simplification}{
  \begin{enumerate}
  \item If $\vdash\pair tu:A$, then $\vdash t:A$ and $\vdash u:A$.
  \item If $\vdash\pair tu:A$, then $A=S(B)$.
  \item If $\vdash \alpha.t:A$, then $\vdash t:A$.
  \item If $\vdash\alpha.t:A$, then $\vdash\beta.t:A$.
  \item If $\vdash \alpha.t:A$, then $A=S(B)$.
  \end{enumerate}
}
\begin{proof}
  ~
  \begin{enumerate}
  \item By Lemma~\ref{lem:generation}, $\vdash t:B$ and $\vdash u:B$, with
    $B\preceq S(B)\preceq A$, then, we conclude by rule $\preceq$.
  \item By Lemma~\ref{lem:generation}, $\vdash t:C$ and $\vdash u:C$, with
    $C\preceq S(C)\preceq A$, but then, by Lemma~\ref{lem:noBasisPrec}, $A=S(B)$
    for some type $B$.
  \item By Lemma~\ref{lem:generation}, $\vdash t:B$, with $S(B)\preceq A$, then,
    we conclude by rule $\preceq$.
  \item By Lemma~\ref{lem:generation}, $\vdash t:B$, with $S(B)\preceq A$, then
    we conclude by rules $S_I^\alpha$ and $\preceq$.
  \item By Lemma~\ref{lem:generation}, $\vdash t:C$ with $S(C)\preceq A$, but
    then, by Lemma~\ref{lem:noBasisPrec}, $A=S(B)$ for some type $B$.
    \qed
  \end{enumerate}
\end{proof}

\recap{Corollary}{cor:basisTermNoSup}{If $b\in\basis$ and $\vdash b:S(A)$, then $\vdash b:A$.}
\begin{proof}
  We proceed by cases on $b$.
  \begin{itemize}
  \item Let $b=\lambda x^{\gB}.t$. Then, by Lemma~\ref{lem:generation},
    $x:\gB\vdash t:B$, with $\gB\Rightarrow B\preceq S(A)$, and so
    $\gB\Rightarrow B\preceq A$, and we conclude by rule $\preceq$.
  \item Let $b=\ket 0$. Then, by Lemma~\ref{lem:generation}, $\B\preceq S(A)$,
    hence $\B\preceq A$ and we conclude by rule $\preceq$.
  \item Let $b=\ket 1$. Analogous to previous case.
  \item Let $b=b_1\times b_2$. Then, by Lemma~\ref{lem:generation}, $\vdash
    b_1:B_1$, $\vdash b_2:B_2$, and $B_1\times B_2\preceq S(A)$. Hence,
    $B_1\times B_2\preceq A$ and we conclude by rule $\preceq$.
    \qed
  \end{itemize}
\end{proof}

\xrecap{Lemma}{Substitution lemma}{lem:substitution}{
  Let $FV(u)=\emptyset$, then if $\Gamma,x:\gB\vdash t:A$, $\Delta\vdash u:\gB$,
  where if $\gB=\B^n$ then $u\in\basis$, we have $\Gamma,\Delta\vdash (u/x)t:A$.
}
\begin{proof}
  Notice that due to Lemma~\ref{lem:linearContext}, $|\Delta|\subseteq\bqtypes$,
  hence, it suffices to consider $\Delta=\emptyset$. We proceed by structural
  induction on $t$.

  The set of terms be divided in the following groups:
  \begin{align*}
    \s{unclassified} &:= x\mid \lambda x^{\gB}.t\\
    \s{arity}^0 &:= \z\mid \ket 0\mid \ket 1\\
    \s{arity}^1(r) &:= \pi_j r\mid \alpha.r\mid \head\ r\mid \tail\ r\mid \Uparrow_r t\mid\Uparrow_\ell t\\
    \s{arity}^2(r)(s) &:= rs\mid \pair rs\mid r\times s\mid \ite{}rs
  \end{align*}
  Hence, we can consider the terms by groups:
  \begin{description}
  \item[\s{unclassified} terms] ~
    \begin{description}
    \item[$t=x$.] By Lemma~\ref{lem:generation}, 
      $|\Gamma|\subseteq\bqtypes$ and $\gB\preceq A$. Since $(u/x)x=u$, we have
      $\vdash (u/x)x:\gB$. Hence, since $\gB\preceq A$, by rule $\preceq$,
      $\vdash(u/x)x:A$. Finally, since $|\Gamma|\subseteq\bqtypes$, by rule $W$,
      we have $\Gamma\vdash (u/x)x:A$.
    \item[$t=y\neq x$.] By Lemma~\ref{lem:generation}, $y:\gB'\in\Gamma$,
      $(|\Gamma|\cup\{\gB\})\setminus\{\gB'\}\subseteq\bqtypes$ and $\gB'\preceq
      A$. Hence, by rule $\preceq$, $y:\gB'\vdash y:A$. Since
      $|\Gamma|\subseteq\bqtypes$, by rule $W$, we have $\Gamma\vdash y:A$.
      Finally, since $(u/x)y=y$, we have $\Gamma\vdash (u/x)y:A$.
    \item[$t=\lambda y^{\gB'}.v$.] Without loss of generality, assume $y$ is
      does not appear free in $u$. By Lemma~\ref{lem:generation},
      $\Gamma',y:\gB'\vdash v:B$, with $\Gamma'\subseteq\Gamma\cup\{x:\gB\}$,
      $\gB'\Rightarrow B\preceq A$ and
      $(|\Gamma|\cup\{\gB\})\setminus|\Gamma'|\subseteq\bqtypes$. By the induction
      hypothesis, $\Gamma'',y:\gB'\vdash (u/x)v:B$, with
      $\Gamma''=\Gamma'\setminus\{x:\gB\}$. Notice that if $x:\gB\in\Gamma'$,
      the induction hypothesis applies directly, in other case, $\gB\in\bqtypes$
      and so by rule $W$ the context can be enlarged to include $x:\gB$, hence
      the induction hypothesis applies in any case. Therefore, by rule
      $\Rightarrow_I$, $\Gamma''\vdash\lambda y^{\gB'}.(u/x)v:\gB'\Rightarrow
      B$. Since $\gB'\Rightarrow B\preceq A$, by rule $\preceq$,
      $\Gamma''\vdash\lambda y^{\gB'}.(u/x)v:A$. Hence, since
      $|\Gamma|\setminus|\Gamma''|\subseteq\bqtypes$, by rule $W$,
      $\Gamma\vdash\lambda y^{\gB'}.(u/x)v:A$. Since $y$ does not appear free in
      $u$, $\lambda y^{\gB'}.(u/x)v = (u/x)(\lambda y^{\gB'}.v)$. Therefore,
      $\Gamma\vdash(u/x)(\lambda y^{\gB'}.v):A$.
    \end{description}
  \item[$\s{arity}^0$ terms] All of these terms are typed by an axiom with a
    type $B$ which, by Lemma~\ref{lem:generation}, $B\preceq A$. Also, by the
    same Lemma, $|\Gamma,x:\gB|\subseteq\bqtypes$. So, we can type with the axiom,
    and empty context, $\vdash\s{arity}^0:B$, and so, by rule $W$,
    $\Gamma\vdash\s{arity}^0:B$. Notice that $\s{arity}^0=(u/x)\s{arity}^0$. We
    conclude by rule $\preceq$.
  \item[$\s{arity}^1(r)$ terms] By Lemma~\ref{lem:generation}, $\Gamma'\vdash
    r:B$, such that by a derivation tree $T$, $\Gamma'\vdash\s{arity}^1(r):C$,
    where $\Gamma'\subseteq(\Gamma\cup\{x:\gB\})$,
    $(|\Gamma|\cup\gB)\setminus|\Gamma'|\subseteq\bqtypes$ and $C\preceq A$. If
    $x:\gB\notin\Gamma'$, then $\gB=B$ and so we can extend $\Gamma'$ with
    $x:\gB$. Hence, in any case, by the induction hypothesis,
    $\Gamma'\setminus\{x:\gB\}\vdash (u/x)r:C$. Then, using the derivation tree
    $T$, $\Gamma'\setminus\{x:\gB\}\vdash\s{arity}^1((u/x)r):C$. Notice that
    $\s{arity}^1((u/x)r)=(u/x)\s{arity}^1(r)$. We conclude by rules $W$ and
    $\preceq$.

  \item[$\s{arity}^2(r)(s)$ terms] By Lemma~\ref{lem:generation},
    $\Gamma_1\vdash r:C$ and $\Gamma_2\vdash s:D$, such that by a typing rule R,
    $\Gamma_1,\Gamma_2\vdash\s{arity}^2(r)(s):E$, with $E\preceq A$, and where
    $(\Gamma_1\cup\Gamma_2)\subseteq(\Gamma\cup\{x:\gB\})$ and
    $(|\Gamma|\cup\gB)\setminus(|\Gamma_1|\cup|\Gamma_2|)\subseteq\bqtypes$.
    Therefore, if $x:\gB\notin\Gamma_i$, $i=1,2$, we can extend $\Gamma_i$ with
    $x:\gB$ using rule $W$. Hence, by the induction hypothesis,
    $\Gamma_1\setminus\{x:\gB\}\vdash (u/x)r:C$ and
    $\Gamma_2\setminus\{x:\gB\}\vdash (u/x)s:D$. So, by rule R,
    $\Gamma_1\setminus\{x:\gB\},\Gamma_2\setminus\{x:\gB\}\vdash
    \s{arity}^2((u/x)r)((u/x)s):E$. Notice that
    $\s{arity}^2((u/x)r)((u/x)s)=(u/x)\s{arity}^2(r)(s)$. We conclude by rules
    $W$ and $\preceq$.
    \qed
  \end{description}
\end{proof}

\xrecap{Theorem}{Subject reduction on closed terms}{thm:SR}{
  For any closed terms $t$ and $u$ and type $A$, if $t\lra[p] u$ and $\vdash
  t:A$, then $\vdash u:A$.
}
\begin{proof}
  We proceed by induction on the rewrite relation.
  \begin{description}
  \item[\rbetab\ and \rbetan] Let $\vdash(\lambda x^{\gB}.t)u:A$, with $\vdash
    u:\gB$, where, if $\gB=\B^n$, then $u\in\basis$. Then by
    Lemma~\ref{lem:generation}, one of the following possibilities happens:
    \begin{enumerate}
    \item\label{it:caseEbeta} $\vdash\lambda x^{\gB}.t:\gB'\Rightarrow B$ and
      $\vdash u:\gB'$, with $B\preceq A$, or
    \item\label{it:caseESbeta} $\vdash\lambda x^{\gB}.t:S(\gB'\Rightarrow B)$
      and $\vdash u:S(\gB')$, with $S(B)\preceq A$.
    \end{enumerate}
    Thus, in any case, by Lemma~\ref{lem:generation} again, $x:\gB\vdash t:C$,
    with, in case \ref{it:caseEbeta}, $\gB\Rightarrow C\preceq \gB'\Rightarrow
    B$ and in case \ref{it:caseESbeta}, $\gB\Rightarrow C\preceq
    S(\gB'\Rightarrow B)$. Hence, $\gB=\gB'$ and in the first case $C\preceq
    B\preceq A$, while in the second, $C\preceq B\preceq S(B)\preceq A$, so, in
    general $C\preceq A$. Since $\vdash u:\gB$, where if $\gB=\B^n$, then
    $u\in\basis$, by Lemma~\ref{lem:substitution}, $\vdash (u/x)t:C$, and by
    rule $\preceq$, $\vdash (u/x)t:A$.

  \item[\riftrue] Let $\vdash\ite{\ket 1}uv:A$. Then, by
    Lemma~\ref{lem:generation}, one of the following possibilities happens:
    \begin{itemize}
    \item $\vdash\ite{}uv:\Psi\Rightarrow B$ and $\vdash\ket 1:\Psi$, with
      $B\preceq A$.
      Then, by
      Lemma~\ref{lem:generation} again,
      $\vdash u:C$, $\vdash v:C$ and $\B\Rightarrow C\preceq\Psi\Rightarrow
      B$. Hence, $\Psi=\B$ and $C\preceq B\preceq A$.
    \item $\vdash\ite{}uv:S(\Psi\Rightarrow B)$ and $\vdash \ket 1:S(\Psi)$,
      with $S(B)\preceq A$. Then, by Lemma~\ref{lem:generation} again,
      $\vdash u:C$, $\vdash v:C$ and $\B\Rightarrow C\preceq S(\Psi\Rightarrow
      B)$. Hence, $\Psi=\B$ and $C\preceq B\preceq S(B)\preceq A$.
    \end{itemize}
    So, by rule $\preceq$, $\vdash u:A$.
  \item[\riffalse] Analogous to case \riftrue.
  \item[\rlinr] Let $\vdash t\pair uv:A$, with $\vdash t:\B^n\Rightarrow B$. Then,
    by Lemma~\ref{lem:generation}, one of the following cases happens:
    \begin{enumerate}
    \item $\vdash t:\gB\Rightarrow C$ and $\vdash\pair uv:\gB$, with $C\preceq
      A$. However, since $\vdash t:\B^n\Rightarrow B$, we have $\gB\in\bqtypes$,
      which is impossible due to Corollary~\ref{cor:simplification}.
    \item $\vdash t:S(\gB\Rightarrow C)$ and $\vdash\pair uv:S(\gB)$, with
      $S(C)\preceq A$. Then, by Corollary~\ref{cor:simplification}, $\vdash
      u:S(\gB)$ and $\vdash v:S(\gB)$. Hence,
      \[
        \hspace{-5mm}  \infer[\preceq] {\vdash\pair{tu}{tv}:A} { \infer[S_I^+]
          {\vdash\pair{tu}{tv}:S(S(C))} { \infer[\Rightarrow_{ES}] {\vdash
              tu:S(C)} { \vdash t:S(\gB\Rightarrow C) & \vdash u:S(\gB) } &
            \infer[\Rightarrow_{ES}] {\vdash tv:S(C)} { \vdash
              t:S(\gB\Rightarrow C) & \vdash v:S(\gB) } } }
      \]
    \end{enumerate}
  \item[\rlinscalr] Let $\vdash t(\alpha.u):A$, with $\vdash t:\B^n\Rightarrow B$.
    Then, by Lemma~\ref{lem:generation}, one of the following cases happens:
    \begin{enumerate}
    \item $\vdash t:\gB\Rightarrow C$ and $\vdash \alpha.u:\gB$, with $C\preceq
      A$. However, since $\vdash t:\B^n\Rightarrow B$, we have $\gB\in\bqtypes$,
      which is impossible due to Corollary~\ref{cor:simplification}.
    \item $\vdash t:S(\gB\Rightarrow C)$ and $\vdash \alpha.u:S(\gB)$, with
      $S(C)\preceq A$. Then, by Corollary~\ref{cor:simplification}, $\vdash
      u:S(\gB)$. Hence,
      \[
        \infer[\preceq] {\vdash \alpha.tu:A} { \infer[S_I^\alpha] {\vdash
            \alpha.tu:S(S(C))} { \infer[\Rightarrow_{ES}] {\vdash tu:S(C)} {
              \vdash t:S(\gB\Rightarrow C) & \vdash u:S(\gB) } } }
      \]
    \end{enumerate}
  \item[\rlinzr] Let $\vdash t\z[\B^n]:A$, with $\vdash t:\B^n\Rightarrow B$. Then, by
    Lemma~\ref{lem:generation}, one of the following cases happens:
    \begin{enumerate}
    \item $\vdash t:\gB\Rightarrow C$ and $\vdash\z[\B^n]:\gB$, with $C\preceq A$.
      Then, by Lemma~\ref{lem:generation} again, $S(\B^n)\preceq \gB$. However,
      since $\vdash t:\B^n\Rightarrow B$, $\gB\in\bqtypes$, which is impossible by
      Lemma~\ref{lem:noBasisPrec}.
    \item $\vdash t:S(\gB\Rightarrow C)$ and $\vdash\z[\B^n]:S(\gB)$, with
      $S(C)\preceq A$. By rule $\tax_{\vec 0}$, $\vdash
      \z[\mathsf{min}(B)]:S(\mathsf{min}(B))$.
      Since $\vdash t:\B^n\Rightarrow B$ and $\vdash t:S(\Psi\Rightarrow C)$, by
      Lemma~\ref{lem:min}, we have $\mathsf{min}(\B^n\Rightarrow
      B)=\mathsf{min}(S(\Psi\Rightarrow C))$, so
      $\mathsf{min}(B)=\mathsf{min}(C)$. Then, by Lemma~\ref{lem:minlqA}, $\mathsf{min}(B)=\mathsf{min}(C)\preceq C$, then $S(\mathsf{min}(B))\preceq S(C)\preceq A$, so we conclude by rule $\preceq$.
    \end{enumerate}
  \item[\rlinl] Let $\vdash\pair tuv:A$. Then by Lemma~\ref{lem:generation}, one
    of the following cases happens:
    \begin{enumerate}
    \item $\vdash\pair tu:\gB\Rightarrow B$, which is impossible by
      Corollary~\ref{cor:simplification}.
    \item $\vdash\pair tu:S(\gB\Rightarrow B)$ and $\vdash v:S(\gB)$, with
      $S(B)\preceq A$. Then, by Corollary~\ref{cor:simplification}, $\vdash
      t:S(\gB\Rightarrow B)$ and $\vdash u:S(\gB\Rightarrow B)$. Hence,
      \[
        \hspace{-5mm} \infer[\preceq] {\vdash\pair{tv}{uv}:A} { \infer[S_I^+]
          {\vdash\pair{tv}{uv}:S(S(B))} { \infer[\Rightarrow_{ES}] {\vdash
              tv:S(B)} { {\vdash t:S(\gB\Rightarrow B)} & {\vdash v:S(\gB)} } &
            \infer[\Rightarrow_{ES}] {\vdash uv:S(B)} { {\vdash
                u:S(\gB\Rightarrow B)} & {\vdash v:S(\gB)} } } }
      \]
    \end{enumerate}
  \item[\rlinscall] Let $\vdash (\alpha.t)u:A$. Then, by
    Lemma~\ref{lem:generation}, one of the following cases happens:
    \begin{enumerate}
    \item $\vdash \alpha.t:\gB\Rightarrow B$, which is impossible by
      Corollary~\ref{cor:simplification}.
    \item $\vdash \alpha.t:S(\gB\Rightarrow B)$ and $\vdash u:S(\gB)$, with
      $S(B)\preceq A$. Then, by Corollary~\ref{cor:simplification}, $\vdash
      t:S(\gB\Rightarrow B)$. Hence,
      \[
        \infer[\preceq] {\vdash \alpha.tu:A} { \infer[S_I^\alpha] {\vdash
            \alpha.tu:S(S(B))} { \infer[\Rightarrow_{ES}] {\vdash tu:S(B)} {
              \vdash t:S(\gB\Rightarrow B) & \vdash u:S(\gB) } } }
      \]
    \end{enumerate}
  \item[\rlinzl] Let $\vdash \z[\B\Rightarrow B] t:A$. Then, by
    Lemma~\ref{lem:generation}, one of the following cases happens:
    \begin{enumerate}
    \item $\vdash \z[\B\Rightarrow B]:\gB\Rightarrow C$ and $\vdash t:\gB$, with
      $C\preceq A$. Then, by Lemma~\ref{lem:generation} again, $S(\B\Rightarrow
      B)\preceq \gB\Rightarrow C$, which is impossible by
      Lemma~\ref{lem:noBasisPrec}.
    \item $\vdash\z[\B\Rightarrow B]:S(\gB\Rightarrow C)$ and $\vdash t:S(\gB)$,
      with $S(C)\preceq A$. By Lemma~\ref{lem:generation} again, $S(\B\Rightarrow
      B)\preceq S(\gB\Rightarrow C)$.
      By Lemma~\ref{lem:min}, $\mathsf{min}(\B\Rightarrow
      B)=\mathsf{min}(\Psi\Rightarrow C)$, so $\mathsf{min}(B)=\mathsf{min}(C)$,
      and by Lemma~\ref{lem:minlqA}, $\mathsf{min}(C)\preceq C$, hence, 
      $\mathsf{min}(B)\preceq C$, and then $S(\mathsf{min}(B))\preceq
      S(C)\preceq A$. By rule $\tax_{\vec 0}$, $\vdash \z[\mathsf{min}(B)]:S(\mathsf{min}(B))$, hence we
      conclude by rule $\preceq$.
    \end{enumerate}
  \item[\rneut] Let $\vdash \pair\z t:A$. Then, by
    Corollary~\ref{cor:simplification}, $\vdash t:A$.
  \item[\runit] Let $\vdash 1.t:A$. Then, by Corollary~\ref{cor:simplification},
    $\vdash t:A$.
  \item[\rzeros] Let $\vdash 0.t:A$, with $\vdash t:B$. Then, we must show that $\vdash\z[\mathsf{min}(B)]:A$.
    
    By Lemma~\ref{lem:generation}, $\vdash t:C$ and $S(C)\preceq A$.
    By Lemma~\ref{lem:min}, $\mathsf{min}(B)=\mathsf{min}(C)$. By
    Lemma~\ref{lem:minlqA}, $\mathsf{min}(C)\preceq C$. Therefore,
    $S(\mathsf{min}(B))=S(\mathsf{min}(C))\preceq S(C)\preceq A$.
    By rule $\tax_{\vec 0}$, $\vdash\z[\mathsf{min}(B)]:S(\mathsf{min}(B))$,
    hence we conclude by rule $\preceq$.
  \item[\rzero] Let $\vdash \alpha.\z[B]:A$.
    By Lemma~\ref{lem:generation}, $\vdash\z[B]:C$ with $S(C)\preceq A$.
    Then, by Lemma~\ref{lem:generation} again, $S(B)\preceq C$.
    In addition, by Lemma~\ref{lem:minlqA}, $\mathsf{min}(B)\preceq B$.
    Therefore, $S(\mathsf{min}(B))\preceq S(B)\preceq C\preceq S(C)\preceq A$.
    Since, by rule $\tax_{\vec 0}$, $\vdash\z[\mathsf{min}(B)]:S(\mathsf{min}(B))$, we
    conclude by rule $\preceq$ that $\vdash\z[\mathsf{min}(B)]:A$.
  \item[\rprod] Let $\vdash \alpha.(\beta.t):A$. By
    Corollary~\ref{cor:simplification}, $\vdash\beta.t:A$. Then, by
    Corollary~\ref{cor:simplification} again, $\vdash(\alpha\times\beta).t:A$.
  \item[\rdists] Let $\vdash \alpha.\pair tu:A$. By Lemma~\ref{lem:generation},
    $\vdash\pair tu:B$, with $S(B)\preceq A$. Then, by
    Corollary~\ref{cor:simplification}, $\vdash t:B$ and $\vdash u:B$. Hence, by
    rule $S_I^\alpha$, $\vdash\alpha.t:S(B)$ and $\vdash\alpha.u:S(B)$. We
    conclude by rules $S_I^+$ and $\preceq$.
  \item[\rfact] Let $\vdash \pair{\alpha.t}{\beta.t}:A$. By
    Corollary~\ref{cor:simplification}, $\vdash\alpha.t:A$. Then, by
    Corollary~\ref{cor:simplification} again, $\vdash(\alpha+\beta).t:A$.
  \item[\rfacto] Let $\vdash \pair{\alpha.t}{t}:A$. By
    Corollary~\ref{cor:simplification}, $\vdash\alpha.t:A$. Then, by
    Corollary~\ref{cor:simplification} again, $\vdash(\alpha+1).t:A$.
  \item[\rfactt] Let $\vdash \pair{t}{t}:A$. By Lemma~\ref{lem:generation},
    $\vdash t:B$, with $S(B)\preceq A$. Then, by rule $S_I^\alpha$, $\vdash
    2.t:S(B)$. We conclude by rule $\preceq$.
  \item[\rzeroS] Let $\vdash\z[A]:B$. Then, by Lemma~\ref{lem:generation},
    $S(A)\preceq B$.
    By Lemma~\ref{lem:minlqA}, $\mathsf{min}(A)\preceq A$, hence
    $S(\mathsf{min}(A))\preceq S(A)$.
    By rule $\tax_{\vec 0}$, $\vdash\z[\mathsf{min}(A)]:S(\mathsf{min}(A))$, and
    since $S(\mathsf{min}(A))\preceq S(A)\preceq B$, we conclude by by rule
    $\preceq$.
  \item[\rcomm] Let $\vdash \pair uv:A$. By Lemma~\ref{lem:generation}, $\vdash
    u:B$ and $\vdash v:B$, with $S(B)\preceq A$. So,
    \[
      \infer[\preceq] {\vdash \pair vu:A} { \infer[S_I^+] {\vdash \pair vu:S(B)}
        {\vdash v:B\qquad\vdash u:B} }
    \]
  \item[\rassoc] Let $\vdash \pair{\pair uv}w:A$. By Lemma~\ref{lem:generation},
    $\vdash\pair uv:B$ and $\vdash w:B$, with $S(B)\preceq A$. Then, by
    Corollary~\ref{cor:simplification}, $\vdash u:B$ and $\vdash v:B$. Hence,
    \[
      \infer[\preceq] {\vdash \pair u{\pair vw}:A} { \infer[S_I^+] {\vdash \pair
          u{\pair vw}:S(S(B))} { \infer[\preceq] {\vdash u:S(B)} {\vdash u:B} &
          \infer[S_I^+] {\vdash\pair vw:S(B)} {\vdash v:B & \vdash w:B} } }
    \]
  \item[\rhead] Let $\vdash\head\ (v\times u):A$, with $v\neq t_1\times t_2$. Hence, by
    Lemma~\ref{lem:generation}, $\vdash v\times u:\B^n$, with $\B\preceq
    A$. Then, by Lemma~\ref{lem:generation} again, $\vdash v:B$ and $\vdash
    u:C$, with $B\times C\preceq \B^n$.
    Lemma~\ref{lem:lqBn}   
    $B\times C=\B^n$, so $B=\B^m$.
    Since $v\in\values$, by Lemma~\ref{lem:generation}, $B$ is not a product,
    and so $B=\B\preceq A$. Therefore, we conclude by rule $\preceq$.
  \item[\rtail] Analogous to case \rhead.
  \item[\rdistsumr] Let $\vdash\Uparrow_r(\pair rs\times
    u):A$. By Lemma~\ref{lem:generation}, $S(B\times C)\preceq A$ and
    $\vdash\pair rs\times u:S(S(B)\times C)$. Then, by the same Lemma,
    $\vdash\pair rs:D$ and $\vdash u:E$, with $D\times E\preceq S(S(B)\times
    C)$, so by Lemma~\ref{lem:SAxBmE}, there exists $F,G,n$ such that
    $S(S(B)\times C)=S^n(F\times G)$ and $D\preceq F$ and $E\preceq G$.
    Therefore, $n=1$ and $S(B)\times C=F\times G$.
    Since $\vdash \pair rs:D$, by Lemma~\ref{lem:generation}, there exists $H$
    such that $S(H)\preceq
    D$, so by transitivity $S(H)\preceq F$. Then, by
    Lemma~\ref{lem:noBasisPrec}, $F$ has the form $S(I)$.
    Therefore, neither $F$ nor $S(B)$ are products, and hence $S(B)=F$ and
    $C=G$.
    Hence, $D\preceq S(B)$ and $E\preceq C$, and hence, $\vdash\pair rs:S(B)$
    and $\vdash u:C$. Then, by Corollary~\ref{cor:simplification}, $\vdash
    r:S(B)$ and $\vdash s:S(B)$. Therefore,
    \[
      \infer[\preceq] {\vdash\pair{\Uparrow_r(r\times
          u)}{\Uparrow_r(s\times u)}:A} { \infer[S_I^+]
        {\vdash\pair{\Uparrow_r(r\times
            u)}{\Uparrow_r(s\times u)}:S(S(B\times C))} {
          \infer[\Uparrow_r] {\vdash\Uparrow_r(r\times
            u):S(B\times C)} { \infer[\preceq] {\vdash r\times u:S(S(B)\times
              C)} { \infer[\times_I] {\vdash r\times u:S(B)\times C} { \vdash
                r:S(B) & \vdash u:C } } } & \infer[\Uparrow_r]
          {\vdash\Uparrow_r(s\times u):S(B\times C)} {
            \infer[\preceq] {\vdash s\times u:S(S(B)\times C)} {
              \infer[\times_I] {\vdash s\times u:S(B)\times C} { \vdash
                s:S(B) & \vdash u:C } } } } }
    \]
  \item[\rdistsuml] Analogous to case \rdistsumr.
  \item[\rdistscalr] Let $\vdash\Uparrow_r((\alpha.r)\times u):A$. By Lemma~\ref{lem:generation}, $S(B\times
    C)\preceq A$, and $\vdash((\alpha.r)\times u):S(S(B)\times C)$. Then, by
    the same Lemma, $\vdash\alpha.r:D$ and $\vdash u:E$, with $D\times E\preceq
    S(S(B)\times C)$, so by Lemma~\ref{lem:SAxBmE}, there exists $F,G,n$ such that
    $S(S(B)\times C)=S^n(F\times G)$ and $D\preceq F$ and $E\preceq G$.
    Therefore, $n=1$ and $S(B)\times C=F\times G$.
    Since $\vdash \alpha.r:D$, by Lemma~\ref{lem:generation}, there exists $H$
    such that $S(H)\preceq
    D$, so by transitivity $S(H)\preceq F$. Then, by
    Lemma~\ref{lem:noBasisPrec}, $F$ has the form $S(I)$.
    Therefore, neither $F$ nor $S(B)$ are products, and hence $S(B)=F$ and
    $C=G$.
    Hence, $D\preceq S(B)$ and $E\preceq C$, so by rule
    $\preceq$, $\vdash\alpha.r:S(B)$ and $\vdash u:C$. By
    Corollary~\ref{cor:simplification}, $\vdash r:S(B)$. Therefore,
    \[
      \infer[\preceq] {\vdash\alpha.\Uparrow_r(r\times
        u):A} { \infer[S_I^\alpha] {\vdash\alpha.\Uparrow_r(r\times u):S(S(B\times C))} { \infer[\Uparrow_r]
          {\vdash\Uparrow_r(r\times u):S(B\times C)} {
            \infer[\preceq] {\vdash r\times u:S(S(B)\times C)} {
              \infer[\times_I] {\vdash r\times u:S(B)\times C} {\vdash r:S(B)
                & \vdash u:C} } } } }
    \]
  \item[\rdistscall] Analogous to case \rdistscalr.
  \item[\rdistzr] Let $\vdash\Uparrow_r(\z[B]\times u):A$.
    By Lemma~\ref{lem:generation}, $S(B\times C)\preceq A$. 
    By lemma~\ref{lem:minlqA}, $\mathsf{min}(B\times C)\preceq B\times C$,
    hence, $S(\mathsf{min}(B\times C))\preceq S(B\times C)$.
    By rule $\tax_{\vec
      0}$, $\vdash\z[\mathsf{min}(B\times C)]:S(\mathsf{min}(B\times C))$.
    Hence, since $S(\mathsf{min}(B\times C))\preceq S(B\times C)\preceq A$, we conclude by rule
    $\preceq$.
  \item[\rdistzl] Analogous to case \rdistzr.
  \item[\rdistcasum] Let $\vdash\Uparrow\pair tu:A$. Then, by
    Lemma~\ref{lem:generation}, $S(C\times D)\preceq A$, $\vdash\Uparrow
    t:S(C\times D)$ and $\vdash\Uparrow u:S(C\times D)$. We
    conclude by rules $S_I^+$ and $\preceq$.
  \item[\rdistcascal] Let $\vdash\Uparrow(\alpha.t):A$. Then, by
    Lemma~\ref{lem:generation}, $S(C\times D)\preceq A$, and $\vdash\Uparrow t:S(C\times D)$. We conclude by rules $S_I^\alpha$ and
    $\preceq$.
  \item[\rcaneutr] Let $\vdash \Uparrow_r(b\times r):A$,
    with $b\in\basis$. Then, by Lemma~\ref{lem:generation}, $\vdash b\times
    r:S(S(B)\times C)$ and $S(B\times C)\preceq A$. Then, by
    Lemma~\ref{lem:generation} again, $\vdash b:D$ and $\vdash r:E$, with
    $D\times E\preceq S(S(B)\times C)$.
    Without lost of generality, let $b=\prod_{i=1}^nb_i$ where each $b_i$ is not
    a product. Then, by Lemma~\ref{lem:generation}, $\vdash b_i:D_i$ with
    $\prod_{i=1}^n D_i\preceq D$, and, by the same lemma, $D_i$ are not
    products.
    Therefore, $D_1\times\prod_{i=2}^n D_i\times E\preceq S(S(B)\times C)$, so by Lemma~\ref{lem:SAxBmE}, there exists $F,G,n$ such that
    $S(S(B)\times C)=S^n(F\times G)$ and $D_1\preceq F$ and $\prod_{i=1}^n D_i\times E\preceq G$.
    Therefore, $n=1$ and $S(B)\times C=F\times G$.
    Since $D_1$ is not a product, $F$ is not a product. Hence, since 
    neither $F$ nor $S(B)$ are products, we have $S(B)=F$ and $C=G$.
    $D_1\preceq S(B)$ and $\prod_{i=2}^nD_i\times E\preceq C$,
    hence, $\vdash b_1:S(B)$ and $\vdash\prod_{i=2}^nb_i\times r:\prod_{i=2}^n
    D_i\times E$. Therefore, by
    Corollary~\ref{cor:basisTermNoSup}, $\vdash b_1:B$, and so, by rule
    $\times_I$, $\vdash b\times r:B\times C$, and by rule $\preceq$, $\vdash
    b\times r:S(B\times C)$.
  \item[\rcaneutl] Analogous to case $\rcaneutr$.
  \item[\rproj] Let $\vdash\pi_j(\sum_{i=1}^n \may[\alpha_i]\prod_{h=1}^m
    b_{hi}):A$.
    Then, by Lemma~\ref{lem:generation}, we have that
    $\B^j\times S(\B^{m'-1})\preceq A$. Hence, we have the derivation from Figure~\ref{fig:deriv}.
    \begin{figure}[t]
      \centering
      \[
        \infer[\preceq]{\vdash\prod\limits_{h=1}^jb_{hk}\times\sum\limits_{i\in P}\left(
            \frac{\alpha_i}{\sqrt{\sum\limits_{r\in P}|\alpha_r|^2}}
          \right)\prod\limits_{h=j+1}^ mb_{hi}:A }
        {
          \infer[\times_I]{\vdash\prod\limits_{h=1}^jb_{hk}\times\sum\limits_{i\in P}\left(
              \frac{\alpha_i}{\sqrt{\sum\limits_{r\in P}|\alpha_r|^2}}
            \right)\prod\limits_{h=j+1}^ mb_{hi}:\B^ j\times S(\B^ {n-j})}
          {
            \infer[\times_I]{\vdash\prod_{h=1}^ jb_{hk}:\B^ j}{\forall h &\infer[\tax_{\ket x}]{\vdash b_{hk}:\B}{}}
            &
            \infer[\preceq]{
              \vdash\sum\limits_{i\in
                P}\left(\frac{\alpha_i}{\sqrt{\sum\limits_{r\in
                      P}|\alpha_r|^2}}\right)\prod\limits_{h=j+1}^mb_{hi}:S(\B^{n-1})}
            {
              \infer[S_I^+]{
                \vdash\sum\limits_{i\in
                  P}\left(\frac{\alpha_i}{\sqrt{\sum\limits_{r\in
                        P}|\alpha_r|^2}}\right)\prod\limits_{h=j+1}^mb_{hi}:S(S(\B^{n-1}))
              }
              {
                \forall i\in P &
                \infer[S_I^\alpha]{
                  \vdash\left(\frac{\alpha_i}{\sqrt{\sum\limits_{r\in
                          P}|\alpha_r|^2}}\right)\prod\limits_{h=j+1}^mb_{hi}:S(\B^{n-1})
                }
                {
                  \infer[\times_I]{
                    \vdash\prod\limits_{h=j+1}^mb_{hi}:\B^{n-1}
                  }
                  {\forall h & \infer[\tax_{\ket x}]{\vdash b_{hi}:\B}{}}
                }
              }
            }
          }
        }
      \]
      \caption{Derivation from case $\rproj$ on Theorem~\ref{thm:SR}.}
      \label{fig:deriv}
    \end{figure}

  \item[Contextual rules]~ Let $t\lra[p] u$. Then,
    \begin{description}
    \item[($tv\lrap uv$)] Let $\vdash tv:A$. By Lemma~\ref{lem:generation}, one
      of the following cases happens:
      \begin{itemize}
      \item $\vdash t:\gB\Rightarrow B$ and $\vdash v:\gB$, with $B\preceq A$.
        Then, by the induction hypothesis, $\vdash u:\gB\Rightarrow B$. We
        conclude by rules $\Rightarrow_E$ and $\preceq$.
      \item $\vdash t:S(\gB\Rightarrow B)$ and $\vdash v:S(\gB)$, with
        $S(B)\preceq A$. Then, by the induction hypothesis, $\vdash
        u:S(\gB\Rightarrow B)$. We conclude by rules $\Rightarrow_{ES}$ and
        $\preceq$.
      \end{itemize}
    \item[($(\lambda x^{B}.v)t\lrap(\lambda x^{B}.v)u$)] Let $\vdash (\lambda
      x^{B}.v)t:A$. By Lemma~\ref{lem:generation}, one of the following cases
      happens:
      \begin{itemize}
      \item $\vdash (\lambda x^{B}.v):\gB\Rightarrow B$ and $\vdash t:\gB$, with
        $B\preceq A$. Then, by the induction hypothesis, $\vdash u:\gB$. We
        conclude by rules $\Rightarrow_E$ and $\preceq$.
      \item $\vdash (\lambda x^{B}.v):S(\gB\Rightarrow B)$ and $\vdash
        t:S(\gB)$, with $S(B)\preceq A$. Then, by the induction hypothesis,
        $\vdash u:S(\gB)$. We conclude by rules $\Rightarrow_{ES}$ and
        $\preceq$.
      \end{itemize}
    \item[($\pair tv\lrap\pair uv$)] Let $\vdash\pair tv:A$. By
      Lemma~\ref{lem:generation}, $\vdash t:B$ and $\vdash v:B$, with
      $S(B)\preceq A$. Then, by the induction hypothesis, $\vdash u:B$. We
      conclude by rules $S_I^+$ and $\preceq$.
    \item[($\alpha.t\lrap \alpha.u$)] Let $\vdash\alpha.t:A$. By
      Lemma~\ref{lem:generation}, $\vdash t:S(B)$, with $S(B)\preceq A$. Then,
      by the induction hypothesis, $\vdash u:S(B)$. We conclude by rules
      $S_I^\alpha$ and $\preceq$.
    \item[($\pi_jt\lrap \pi_ju$)] Let $\vdash\pi_jt:A$. By
      Lemma~\ref{lem:generation}, $\vdash t:S(\B^n)$, and $\B^j\times
      S(\B^{n-j})\preceq A$ Then, by the induction hypothesis, $\vdash
      u:S(\B^n)$. We conclude by rules $S_E$ and $\preceq$.
    \item[($t\times v\lrap u\times v$)] Let $\vdash t\times v:A$. By
      Lemma~\ref{lem:generation}, $\vdash t:B$ and $\vdash v:C$, with $B\times
      C\preceq A$. Then, by the induction hypothesis, $\vdash u:B$. We conclude
      by rules $\times_I$ and $\preceq$.
    \item[($v\times t\lrap v\times u$)] Analogous to previous case.
    \item[($\Uparrow_r t\lrap\Uparrow_r u$)] Let $\vdash\Uparrow_r t:A$. By
      Lemma~\ref{lem:generation}, $\vdash t:S(S(B)\times C)$, and $S(B\times
      C)\preceq A$. Then, by the induction hypothesis, $\vdash u:S(S(B)\times
      C)$, and so, by rule $\Uparrow_r$, $\vdash\Uparrow_r u:S(B\times C)$. We
      conclude by rule $\preceq$.
    \item[($\Uparrow_\ell t\lrap\Uparrow_\ell u$)] Let $\vdash\Uparrow_\ell t:A$. By
      Lemma~\ref{lem:generation}, $\vdash t:S(S(B)\times C)$, and $S(B\times
      C)\preceq A$. Then, by the induction hypothesis, $\vdash u:S(S(B)\times
      C)$, and so, by rule $\Uparrow_\ell$, $\vdash\Uparrow_\ell u:S(B\times C)$. We
      conclude by rule $\preceq$.
    \item[($\head\ t\lrap \head\ u$)] Let $\vdash\head\ t:A$. By
      Lemma~\ref{lem:generation}, $\vdash t:\B^n$, with $\B\preceq A$.
      Then, by the induction hypothesis, $\vdash u:\B^n$. We conclude by
      rules $\times_{Er}$ and $\preceq$.
    \item[($\tail\ t\lrap \tail\ u$)] Let $\vdash\tail\ t:A$. By
      Lemma~\ref{lem:generation}, $\vdash t:\B^n$, with $\B^{n-1}\preceq A$.
      Then, by the induction hypothesis, $\vdash u:\B^n$. We conclude by
      rules $\times_{Er}$ and $\preceq$.
      \qed
    \end{description}
  \end{description}
\end{proof}

\section{Detailed proofs of Section~\ref{sec:SN} (Strong normalisation)}\label{ap:SN}
\recap{Lemma}{lem:size}{
  If \( t \lra r \) by any of the rules in the groups linear distribution,
  vector space axioms or lists, or their contextual closure, then $\size r \geq \size t$. Moreover, $\size r
  = \size t$ if and only if the rule is $\rzeroS$.
}
\begin{proof}
  By induction and rule by rule analysis:
 \begin{description}
 \item[\rlinr:] \( {t}{(u + v)} \lra {t}{u} + {t}{v} \).
   \begin{align*}
     \size{{t}{(u + v)}}
     & = (3 \size{t} + 2)(3 \size{u + v} + 2) \\
     & = (3 \size{t} + 2)(3 (2 + \size{u} + \size{v}) + 2) \\
     & = (3 \size{t} + 2)(8 + 3 \size{u} + 3 \size{v}) \\
     & = 4 (3 \size{t} + 2) + (3 \size{t} + 2)(4 + 3 \size{u} + 3 \size{v}) \\
     & = 12 \size{t} + 8 + (3 \size{t} + 2)(4 + 3 \size{u} + 3 \size{v}) \\
     & = 12 \size{t} + 8 + (3 \size{t} + 2)((3 \size{u} + 2) + (3 \size{v} + 2)) \\
     & = 12 \size{t} + 8 + (3 \size{t} + 2)(3 \size{u} + 2) + (3 \size{t} + 2)(3 \size{v} + 2) \\
     & = 12 \size{t} + 8 + \size{{t}{u}} + \size{{t}{v}} \\
     & = 12 \size{t} + 6 + \size{{t}{u} + {t}{v}} \\
     & > \size{{t}{u} + {t}{v}}
   \end{align*}
 \item[\rlinscalr:] \( {t}{(\alpha . u)} \lra \alpha . {t}{u} \)
   \begin{align*}
     \size{{t}{(\alpha . u)}}
     & = (3 \size{t} + 2)(3 \size{\alpha . u} + 2) \\
     & = (3 \size{t} + 2)(3 (1 + 2 \size{u}) + 2) \\
     & = (3 \size{t} + 2)(6 \size{u} + 5) \\
     & = 3 \size{t} + 2 + (3 \size{t} + 2)(6 \size{u} + 4) \\
     & = 3 \size{t} + 2 + 2 (3 \size{t} + 2)(3 \size{u} + 2) \\
     & = 3 \size{t} + 2 + 2 \size{{t}{u}} \\
     & = 3 \size{t} + 1 + \size{\alpha . {t}{u}} \\
     & > \size{\alpha . {t}{u}}
   \end{align*}
 \item[\rlinzr:] \( {t}{\nullvec{\B}} \lra \nullvec{\mathsf{min}(A)} \)
   \begin{align*}
     \size{{t}{\nullvec{\B}}} & = (3 \size{t} + 2) (3 \size{\nullvec{\B}} +
                                2) > 0 = \size{\nullvec{\mathsf{min}(A)}}
   \end{align*}
 \item[\rlinl:] \( {(t + u)}{v} \lra {t}{v} + {u}{v} \)
   \begin{align*}
     \size{{(t + u)}{v}}
     & = (3 \size{t + u} + 2)(3 \size{v} + 2) \\
     & = (3 (2 + \size{t} + \size{u}) + 2)(3 \size{v} + 2) \\
     & = (3 \size{t} + 3 \size{u} + 8)(3 \size{v} + 2) \\
     & = (3 \size{t} + 3 \size{u} + 4)(3 \size{v} + 2) + 4 (3 \size{v} + 2) \\
     & = ((3 \size{t} + 3 \size{u} + 4)(3 \size{v} + 2) + 2) + 12 \size{v} + 6 \\
     & = \size{{t}{v} + {u}{v}} + 12 \size{v} + 6 \\
     & > \size{{t}{v} + {u}{v}}
   \end{align*}
 \item[\rlinscall:] \( {(\alpha . t)}{u} \lra \alpha . {t}{u} \)
   \begin{align*}
     \size{{(\alpha . t)}{u}}
     & = (3 \size{\alpha . t} + 2)(3 \size{u} + 2) \\
     & = (3 (1 + 2 \size{t}) + 2)(3 \size{u} + 2) \\
     & = (6 \size{t} + 5)(3 \size{u} + 2) \\
     & = (6 \size{t} + 4)(3 \size{u} + 2) + (3 \size{u} + 2) \\
     & = 2 (3 \size{t} + 2)(3 \size{u} + 2) + 3 \size{u} + 2 \\
     & = \size{\alpha . {t}{u}} + 3 \size{u} + 1 \\
     & > \size{\alpha . {t}{u}}
   \end{align*}
 \item[\rlinzl:] \( {\z[\B \Rightarrow A]}{t} \lra \z[\mathsf{min}(A)] \)
   \[
     \size{{\z[\B \Rightarrow A]}{t}} = (3 \size{\z[\B \Rightarrow A]} + 2)(3 \size{t} + 2) = 6 \size{t} + 4 > 0 = \size{\z[\mathsf{min}(A)]}
   \]
 \item[\rneut:] \( \z + t \lra t \)
   \[
     \size{\nullvec{A} + t} = 2 + \size{\nullvec{A}} + \size{t} = 2 + \size{t} > \size{t}
   \]
 \item[\runit:] \( 1 . t \lra t \)
   \[
     \size{1 . t} = 1 + 2 \size{t} > \size{t}
   \]
 \item[\rzeros:] \( 0 . t \lra \z[\mathsf{min}(A)]\)
   \[
     \size{0. t} = 1 + 2 \size{t} > 0 = \size{\z[\mathsf{min}(A)]}
   \]
 \item[\rzero:] \( \alpha . \nullvec{A} \lra \z[\mathsf{min}(A)] \)
   \[
     \size{\alpha . \nullvec{A}} = 1 + 2 \size{\nullvec{A}} = 1 > 0 = \size{\z[\mathsf{min}(A)]}
   \]
 \item[\rprod:] \( \alpha . (\beta . t) \lra (\alpha \times \beta) . t \)
   \begin{align*}
     \size{\alpha . (\beta . t)}
     & = 1 + 2 \size{\beta . t} \\
     & = 1 + 2 (1 + 2 \size{t}) \\
     & = 3 + 4 \size{t} \\
     & > 1 + 2 \size{t} \\
     & = \size{(\alpha \times \beta) . t}
   \end{align*}
 \item[\rdists:] \( \alpha . (t + u) \lra (\alpha . t + \alpha . u) \)
   \begin{align*}
     \size{\alpha . (t + u)}
     & = 1 + 2 \size{t + u} \\
     & = 5 + 2 \size{t} + 2 \size{u} \\
     & = 3 + \size{\alpha . t} + \size{\alpha . u} \\
     & = 1 + \size{\alpha . t + \alpha . u} \\
     & > \size{\alpha . t + \alpha . u}
   \end{align*}
 \item[\rfact:] \( \alpha . t + \beta . t \lra (\alpha + \beta) . t \)
   \begin{align*}
     \size{\alpha . t + \beta . t}
     & = 2 + \size{\alpha . t} + \size{\beta . t} \\
     & = 4 + 4 \size{t} \\
     & > 1 + 2 \size{t} \\
     & = \size{(\alpha + \beta) . t}
   \end{align*}
 \item[\rfacto:] \( \alpha . t + t \lra (\alpha + 1) . t \)
   \begin{align*}
     \size{\alpha . t + t}
     & = 2 + \size{\alpha . t} + \size{t} \\
     & = 3 + 3 \size{t} \\
     & > 1 + 2 \size{t} \\
     & = \size{(\alpha + 1) . t}
   \end{align*}
 \item[\rfactt:] \( t + t \lra 2 . t \)
   \[
     \size{t + t} = 2 + 2 \size{t} > 1 + 2 \size{t} = \size{2 . t}
   \]
 \item[\rzeroS:] \(\z\lra\z[\mathsf{min}(A)]\)
   \[
     \size{\z}=0=\size{\z[\mathsf{min}(A)]}
   \]

 \item[\rhead:] \( \head\ t \times r \lra t \)
   \[
     \size{\head\ t \times r} = 1 + \size{t \times r} = 1 + \size{t} + \size{r} > \size{t}
   \]
 \item[\rtail:] \( \tail\ {t \times r} \lra r \)
   \[
     \size{\tail\ {t \times r}} = 1 + \size{t \times r} = 1 + \size{t} + \size{r} > \size{r}
   \]
 \item[Contextual rules:] If $t\lra r$, then, by the induction hypothesis,
   $\size t\geq \size r$, and hence, $\size{C[t]}\geq \size{C[r]}$, where $C[\cdot]$ is
   a context with one hole.
    \qed
  \end{description}
\end{proof}

\recap{Lemma}{lem:ri_in_snset_implies_sum_ri_in_snset}{
  If for every \( i \in \{1, \ldots, n\} \) we have \( r_i \in \SN \),
  then \( \sum_{i = 1}^{n} \may r_i \in \SN \).
}
\begin{proof}
  Induction on the lexicographic order of \( (\sum_{i = 1}^{n} |r_i|,
  \size{\sum_{i = 1}^{n} \may r_i}) \) to show that \( \Red{\sum_{i =
      1}^{n}\may r_i} \subseteq \SN \). Let $t\in\Red{\sum_{i=1}^n\may r_i}$.
  The possibilities are:
  \begin{itemize}
  \item \( t = \sum_{i = 1}^n\may s_i \) where for all \( i \neq k \), \(
    s_i = r_i \) and \( r_k \lrap s_k \). Since \( \sum_{i = 1}^n
    \lpl{s_i} < \sum_{i = 1}^n \lpl{r_i} \), we conclude by the induction
    hypothesis.
  \item \( t = \sum_{i = 1}^n s_i \) where for all \( i \neq k \), \(
    s_i = \may r_i \) and \( \alpha_k . r_k \lra s_k \).
    Then, the reduction $\alpha_k . r_k\lra s_k$, is by one of the following
    rules: \runit, \rzeros, \rzero, \rprod, or \rdists.
    In any of these cases \( \sum_{i = 1}^{n} |s_i| \leq \sum_{i = 1}^{n} |r_i| \) and, by Lemma~\ref{lem:size}, \( \size{t} < \size{\sum_{i = 1}^{n}\may r_i} \). Hence, we conclude by the induction
    hypothesis.
  \item \( t = \sum_{\substack{i \neq j\\ i\neq k}}
    \may  r_i + ([\alpha_j] + [\alpha_k]) r_j \), where \( r_j = r_k \) (rule
    \rfact, \rfacto, or \rfactt). In this case \( (\sum_{i
      \neq j} |r_i|) + |r_j| \leq \sum_i |r_i| \) and by Lemma~\ref{lem:size},
    \( \size{t} < \size{\sum_{i = 1}^{n} \may r_i}\). Hence, we conclude by the
    induction hypothesis. \qed
  \end{itemize}
\end{proof}

\recap{Lemma}{lem:t_implies_proj_t}{
  If \( t \in \SN \), then \( \pi_{j}{t} \in \SN \).
}
\begin{proof}
  We show by induction on $|t|$ that $\Red{\pi_{j}{t}}\subseteq\SN$.
  Let $r\in\Red{\pi_jt}$. The possibilities are:
  \begin{itemize}
  \item \( r=\pi_{j}{t'} \) where \( t \lrap t' \). Since \( \lpl{t'} < \lpl{t} \)
    and \( t' \in \SN \), we conclude by the induction hypothesis.
  \item \( r = \prod_{h = 1}^j b_{hk} \times \sum_{i \in P}
    \left(\frac{\alpha_i}{\sqrt{\sum_{r \in P} |\alpha_r|^2}}\right)
    \prod_{h=j+1}^m b_{hi}\), $t=\sum_{i=1}^n\may\prod_{h=1}^m b_{hi}$.
    Any sequence starting on $r$ will only use vector space axioms rules, which,
    by Lemma~\ref{lem:size} reduce the size of the term, except for
    \rzeroS, which anyway can be used only a finite number of times. Therefore, $r\in\SN$.
    \qed
  \end{itemize}
\end{proof}

\recap{Lemma}{lem:cr}{
  For all \( A \), the following properties hold:
  \begin{description}
    \item[(CR1)] If \( t \in \interp{A} \), then \( t \in \SN \).
    \item[(CR2)] If \( t \in \interp{A} \), then \( \Red{t} \subseteq \interp{A} \).
    \item[(CR3)] If \( t \in \mathcal{N} \), $t$ has the same type as all the
      terms in $\interp A$, and \( \Red{t} \subseteq \interp{A} \) then \( t \in \interp{A} \).
    \item[(HAB)] For all \( x^A \), \( x \in \interp{A} \).
    \item[(LIN1)] If \( t \in \interp{A} \) and \( r \in \interp{A} \), then \( t + r \in \interp{A} \).
    \item[(LIN2)] If \( t \in \interp{A} \) then \( \alpha . t \in \interp{A} \).
    \item[(NULL)] \( \nullvec{A} \in \interp{A} \)
  \end{description}
}
\begin{proof}
  We proceed by induction over \( A \).
  \begin{itemize}
    \item Let \( A = \B \).
    \begin{description}
      \item[CR1]
       Let \( t \in \interp{\B} \). By definition, \( \interp{\B} \subseteq \SN \), so \( t \in \SN \).
      \item[CR2]
       Let \( t \in \interp{\B} \). By definition, \( \interp{\B} \subseteq
      \SN \), so \( t \in \SN \) and \( \Red{t} \subseteq \SN \). Furthermore,
      since \( t \in \interp{\B} \), we have \(  t : S(\B) \). Let \( r
      \in \Red{t} \), then \(  r : S(\B) \). Therefore, by definition, \( \Red{t} \subseteq \interp{\B} \).
      \item[CR3]
       Let \( t \in \mathcal{N} \) and $t:S(\B)$ where \( \Red{t} \subseteq \interp{\B} \).
      Since \( \Red{t} \subseteq \interp{\B} \subseteq \SN \), we have \( t \in
      \SN \).
      Therefore, by definition, \(
       t : \interp{\B} \).
      \item[HAB]
       Since \(  x^{\B} \in \SN \) and \(  x^{\B} : \B \preceq S(\B) \), we have, by definition, \( x^{\B} \in \interp{\B} \).
      \item[LIN1]
       Since \( t \in \interp{\B} \) and \( r \in \interp{\B} \), we have, by definition of \( \interp{\B} \), that \( t : S(\B) \), \( r : S(\B) \), \( t \in \SN \) and \( r \in \SN \). Then, \( t + r : S(S(\B)) \preceq S(\B) \) and, by Lemma~\ref{lem:ri_in_snset_implies_sum_ri_in_snset}, \( t + r \in \SN \). Therefore, by definition, \( t + r \in \interp{\B} \).
      \item[LIN2]
       Since \( t \in \interp{\B} \), we have, by definition of \( \interp{\B} \) that \( t : S(\B) \) and \( t \in \SN \). Then, \( \alpha . t : S(S(\B)) \preceq S(\B) \) and, by Lemma~\ref{lem:ri_in_snset_implies_sum_ri_in_snset}, \( \alpha . t \in \SN \). Therefore, by definition, \( \alpha . t \in \interp{\B} \).
      \item[NULL]
       Since \( \nullvec{\B} : S(\B) \) and \( \nullvec{\B} \in \SN \), we have by definition of \( \interp{\B} \) that \( \nullvec{\B} \in \interp{\B} \).
    \end{description}
    \item Let \( A = B \times C \)
    \begin{description}
      \item[CR1]
         Since \( t \in \interp{B \times C} \), we have by definition that \( t \in \SN \).
      \item[CR2]
         Since \( t \in \interp{B \times C} \), we have \( t : S(S(B) \times
        S(C)) \). Let \( t' \in \Red{t} \). Since $t'\in\Red t$,
          we have that \( t' : S(S(B) \times S(C)) \). On the other hand, since \( t \in \interp{B \times C} \), we have that \( t \in \SN \). Then, \( t' \in \SN \).
         Therefore, by definition, \( t' \in \interp{B \times C} \), which means \( \Red{t} \subseteq \interp{B \times C} \)
      \item[CR3]
         Let \( t \in \mathcal{N} \) and $t:S(S(B)\times S(C))$ where \( \Red{t}
        \subseteq \interp{B \times C} \).
        Since \( \interp{B \times C}\subseteq\SN \), we have \( t \in \SN \).
        Therefore, by definition, \( t \in \interp{B \times C} \).
      \item[HAB]
         Since \( x^{B \times C} : B \times C \preceq S(S(B) \times S(C)) \) and \( x^{B \times C} \in \SN \), we have by definition that \( x^{B \times C} \in \interp{B \times C} \).
      \item[LIN1]
         Since \( t \in \interp{B \times C} \), we have that \( t : S(S(B) \times S(C)) \) and \( t \in \SN \). Similarly, we have that \( r : S(S(B) \times S(C)) \) and \( r \in \SN \). Then, \( t + r : S(S(B) \times S(C)) \) and, by Lemma~\ref{lem:ri_in_snset_implies_sum_ri_in_snset}, \( t + r \in \SN \).
         Therefore, by definition, \( t + r \in \interp{B \times C} \).
      \item[LIN2]
         Since \( t \in \interp{B \times C} \), we have that \( t : S(S(B) \times S(C)) \) and \( t \in \SN \). Then, \( \alpha . t : S(S(S(B) \times S(C))) \preceq S(S(B) \times S(C)) \) and, by Lemma~\ref{lem:ri_in_snset_implies_sum_ri_in_snset}, \( \alpha . t \in \SN \).
         Therefore, by definition, \( \alpha . t \in \interp{B \times C} \).
      \item[NULL]
         Since \( \nullvec{B \times C} : S(B \times C) \preceq S(S(B) \times S(C)) \) and \( \nullvec{B \times C} \in \SN \), we have by definition that \( \nullvec{B \times C} \in \interp{B \times C} \).
    \end{description}
    \item Let \( A = \Psi \Rightarrow B \)
    \begin{description}
      \item[CR1]
         Given \( t \in \interp{\Psi \Rightarrow B} \), we want to show that \( t \in \SN \). Let \( r \in \interp{\Psi} \) (note that by induction hypothesis (HAB), such \( r \) exists). By definition, we have that \( \app{t}{r} \in \interp{B} \).
        And by induction hypothesis, we have that \( tr \in \SN \), which means, \( \lpl{tr} \) is finite.
        And since \( \lpl{t} \leq \lpl{\app{t}{r}} \), we have that \( \lpl{t} \) is finite, and therefore, \( t \in \SN \).
      \item[CR2]
         Given \( t \in \interp{\Psi \Rightarrow B} \), we want to show that \( \Red{t} \subseteq \interp{\Psi \Rightarrow B} \), which means that given \( t' \in \Red{t} \), \( t' \in \interp{\Psi \Rightarrow B} \). By definition of \( \interp{\Psi \Rightarrow B} \), this is the same as showing that \( t' : S(\Psi \Rightarrow B) \) and, for all \( r \in \interp{\Psi} \), \( \app{t'}{r} \in \interp{B} \).
         Since \( t \in \interp{\Psi \Rightarrow B} \), we have by definition
        that, for all \( r \in \interp{\Psi} \), \( \app{t}{r} \in \interp{B}
        \). And by induction hypothesis, this implies that, for all \( r \in
        \interp{\Psi} \), \( \Red{\app{t}{r}} \subseteq \interp{B} \). In
        particular, given \( t' \in \Red{t} \), we have that, for all \( r \in
        \interp{\Psi} \), \( \app{t'}{r} \in \Red{\app{t}{r}} \subseteq
        \interp{B} \). And since \( t \in \interp{\Psi \Rightarrow B} \), we
        have by definition that \( t : S(\Psi \Rightarrow B) \). Since
        $t'\in\Red t$, we have that \( t' : S(\Psi \Rightarrow B) \).
      \item[CR3]
         Given \( t \in \mathcal{N} \) and $t:S(\Psi\Rightarrow\B)$ where \(
        \Red{t} \subseteq \interp{\Psi \Rightarrow B} \), we want to show that
        \( t \in \interp{\Psi \Rightarrow B} \). By definition, this is the same
        than showing that for all \( r \in \interp{\Psi} \), \( \app{t}{r} \in
        \interp{B} \). By induction hypothesis, it suffices to show that for all
        \( r \in \interp{\Psi} \), \( \Red{\app{t}{r}} \in \interp{B} \). Notice
        that if $r\in\interp{\Psi}$, then $r:S(\Psi)$, and since
        $t:S(\Psi\Rightarrow B)$, we have $\app tr:S(B)$.
        Let \( r \in \interp{\Psi} \). By induction hypothesis (CR1), we have that \( r \in \SN \), which means \( \lpl{r} \) exists. Therefore, we can proceed by induction (2) over \( (\lpl{r}, \size{r}) \).
        We analyze the reducts of \( \app{t}{r} \):
        \begin{itemize}
          \item \( \app{t}{r} \reducesto \app{t'}{r} \) where \( t \reducesto t' \)
             Since \( t' \in \Red{t} \), we have that \( t' \in \interp{\Psi \Rightarrow B} \). And since \( r \in \interp{\Psi} \), we have by definition that \( \app{t'}{r} \in \interp{B} \).
          \item \( \app{t}{r} = \app{(\ite{}{u}{v})}{r} \reducesto \app{(\ite{}{u}{v})}{r'} \) where \( r \reducesto r' \)
             Since \( r \in \interp{\Psi} \), we have by induction hypothesis (CR2) that \( r' \in \interp{\Psi} \). And since \( \lpl{r'} < \lpl{r} \), we have by induction hypothesis (2) that \( \app{(\ite{}{u}{v})}{r'} \in \interp{B} \).
          \item \( \app{t}{r} = \app{t}{(r_1 + r_2)} \reducesto \app{t}{r_1} + \app{t}{r_2} \)
             Since \( \lpl{r_1} \leq \lpl{r} \) and, by Lemma~\ref{lem:size}, \( \size{r_1} < \size{r} \), we have by induction hypothesis (2) that \( \app{t}{r_1} \in \interp{B} \). Similarly, we have that \( \app{t}{r_2} \in \interp{B} \). Therefore, by Lemma~\ref{lem:cr} (LIN1), we have that \( \app{t}{r_1} + \app{t}{r_2} \in \interp{B} \).
          \item \( \app{t}{r} = \app{t}{(\alpha . r_1)} \reducesto \alpha . \app{t}{r_1} \)
             Since \( \lpl{r_1} \leq \lpl{r} \) and, by Lemma~\ref{lem:size}, \( \size{r_1} < \size{r} \), we have by induction hypothesis (2) that \( \app{t}{r_1} \in \interp{B} \). Therefore, by Lemma~\ref{lem:cr} (LIN2), we have that \( \alpha . \app{t}{r_1} \in \interp{B} \).
          \item \( \app{t}{r} = \app{t}{\nullvec{\Psi}} \reducesto \nullvec{\mathsf{min}(B)} \)
             By Lemma~\ref{lem:cr} (NULL), we have that \( \nullvec{\mathsf{min}(B)}
             \in \interp{\mathsf{min}(B)} \).
             By Lemma~\ref{lem:minlqA},
             $\mathsf{min}(B)\preceq B$, and $\z[\mathsf{min}(B)]\in\SN$, hence,
             by definition, $\z[\mathsf{min}(B)]\in\interp B$.
        \end{itemize}
      \item[HAB]
         By definition of \( \interp{\Psi \Rightarrow B} \), and since \( x^{\Psi \Rightarrow B} : S(\Psi \Rightarrow B) \), it suffices to show that, for all \( t \in \interp{\Psi} \), we have that \( \app{x^{\Psi \Rightarrow B}}{t} \in \interp{B} \).
         Let \( t \in \interp{\Psi} \). Since \( \app{x^{\Psi \Rightarrow B}}{t} \in \mathcal{N} \), it suffices to show that \( \Red{\app{x^{\Psi \Rightarrow B}}{t}} \subseteq \interp{B} \). Since \( t \in \interp{\Psi} \), we have by induction hypothesis (CR1) that \( t \in \SN \). Therefore, we can proceed by induction (2) over \( (\lpl{t}, \size{t}) \). We analyze the possible reducts of \( \app{x^{\Psi \Rightarrow B}}{t} \):
        \begin{itemize}
          \item \( \app{x^{\Psi \Rightarrow B}}{t} = \app{x^{\Psi \Rightarrow B}}{(t_1 + t_2)} \reducesto \app{x^{\Psi \Rightarrow B}}{t_1} + \app{x^{\Psi \Rightarrow B}}{t_2} \)
             Since \( \lpl{t_1} \leq \lpl{t} \) and, by Lemma~\ref{lem:size}, \( \size{t_1} < \size{t} \), we have by induction hypothesis (2) that \( \app{x^{\Psi \Rightarrow B}}{t_1} \in \interp{B} \). Similarly, we have that \( \app{x^{\Psi \Rightarrow B}}{t_2} \in \interp{B} \). Therefore, by induction hypothesis (LIN1), we have that \( \app{x^{\Psi \Rightarrow B}}{t_1} + \app{x^{\Psi \Rightarrow B}}{t_2} \in \interp{B} \).
          \item \( \app{x^{\Psi \Rightarrow B}}{t} = \app{x^{\Psi \Rightarrow B}}{(\alpha . t_1)} \reducesto \alpha . \app{x^{\Psi \Rightarrow B}}{t_1} \)
             Since \( \lpl{t_1} \leq \lpl{t} \) and, by Lemma~\ref{lem:size}, \( \size{t_1} < \size{t} \), we have by induction hypothesis (2) that \( \app{x^{\Psi \Rightarrow B}}{t_1} \in \interp{B} \). Therefore, by induction hypothesis (LIN2), we have that \( \alpha . \app{x^{\Psi \Rightarrow B}}{t_1} \in \interp{B} \).
	   \item \( \app{x^{\Psi \Rightarrow B}}{t} = \app{x^{\Psi \Rightarrow B}}{(\nullvec{\Psi})} \reducesto \nullvec{\mathsf{min}(B)} \).
             By induction hypothesis (NULL), we have that \( \nullvec{\mathsf{min}(B)} \in \interp{\mathsf{min}(B)} \).
             By Lemma~\ref{lem:minlqA},
             $\mathsf{min}(B)\preceq B$, and $\z[\mathsf{min}(B)]\in\SN$, hence,
             by definition, $\z[\mathsf{min}(B)]\in\interp B$.
        \end{itemize}
      \item[LIN1]
         By definition of \( \interp{\Psi \Rightarrow B} \), it suffices to show that \( t + r : S(\Psi \Rightarrow B) \) and, for all \( s \in \interp{\Psi} \), \( \app{(t + r)}{s} \in \interp{B} \).
         Since \( t \in \interp{\Psi \Rightarrow B} \) and \( r \in \interp{\Psi \Rightarrow B} \), we have that \( t : S(\Psi \Rightarrow B) \), \( r : S(\Psi \Rightarrow B) \) and, for all \( s \in \interp{\Psi} \), \( \app{t}{s} \in \interp{B} \) and \( \app{r}{s} \in \interp{B} \). Therefore, \( t + r : S(S(\Psi \Rightarrow B)) \preceq S(\Psi \Rightarrow B) \) and \( \app{(t + r)}{s} : S(B) \).
         It remains to show that, for all \( s \in \interp{\Psi} \), \( \app{(t + r)}{s} \in \interp{B} \). Since \( \app{(t + r)}{s} \in \mathcal{N} \) and, by type derivation, \( \app{(t + r)}{s} : S(B) \), we have by induction hypothesis (CR3) that it is sufficient to show that, for all \( s \in \interp{\Psi} \), \( \Red{\app{(t + r)}{s}} \subseteq \interp{B} \). Since \( \app{t}{s} \in \interp{B} \), \( \app{r}{s} \in \interp{B} \), y \( s \in \interp{\Psi} \), we have by induction hypothesis (CR1) that \( t \in \SN \), \( r \in \SN \) y \( s \in \SN \). Therefore, we can proceed by induction (2) over \( (\lpl{t} + \lpl{r} + \lpl{s}, \size{\app{(t + r)}{s}}) \). We analyze the possible reducts of \( \app{(t + r)}{s} \):
        \begin{itemize}
          \item \( \app{(t + r)}{s} \reducesto \app{(t' + r)}{s} \) where \( t \reducesto t' \)
           Since \( \lpl{t'} < \lpl{t} \), we have by induction hypothesis that \( \app{(t' + r)}{s} \in \interp{B} \).
          \item \( \app{(t + r)}{s} \reducesto \app{(t + r')}{s} \) where \( r \reducesto r' \)
           Analogous to the previous case.
          \item \( \app{(t + r)}{(s_1 + s_2)} \reducesto \app{(t + r)}{s_1} + \app{(t + r)}{s_2} \)
           Since \( \lpl{s_1} \leq \lpl{s} \) and \( \size{\app{(t + r)}{s_1}} < \size{\app{(t + r)}{s}} \), we have by induction hypothesis (2) that \( \app{(t + r)}{s_1} \in \interp{B} \). Similarly, we have that \( \app{(t + r)}{s_2} \in \interp{B} \). Therefore, by induction hypothesis, we have that \( \app{(t + r)}{s_1} + \app{(t + r)}{s_2} \in \interp{B} \).
          \item \( \app{(t + r)}{(\alpha . s_1)} \reducesto \alpha . \app{(t + r)}{s_1} \)
           Since \( \lpl{s_1} \leq \lpl{s} \) and \( \size{\app{(t + r)}{s_1}} < \size{\app{(t + r)}{s}} \), we have by induction hypothesis (2) that \( \app{(t + r)}{s_1} \in \interp{B} \). Therefore, by induction hypothesis (LIN2), we have that \( \alpha . \app{(t + r)}{s_1} \in \interp{B} \).
          \item \( \app{(t + r)}{\nullvec{\Psi}} \reducesto \nullvec{\mathsf{min}(B)} \)
             By induction hypothesis (NULL), we have that \( \nullvec{\mathsf{min}(B)} \in \interp{\mathsf{min}(B)} \).
             By Lemma~\ref{lem:minlqA},
             $\mathsf{min}(B)\preceq B$, and $\z[\mathsf{min}(B)]\in\SN$, hence,
             by definition, $\z[\mathsf{min}(B)]\in\interp B$.
          \item \( \app{(t + r)}{s} \reducesto \app{t}{s} + \app{r}{s} \)
           Since \( \app{t}{s} \in \interp{B} \) and \( \app{r}{s} \in \interp{B} \), we have by induction hypothesis that \( \app{t}{s} + \app{r}{s} \in \interp{B} \).
        \end{itemize}
      \item[LIN2]
         By definition of \( \interp{\Psi \Rightarrow B} \), it suffices to show that \( \alpha . t : S(\Psi \Rightarrow B) \) and, for all \( s \in \interp{\Psi} \), \( \app{(\alpha . t)}{s} \in \interp{B} \).
         Since \( t \in \interp{\Psi \Rightarrow B} \), we have that \( t : S(\Psi \Rightarrow B) \) and, for all \( s \in \interp{\Psi} \), \( \app{t}{s} \in \interp{B} \). Therefore, \( \alpha . t : S(S(\Psi \Rightarrow B)) \preceq S(\Psi \Rightarrow B) \) and \( \app{(\alpha . t)}{s} : S(B) \).
         It remains to show that, for all \( s \in \interp{\Psi} \), \( \app{(\alpha . t)}{s} \in \interp{B} \). Since \( \app{(\alpha . t)}{s} \in \mathcal{N} \), we have by induction hypothesis (CR3) that it is sufficient to show that, for all \( s \in \interp{\Psi} \), \( \Red{\app{(\alpha . t)}{s}} \subseteq \interp{B} \). Since \( \app{t}{s} \in \interp{B} \) and \( s \in \interp{\Psi} \), we have by induction hypothesis (CR1) that \( t \in \SN \) and \( s \in \SN \). Therefore, we can proceed by induction (2) over \( (\lpl{t} + \lpl{s}, \size{\app{(\alpha . t)}{s}}) \). We analyze the possible reducts of \( \app{(\alpha . t)}{s} \):
        \begin{itemize}
          \item \( \app{(\alpha . t)}{s} \reducesto \app{(\alpha . t')}{s} \) where \( t \reducesto t' \)
           Since \( \lpl{t'} < \lpl{t} \), we have by induction hypothesis (2) that \( \app{(\alpha . t')}{s} \in \interp{B} \).
          \item \( \app{(\alpha . t)}{s} = \app{(\alpha . t)}{(s_1 + s_2)} \reducesto \app{(\alpha . t)}{s_1} + \app{(\alpha . t)}{s_2} \)
           Since \( \lpl{t} + \lpl{s_1} \leq \lpl{t} + \lpl{s} \) and \( \size{\app{(\alpha . t)}{s_1}} < \size{\app{(\alpha . t)}{s}} \), we have by induction hypothesis that \( \app{(\alpha . t)}{s_1} \in \interp{B} \). Similarly, \( \app{(\alpha . t)}{s_2} \in \interp{B} \). Therefore, by induction hypothesis (LIN1), \( \app{(\alpha . t)}{s_1} + \app{(\alpha . t)}{s_2} \in \interp{B} \).
          \item \( \app{(\alpha . t)}{s} = \app{(\alpha . t)}{(\beta . s_1)} \reducesto \beta . \app{(\alpha. t)}{s_1} \)
           Since \( \lpl{t} + \lpl{s_1} \leq \lpl{t} + \lpl{s} \) and \( \size{\app{(\alpha . t)}{s_1}} < \size{\app{(\alpha . t)}{s}} \), we have by induction hypothesis (2) that \( \app{(\alpha . t)}{s_1} \in \interp{B} \). Therefore, by induction hypothesis, \( \beta . \app{(\alpha . t)}{s_1} \in \interp{B} \).
          \item \( \app{(\alpha . t)}{s} = \app{(\alpha . t)}{\nullvec{\Psi}} \reducesto \nullvec{\mathsf{min}(B)} \)
             By induction hypothesis (NULL), we have that \( \nullvec{\mathsf{min}(B)} \in \interp{\mathsf{min}(B)} \).
             By Lemma~\ref{lem:minlqA},
             $\mathsf{min}(B)\preceq B$, and $\z[\mathsf{min}(B)]\in\SN$, hence,
             by definition, $\z[\mathsf{min}(B)]\in\interp B$.
          \item \( \app{(\alpha . t)}{s} \reducesto \alpha . \app{t}{s} \)
           Since \( \app{t}{s} \in \interp{B} \), we have by induction hypothesis that \( \alpha . \app{t}{s} \in \interp{B} \).
        \end{itemize}
      \item[NULL]
         We want to show that \( \nullvec{\Psi \Rightarrow B} \in \interp{\Psi
\Rightarrow B} \). By definition of \( \interp{\Psi \Rightarrow B} \), this is
equivalent to showing that, for all \( t \in \interp{\Psi} \), \(
\app{\nullvec{\Psi \Rightarrow B}}{t} \in \interp{B} \). Since \(
\app{\nullvec{\Psi \Rightarrow B}}{t} \in \mathcal{N} \) and \(
\app{\nullvec{\Psi \Rightarrow B}}{t} : S(B) \), we have by induction hypothesis
(CR3) that this is equivalent to showing that \( \Red{\app{\nullvec{\Psi
\Rightarrow B}}{t}} \subseteq \interp{B} \). Since the only possible reduct of
\( \app{\nullvec{\Psi \Rightarrow B}}{t} \) is \( \nullvec{\mathsf{min}(B)} \), it suffices to
show that \( \nullvec{\mathsf{min}(B)} \in \interp{B} \).
             By induction hypothesis, we have that \( \nullvec{\mathsf{min}(B)} \in \interp{\mathsf{min}(B)} \).
             By Lemma~\ref{lem:minlqA},
             $\mathsf{min}(B)\preceq B$, and $\z[\mathsf{min}(B)]\in\SN$, hence,
             by definition, $\z[\mathsf{min}(B)]\in\interp B$.
    \end{description}
    \item Let \( A = S(B) \)
    \begin{description}
      \item[CR1]
         Since \( t \in \interp{S(B)} \), we have by definition that \( t \in \SN \).
      \item[CR2]
         Given \( t \in \interp{S(B)} \), we want to show that \( \Red{t} \subseteq \interp{S(B)} \).
        By definition, \( t \in \SN \), and so, \( \Red{t} \subseteq \SN \). Therefore, \( \Red{t} \subseteq \interp{S(B)} \).
      \item[CR3]
         Given \( t \in \mathcal{N} \) where $t:S(B)$ and \( \Red{t} \subseteq \interp{S(B)} \), we want to show that \( t \in \interp{S(B)} \).
        By definition, \( \Red{t}\subseteq\interp{S(B)} \subseteq \SN \), and so, \( t \in \SN \). Therefore, by definition, \( t \in \interp{S(B)} \).
      \item[HAB]
       Since \( x^{S(B)} : S(B) \) and \( x^{S(B)} \in \SN \), we have by definition that \( x^{S(B)} \in \interp{S(B)} \).
      \item[LIN1]
         Since \( t \in \interp{S(B)} \) and \( r \in \interp{S(B)} \), we have by definition of \( \interp{S(B)} \) that \( t : S(B) \), \( r : S(B) \), \( t \in \SN \) and \( r \in \SN \). Therefore, \( t + r : S(S(B)) \preceq S(B) \) and, by Lemma~\ref{lem:ri_in_snset_implies_sum_ri_in_snset}, \( t + r \in \SN \). Therefore, by definition, \( t + r \in \interp{S(B)} \).
      \item[LIN2]
         Since \( t \in \interp{S(B)} \), we have by definition of \( \interp{S(B)} \) that \( t : S(B) \) and \( t \in \SN \). Then, \( \alpha . t : S(S(B)) \preceq S(B) \) and \( \alpha . t \in \SN \). Therefore, by definition, \( \alpha . t \in \interp{S(B)} \).
      \item[NULL]
         We want to show \( \nullvec{S(B)} \in \interp{S(B)} \). Since \(
         \nullvec{S(B)} : S(S(B)) \preceq S(B) \) and \( \nullvec{S(B)} \in \SN
         \), we have by definition that \( \nullvec{S(B)} \in \interp{S(B)} \).
         \qed
    \end{description}
  \end{itemize}
\end{proof}

\recap{Lemma}{lem:a_subset_b}{
  If \( A \preceq B \) then \( \interp{A} \subseteq \interp{B} \).
}
\begin{proof}
  We proceed by induction on the relation $\preceq$.
  \begin{itemize}
  \item \( \vcenter{\infer{A \preceq A}{}} \). Trivial by the reflexivity of set inclusion.
  \item \( \vcenter{\infer{A \preceq C}{A \preceq B & B \preceq C}} \). Trivial by the transitivity of set inclusion.
  \item \( \vcenter{\infer{A \preceq S(A)}{}} \). Let \( t \in \interp{A} \). By definition, \( t : S(A) \preceq S(S(A)) \). And by Lemma~\ref{lem:cr} (CR1), \( t \in \SN \). Therefore, by definition, \( t \in \interp{S(A)} \).
  \item \( \vcenter{\infer{S(S(A)) \preceq S(A)}{}} \).
    Let \( t \in \interp{S(S(A))}\). By definition, \( t : S(S(A)) \preceq S(A) \) and \( t \in \SN \). Therefore, by definition, \( t \in \interp{S(A)} \).
  \item \( \vcenter{\infer{\Psi \Rightarrow A \preceq \Psi \Rightarrow B}{A
        \preceq B}} \).
    Let \( t\in\interp{\Psi\Rightarrow A} \). By definition, \( t : S(\Psi \Rightarrow A) \) and, for all \( r \in \interp{\Psi} \), \( \app{t}{r} \in \interp{A} \). By induction hypothesis, \( \app{t}{r} \in \interp{A} \subseteq \interp{B} \). And since \( A \preceq B \), \( t : S(\Psi \Rightarrow A) \preceq S(\Psi \Rightarrow B) \). Therefore, by definition, \( t \in \interp{\Psi \Rightarrow B} \).
  \item \( \vcenter{\infer{S(A) \preceq S(B)}{A \preceq B}} \). Let \( t \in \interp{S(A)} \). By definition, \( t : S(A) \) and \( t \in \SN \). And since \( A \preceq B \), we have \( t : S(A) \preceq S(B) \). Therefore, by definition, \( t \in \interp{S(B)} \).
  \item \( \vcenter{\infer{A \times C \preceq B \times C}{A \preceq B}} \). Let \( t \in \interp{A \times C} \). By definition, \( t \in \SN \) and \( t : S(S(A) \times S(C)) \). And since \( A \preceq B \), we have that \( S(A) \preceq S(B) \), and so, \( t : S(S(B) \times S(C)) \). Therefore, by definition, \( t \in \interp{B \times C} \).
  \item \( \vcenter{\infer{A \times C \preceq B \times C}{A \preceq B}} \). Analogous to the previous case.
  \qed
  \end{itemize}
\end{proof}

\xrecap{Lemma}{Adequacy}{lem:adequacy}{
  If \( \Gamma \vdash t:A \) and \( \theta \vDash \Gamma \) then \( \theta(t)
  \in \interp{A} \).
}
\begin{proof}
  By induction in the derivation of $\Gamma\vdash t:A$.
  We proceed by cases.
  \begin{itemize}
  \item \( \vcenter{\infer[\tax]{x^\Psi \vdash x : \Psi}{}} \). Since
    $\theta\vDash x^\Psi$, we have $\theta(x)\in\interp{\Psi}$.

  \item \(\vcenter{\infer[\tax_{\vec 0}]{\vdash \nullvec{A}: S(A)}{}}\). By Lemma~\ref{lem:cr} (NULL) and Lemma~\ref{lem:a_subset_b}, $\theta(\z)=\z\in\interp{S(A)}$.

  \item \(\vcenter{\infer[\tax_{\ket 0}]{\vdash \ket{0}: \B}{}}\). By definition, \( \theta(\ket{0}) = \ket{0} \in \SN \). And since \( \ket{0} : S(\B) \), we have by definition that \( \ket{0} \in \interp{\B} \).

  \item \(\vcenter{\infer[\tax_{\ket 1}]{\vdash \ket{1}: \B}{}}\). Analogous to the previous case.

  \item \( \vcenter{\infer[S_I^\alpha]{\Gamma \vdash \alpha . t : S(A)}{\Gamma \vdash
        t: A}} \). By the induction hypothesis, \( \theta(t)\in\interp A \), and by Lemma~\ref{lem:cr} (LIN2), \( \alpha . \theta(t) = \theta(\alpha . t) \in \interp{A} \). Finally, by Lemma~\ref{lem:a_subset_b}, \( \theta(\alpha . t) \in \interp{S(A)} \).

  \item \( \vcenter{\infer[S_I +]{\Gamma, \Delta \vdash (t + u): S(A)}{\Gamma \vdash t : A & \Delta
        \vdash u : A}} \).
    By the induction hypothesis, $\theta_1(t)\in\interp A$ and
    $\theta_2(u)\in\interp A$, with $\theta_1\vDash\Gamma$ and
    $\theta_2\vDash\Delta$. Then, since \( \Gamma \) and \( \Delta \) are disjoint, $\theta_1\cup\theta_2(t+u)=\theta_1(t)+\theta_2(u)$. And by Lemma~\ref{lem:cr} (LIN1), \( \theta_1(t)+\theta_2(u) \in \interp{A} \). Therefore, by Lemma~\ref{lem:a_subset_b}, \( \theta_1(t)+\theta_2(u)\in\interp{S(A)} \).

  \item \( \vcenter{\infer[S_E]{\Gamma \vdash \pi_{j}{t} : \B^j \times S(\B^{n - j})}{\Gamma
        \vdash t : S(\B^n)}} \).
    We want to show that if \( \theta \models \Gamma \), then \( \theta(\proj{j}{t}) = \proj{j}{\theta(t)} \in \interp{\B^j \times S(\B^{n - j})} \). By definition, it suffices to show that \( \proj{j}{\theta(t)} : S(S(\B^j) \times S(S(\B^{n - j})) \) and \( \proj{j}{\theta(t)} \in \SN \).
    By induction hypothesis, \( \theta(t) \in \interp{S(\B^n)} \), which implies that \( \theta(t) : S(S(\B^n)) \preceq S(\B^n) \). Therefore, \( \proj{j}{\theta(t)} : \B^j \times S(\B^{n - j}) \preceq S(S(\B^j) \times S(S(\B^{n - j}))) \).
    On the other hand, since \( \theta(t) \in \interp{S(\B^n)} \subseteq \SN \), we have by Lemma~\ref{lem:t_implies_proj_t} that \( \proj{j}{\theta(t)} \in \SN \).
    Therefore, \( \proj{j}{\theta(t)} \in \interp{\B^j \times S(\B^{n - j})} \).

  \item \( \vcenter{\infer[\preceq]{\Gamma \vdash t : B}{\Gamma \vdash t: A & A
        \preceq B}} \).
    By the induction hypothesis $\theta(t)\in\interp A$, and by
    Lemma~\ref{lem:a_subset_b}, $\interp A\subseteq\interp B$.

  \item \( \vcenter{\infer[\tif]{\Gamma \vdash \ite{}{t}{r}: \B \Rightarrow A}{\Gamma \vdash t : A
        & \Gamma \vdash r : A}} \).
    We want to show that if \( \theta \models \Gamma \), then \( \theta(\ite{}{t}{r}) = \ite{}{\theta(t)}{\theta(r)} \in \interp{\B \Rightarrow A} \). By definition, this is equivalent to showing that \( \ite{}{\theta(t)}{\theta(r)} : S(\B \Rightarrow A) \) and, for all \( s \in \interp{\B}, \ite{s}{\theta(t)}{\theta(r)} \in \interp{A} \).
    \\ By induction hypothesis, we have that \( \theta(t) \in \interp{A} \) and \( \theta(r) \in \interp{A} \), which implies that \( \theta(t) : S(A) \) and \( \theta(r) : S(A) \). Therefore, \( \ite{}{\theta(t)}{\theta(r)} : \B \Rightarrow S(A) \preceq S(\B \Rightarrow A) \) and \( \ite{s}{\theta(t)}{\theta(r)} : S(A) \).
    \\ And since \( \ite{s}{\theta(t)}{\theta(r)} \in \mathcal{N} \) (because such a term is actually an application), we have by Lemma~\ref{lem:cr} (CR3) that it suffices to show that \( \Red{\ite{s}{\theta(t)}{\theta(r)}} \subseteq \interp{A} \).
    \\ We proceed by induction (2) over \( (\lpl{s}, \size{\ite{s}{\theta(t)}{\theta(r)}}) \). We analyze each of the reducts of \( \ite{s}{\theta(t)}{\theta(r)} \):
    \begin{itemize}
    \item \( \ite{s}{\theta(t)}{\theta(r)} \reducesto \ite{u}{\theta(t)}{\theta(r)} \) where \( s \reducesto u \)
      \\ Since \( \lpl{u} < \lpl{s} \), we have by induction hypothesis (2) that \( \ite{u}{\theta(t)}{\theta(r)} \in \interp{A} \).
    \item \( \ite{s}{\theta(t)}{\theta(r)} \reducesto \theta(t) \) where \( s = \ket{1} \)
      \\ By induction hypothesis, \( \theta(t) \in \interp{A} \).
    \item \( \ite{s}{\theta(t)}{\theta(r)} \reducesto \theta(r) \) where \( s = \ket{0} \)
      \\ By induction hypothesis, \( \theta(r) \in \interp{A} \).
    \item \( \ite{s}{\theta(t)}{\theta(r)} = \ite{(s_1 + s_2)}{\theta(t)}{\theta(r)} \reducesto \ite{s_1}{\theta(t)}{\theta(r)} + \ite{s_2}{\theta(t)}{\theta(r)} \)
      \\ Since \( \lpl{s_1} \leq \lpl{s} \) and \( \size{\ite{s_1}{\theta(t)}{\theta(r)}} < \size{\ite{s}{\theta(t)}{\theta(r)}} \), we have by induction hypothesis (2) that \( \ite{s_1}{\theta(t)}{\theta(r)} \in \interp{A} \). Similarly, \( \ite{s_2}{\theta(t)}{\theta(r)} \in \interp{A} \). Therefore, by Lemma~\ref{lem:cr} (LIN1), \( \ite{s_1}{\theta(t)}{\theta(r)} + \ite{s_2}{\theta(t)}{\theta(r)} \in \interp{A} \).
    \item \( \ite{s}{\theta(t)}{\theta(r)} = \ite{(\alpha . s_1)}{\theta(t)}{\theta(r)} \reducesto \alpha . \ite{s_1}{\theta(t)}{\theta(r)} \)
      \\ Since \( \lpl{s_1} \leq \lpl{s} \) and \( \size{\ite{s_1}{\theta(t)}{\theta(r)}} < \size{\ite{s}{\theta(t)}{\theta(r)}} \), we have by induction hypothesis (2) that \( \ite{s_1}{\theta(t)}{\theta(r)} \in \interp{A} \). Therefore, by Lemma~\ref{lem:cr} (LIN2), \( \alpha . \ite{s_1}{\theta(t)}{\theta(r)} \in \interp{A} \).
    \item \( \ite{s}{\theta(t)}{\theta(r)} = \ite{\nullvec{\B}}{\theta(t)}{\theta(r)} \reducesto \nullvec{A} \)
      \\ By Lemma~\ref{lem:cr} (NULL), \( \nullvec{A} \in \interp{A} \).
    \end{itemize}

  \item \( \vcenter{\infer[\Rightarrow_I]{\Gamma \vdash \lambda x^\Psi.t: \Psi
        \Rightarrow A}{\Gamma, x : \Psi \vdash t : A}} \).
    We want to show that if \( \theta' \models \Gamma \), then \( \theta'(\abstr{\vrbl{x}{\Psi}}{t}) = (\abstr{\vrbl{x}{\Psi}}{\theta'(t)}) \in \interp{\Psi \Rightarrow A} \), which is equivalent to showing that \( (\abstr{\vrbl{x}{\Psi}}{\theta'(t)}) : S(\Psi \Rightarrow A) \) and, for all \( r \in \interp{\Psi}, \app{(\abstr{\vrbl{x}{\Psi}}{\theta'(t)})}{r} \in \interp{A} \).
    \\ By induction hypothesis, \( \theta(t) \in \interp{A} \), which implies that \( \theta(t) : S(A) \). Therefore, \( (\abstr{\vrbl{x}{\Psi}}{\theta'(t)}) : \Psi \Rightarrow S(A) \preceq S(\Psi \Rightarrow A) \) and \( \app{(\abstr{\vrbl{x}{\Psi}}{\theta'(t)})}{r} : S(A) \). And since \( \app{(\abstr{\vrbl{x}{\Psi}}{\theta'(t)})}{r} \in \mathcal{N} \), Lemma~\ref{lem:cr} (CR3) tells us that if \( \Red{\app{(\abstr{\vrbl{x}{\Psi}}{\theta'(t)})}{r}} \subseteq \interp{A} \), then \( \app{(\abstr{\vrbl{x}{\Psi}}{\theta'(t)})}{r} \in \interp{A} \).
    We are going to show that in fact \( \Red{\app{(\abstr{\vrbl{x}{\Psi}}{\theta'(t)})}{r}} \subseteq \interp{A} \).
    \\ Since \( r \in \interp{\Psi} \), we have by Lemma~\ref{lem:cr} (CR2) that \( r \in \SN \). Therefore, we can proceed by induction (2) over \( (\lpl{r}, \size{\app{(\abstr{\vrbl{x}{\Psi}}{\theta'(t)})}{r}} \). We analyze each of the reducts of \( \app{(\abstr{\vrbl{x}{\Psi}}{\theta'(t)})}{r} \):
    \begin{itemize}
    \item \( \app{(\abstr{\vrbl{x}{\Psi}}{\theta'(t)})}{r} \reducesto \theta'(t)[r/x] \)
      \\ We want to show that \( \theta'(t)[r/x] \in \interp{A} \). By definition, \( \theta'(t)[r/x] = \theta(t) \), and by induction hypothesis, \( \theta(t) \in \interp{A} \). Therefore, \( \theta'(t)[r/x] \in \interp{A} \).
    \item \( \app{(\abstr{\vrbl{x}{\Psi}}{\theta'(t)})}{r} \reducesto \app{(\abstr{\vrbl{x}{\Psi}}{\theta'(t)})}{r'} \) where \( r \reducesto r' \)
      \\ By induction hypothesis (2), \( \app{(\abstr{\vrbl{x}{\Psi}}{\theta'(t)})}{r'} \in \interp{A} \).
    \item \( \app{(\abstr{\vrbl{x}{\Psi}}{\theta'(t)})}{r} = \app{(\abstr{\vrbl{x}{\Psi}}{\theta'(t)})}{(r_1 + r_2)} \reducesto \app{(\abstr{\vrbl{x}{\Psi}}{\theta'(t)})}{r_1} + \app{(\abstr{\vrbl{x}{\Psi}}{\theta'(t)})}{r_2} \)
      \\ Since \( \lpl{r_1} \leq \lpl{r} \) y \( \size{\app{(\abstr{\vrbl{x}{\Psi}}{\theta'(t)})}{r_1}} < \size{\app{(\abstr{\vrbl{x}{\Psi}}{\theta'(t)})}{r}} \), we have by induction hypothesis (2) that \( \app{(\abstr{\vrbl{x}{\Psi}}{\theta'(t)})}{r_1} \in \interp{A} \). Similarly, we have that \( \app{(\abstr{\vrbl{x}{\Psi}}{\theta'(t)})}{r_2} \in \interp{A} \). Therefore, by Lemma~\ref{lem:cr} (LIN1), we have that \( \app{(\abstr{\vrbl{x}{\Psi}}{\theta'(t)})}{r_1} + \app{(\abstr{\vrbl{x}{\Psi}}{\theta'(t)})}{r_2} \in \interp{A} \).
    \item \( \app{(\abstr{\vrbl{x}{\Psi}}{\theta'(t)})}{r} = \app{(\abstr{\vrbl{x}{\Psi}}{\theta'(t)})}{(\alpha . r_1)} \reducesto \alpha . \app{(\abstr{\vrbl{x}{\Psi}}{\theta'(t)})}{r_1} \)
      \\ Since \( \lpl{r_1} \leq \lpl{r} \) y \( \size{\app{(\abstr{\vrbl{x}{\Psi}}{\theta'(t)})}{r_1}} < \size{\app{(\abstr{\vrbl{x}{\Psi}}{\theta'(t)})}{r}} \), we have by induction hypothesis (2) that \( \app{(\abstr{\vrbl{x}{\Psi}}{\theta'(t)})}{r_1} \in \interp{A} \). Therefore, by Lemma~\ref{lem:cr} (LIN2), we have that \( \alpha . \app{(\abstr{\vrbl{x}{\Psi}}{\theta'(t)})}{r_1} \in \interp{A} \).
    \item \( \app{(\abstr{\vrbl{x}{\Psi}}{\theta'(t)})}{r} = \app{(\abstr{\vrbl{x}{\Psi}}{\theta'(t)})}{\nullvec{\B^n}} \reducesto \nullvec{A} \)
      \\ By Lemma~\ref{lem:cr} (NULL), we have that \( \nullvec{A} \in \interp{A} \).
    \end{itemize}

  \item \( \vcenter{\infer[\Rightarrow_E]{\Gamma, \Delta \vdash {t}{u} :
        A}{\Gamma \vdash t : \Psi \Rightarrow A & \Delta \vdash u : \Psi}} \) We must show that if \( \theta \models \Gamma, \Delta \), then \( \theta(\app{t}{u}) \in \interp{A} \).
    Since \( \Gamma \) and \( \Delta \) are disjoint, we have that \( \theta(\app{t}{u}) = (\theta_1 \cup \theta_2)(\app{t}{u}) = \app{\theta_1(t)}{\theta_2(u)} \), where \( \theta_1 \models \Gamma \) and \( \theta_2 \models \Delta \).
    Therefore, it suffices to show that \( \app{\theta_1(t)}{\theta_2(u)} \in \interp{A} \).
    By induction hypothesis and definition of \( \interp{\Psi \Rightarrow A} \), we have that \( \app{\theta_1(t)}{\theta_2(u)} \in \interp{A} \).

  \item \( \vcenter{\infer[\Rightarrow_{ES}]{\Gamma, \Delta \vdash t u :
        S(A)}{\Gamma \vdash t : S(\Psi \Rightarrow A) & \Delta \vdash u :
        S(\Psi)}} \). We want to show that if \( \theta \models \Gamma, \Delta \), then \( \theta(\app{t}{u}) \in \interp{S(A)} \).
    Since \( \Gamma \) y \( \Delta \) are disjoint, we have that \( \theta(\app{t}{u}) = (\theta_1 \cup \theta_2)(\app{t}{u}) = \app{\theta_1(t)}{\theta_2(u)} \), where \( \theta_1 \models \Gamma \) y \( \theta_2 \models \Delta \).
    Therefore, it suffices to show that \( \app{\theta_1(t)}{\theta_2(u)} \in \interp{S(A)} \).
    \\ Since \( \theta_1(t) \in \interp{S(\Psi \Rightarrow A)} \) and \( \theta_2(u) \in \interp{S(\Psi)} \), we have by definition that \( \theta_1(t) : S(\Psi \Rightarrow A) \) y \( \theta_2(u) : S(\Psi) \). Therefore, \( \app{\theta_1(t)}{\theta_2(u)} : S(A) \).
    \\ On the other hand, we need to show that \( \app{\theta_1(t)}{\theta_2(u)} \in \SN \). To do that, it is sufficient to show that \( \Red{\app{\theta_1(t)}{\theta_2(u)}} \subseteq \SN \). Since \( \theta_1(t) \in \interp{S(\Psi \Rightarrow A)} \subseteq \SN \) and \( \theta_2(u) \in \interp{S(\Psi)} \subseteq \SN \), we can proceed by induction (2) over \( (\lpl{\theta_1(t)} + \lpl{\theta_2(u)}, \size{\app{\theta_1(t)}{\theta_2(u)}} \). We analyze the possible reducts of \( \app{\theta_1(t)}{\theta_2(u)} \):
    \begin{itemize}
    \item \( \app{\theta_1(t)}{\theta_2(u)} \reducesto \app{t'}{\theta_2(u)} \) where \( \theta_1(t) \reducesto t' \)
      \\ Since \( \lpl{t'} < \lpl{\theta(t)} \), we have by induction hypothesis (2) that \( \app{t'}{\theta_2(u)} \in \SN \).
    \item \( \app{\theta_1(t)}{\theta_2(u)} \reducesto \app{\theta_1(t)}{u'} \) where \( \theta_2(u) \reducesto u' \)
      \\ Analogous to the previous case.
    \item \( \app{\theta_1(t)}{\theta_2(u)} = \app{(\abstr{\vrbl{x}{\B^n}}{t_1})}{\theta_2(u)} \reducesto t_1[\theta_2(u)/x] \)
      \\ Since \( \theta_1(t) = (\abstr{\vrbl{x}{\B^n}}{t_1}) : S(\Psi \Rightarrow A) \), we have by Lemma~\ref{lem:generation}  that \( t_1 : A \) with a smaller derivation tree. Then, by induction hypothesis (Adequacy), we have that \( t_1 \in \interp{A} \), and therefore, by Lemma~\ref{lem:cr} (CR1), we have that \( t_1 \in \SN \).
    \item \( \app{\theta_1(t)}{\theta_2(u)} = \app{(\abstr{\vrbl{x}{S(\Psi)}}{t_1})}{\theta_2(u)} \reducesto t_1[\theta_2(u)/x] \)
      \\ Analogous to the previous case.
    \item \( \app{\theta_1(t)}{\theta_2(u)} = \ite{\ket{1}}{t_1}{t_2} \reducesto t_1 \)
      \\ Since \( \theta_1(t) = (\ite{}{t_1}{t_2}) : S(\Psi \Rightarrow A) \), we have by Lemma~\ref{lem:generation}  that \( t_1 : A \) with a smaller derivation tree. Then, by induction hypothesis (Adequacy), we have that \( t_1 \in \interp{A} \), and therefore, by Lemma~\ref{lem:cr} (CR1), we have that \( t_1 \in \SN \).
    \item \( \app{\theta_1(t)}{\theta_2(u)} = \ite{\ket{0}}{t_1}{t_2} \reducesto t_2 \)
      \\ Analogous to the previous case.
    \item \( \app{\theta_1(t)}{\theta_2(u)} = \app{\theta_1(t)}{(u_1 + u_2)} \reducesto \app{\theta_1(t)}{u_1} + \app{\theta_1(t)}{u_2} \)
      \\ Since \( \lpl{u_1} \leq \lpl{\theta_2(u)} \) and, by Lemma~\ref{lem:size}, \( \size{\app{\theta_1(t)}{u_1}} < \size{\app{\theta_1(t)}{\theta_2(u)}} \), we have by induction hypothesis (2) that \( \app{\theta_1(t)}{u_1} \in \SN \). Similarly, we have that \( \app{\theta_1(t)}{u_2} \in \SN \). Therefore, by Lemma~\ref{lem:ri_in_snset_implies_sum_ri_in_snset}, we have that \( \app{\theta_1(t)}{u_1} + \app{\theta_1(t)}{u_2} \in \SN \).
    \item \( \app{\theta_1(t)}{\theta_2(u)} = \app{(t_1 + t_2)}{\theta_2(u)} \reducesto \app{t_1}{\theta_2(u)} + \app{t_2}{\theta_2(u)} \)
      \\ Analogous to the previous case.
    \item \( \app{\theta_1(t)}{\theta_2(u)} = \app{\theta_1(t)}{(\alpha . u_1)} \reducesto \alpha . \app{\theta_1(t)}{u_1} \)
      \\ Since \( \lpl{u_1} \leq \lpl{\theta_2(u)} \) and, by Lemma~\ref{lem:size}, \( \size{\app{\theta_1{t}}{u_1}} < \size{\app{\theta_1{t}}{\theta_2(u)}} \), we have by induction hypothesis (2) that \( \app{\theta_1(t)}{u_1} \in \SN \). Therefore, by Lemma~\ref{lem:ri_in_snset_implies_sum_ri_in_snset}, we have that \( \alpha . \app{\theta_1(t)}{u_1} \in \SN \).
    \item \( \app{\theta_1(t)}{\theta_2(u)} = \app{(\alpha . t_1)}{\theta_2(u)} \reducesto \alpha . \app{t_1}{\theta_2(u)} \)
      \\ Analogous to the previous case..
    \item \( \app{\theta_1(t)}{\theta_2(u)} = \app{\theta_1(t)}{\nullvec{\B^n}} \reducesto \nullvec{A} \)
      \\ By definition, \( \nullvec{A} \in \interp{S(A)} \subseteq \SN \).
    \item \( \app{\theta_1(t)}{\theta_2(u)} = \app{\nullvec{\B^n \Rightarrow A}}{\theta_2(u)} \reducesto \nullvec{A} \)
      \\ Analogous to the previous case.
    \end{itemize}

  \item \( \vcenter{\infer[W]{\Gamma, x^{\B^n} \vdash t : A}{\Gamma \vdash t: A}}
    \). By definition of \( \theta \), we have that if \( \theta \models \Gamma, \vrbl{x}{\B^n} \), then \( \theta \models \Gamma \). And by induction hypothesis, \( \theta(t) \in \interp{A} \).

  \item $\vcenter{\infer[C]{\Gamma, x : \B^n \vdash (x/y) t : A}{\Gamma, x : \B^n,
        y : \B^n \vdash t : A}}$. By definition, \( \theta'(t[x/y]) = \theta(t) \). And by induction hypothesis, \( \theta(t) \in \interp{A} \). Therefore, \( \theta'(t[x/y]) \in \interp{A} \).

  \item $\vcenter{\infer[\times_I]{\Gamma, \Delta \vdash t \times u : A \times B}{\Gamma \vdash t : A & \Delta \vdash u : B}}$.
    Since \( \Gamma \) and \( \Delta \) are disjoint, \( \theta(t \times u) = (\theta_1 \cup \theta_2)(t \times u) = \theta_1(t) \times \theta_2(u) \), where \( \theta_1 \models \Gamma \) and \( \theta_2 \models \Delta \).
    Therefore, it suffices to show that \( \theta_1(t) \times \theta_2(u) \in \interp{A \times B} \).
    \\ Since \( \theta_1(t) \in \interp{A} \), we have by definition that \( \theta_1(t) : S(A) \) and, by Lemma~\ref{lem:cr} (CR1), \( \theta_1(t) \in \SN \). Similarly, \( \theta_2(u) : S(B) \) and \( \theta_2(u) \in \SN \).
    \\ Then, \( \theta_1(t) \times \theta_2(u) : S(A) \times S(B) \preceq S(S(A) \times S(B)) \) and \( \theta_1(t) \times \theta_2(u) \in \SN \). Therefore, by definition of \( \interp{A \times B} \), \( \theta_1(t) \times \theta_2(u) \in \interp{A \times B} \).

  \item $\vcenter{\infer[\times_{Er}]{\Gamma \vdash \head\ t : \B}{\Gamma \vdash
        t : \B^n}}$.
    \\ Since \( \head\ t : \B \), we have by Lemma~\ref{lem:substitution} that \( \theta(\head\ t) = \head\ \theta(t) : \B \preceq S(\B) \). And by induction hypothesis, \( \theta(t) \in \interp{\B^n} \subseteq \SN \). Then, \( \head\ \theta(t) \in \SN \). Therefore, by definition, \( \head\ \theta(t) \in \interp{\B} \).

  \item $\vcenter{\infer[\times_{El}]{\Gamma \vdash \tail\ t : \B^{n-1}}{\Gamma \vdash
        t : \B^n}}$.
    \\ Since \( \tail\ t : \B^n \), we have by Lemma~\ref{lem:substitution} that \( \theta(\tail\ t) = \tail\ \theta(t) : \B^n = \B \times \B^{n - 1} \preceq S(S(\B) \times S(\B^{n - 1})) \). Furthermore, by induction hypothesis, \( \theta(t) \in \interp{\B^n} \), and by Lemma~\ref{lem:cr} (CR1), \( \theta(t) \in \SN \). Then, \( \tail\ \theta(t) \in \SN \). Therefore, by definition, \( \tail\ \theta(t) \in \interp{\B} \).

  \item $\vcenter{\infer[\Uparrow_r]{\Gamma \vdash \Uparrow_r t : S(A
        \times B)}{\Gamma \vdash t : S(S(A) \times B)}}$.
    \\ Por definición de \( \interp{S(S(A) \times B)} \), tenemos que \( \theta(t) : S(S(A) \times B) \) y \( \theta(t) \in \SN \). Entonces, \( \Uparrow_r \theta(t) : S(A \times B) \) y \( \Uparrow_r \theta(t) \in \SN \). Por lo tanto, \( \Uparrow_r \theta(t) \in \interp{S(A \times B)} \), que es lo que queríamos mostrar.

  \item $\vcenter{\infer[\Uparrow_\ell]{\Gamma \vdash \Uparrow_\ell t : S(A
        \times B)}{\Gamma \vdash t : S(A \times S(B))}}$.
    Analogous to the previous case. \qed
  \end{itemize}
\end{proof}

\section{Detailed proofs of Section~\ref{sec:denSem} (Interpretation)}\label{ap:denSem}
\recap{Lemma}{lem:inc}{If $A\preceq B$, then $\den A\subseteq\den B$.}
\begin{proof}
  We proceed by induction on the relation $\preceq$.
  \begin{itemize}
  \item $\den{A} \subseteq \den{A}$ by the reflexivity of set inclusion.
  \item Let $A\preceq B$ and $B\preceq C$. By the induction hypothesis, $\den A\subseteq \den B$ and $\den B\subseteq \den C$. Then $\den A\subseteq \den C$ by the transitivity of set inclusion.
  \item $\den A\subseteq\gen\den A=\den{S(A)}$.
  \item $\den{S(S(A))}=\gen(\gen\den A)=\gen\den A=\den{S(A)}$.
  \item Let $A\preceq B$ and $\den A\subseteq \den B$. Then,
    \begin{itemize}
    \item $\begin{aligned}[t] \den{\gB\Rightarrow A}
        &=\den{\gB}\Rightarrow\den A\\
        &=\{f~|~\forall a\in\den{\gB}, fa\in\den A\}\\
        &\subseteq\{f~|~\forall a\in\den{\gB},fa\in\den B\}\\
        &=\den{\gB}\Rightarrow\den B\\
        &=\den{\gB\Rightarrow B}
      \end{aligned}$
    \item $\begin{aligned}[t] \den{A\times C}
        &=\den A\times\den C\\
        &=\{a\times c~|~a\in\den A,c\in\den C\}\\
        &\subseteq\{b\times c~|~b\in\den B,c\in\den C\}\\
        &=\den B\times\den C\\
        &=\den{B\times C}
      \end{aligned}$
    \item $\den{C\times A}\subseteq\den{C\times B}$ by an analogous reasoning
      to the previous one.
    \item $\den{S(A)}=\gen\den A\subseteq\gen\den B=\den{S(B)}$.
      \qed
    \end{itemize}
  \end{itemize}
\end{proof}

\recap{Lemma}{lem:subsDen}{
  If $\Gamma\vdash t:A$ and $\phi,x\mapsto S,y\mapsto S$ is a
  $\Gamma$-valuation, then $\den t_{\phi,x\mapsto S,y\mapsto
    S}=\den{(x/y)t}_{\phi,x\mapsto S}$.
}
\begin{proof}
  We proceed by induction on $t$.
  \begin{description}
  \item[Independent cases.] The cases where $t$ does not includes $x$ nor $y$
    and the denotation does not depends on the valuation, are trivial. Those
    cases are: $\ket 0$, $\ket 1$, $\z$ and $\ite{}{}{}$.
  \item[Let $t=x$.] Then, $\den x_{\phi,x\mapsto S,y\mapsto S}=S=\den
    x_{\phi,x\mapsto S}$.
  \item[Let $t=y$.] Then, $\den y_{\phi,x\mapsto S,y\mapsto S}=S=\den
    x_{\phi,x\mapsto S}$.
  \item[Let $t=z$.] Then, $\den z_{\phi,x\mapsto S,y\mapsto S}=\phi z=\den
    z_{\phi,x\mapsto S}$.
  \item[Let $t=\lambda z^{\gB}.r$.] Then,
    \begin{align*}
      \den{\lambda z^{\gB}.r}_{\phi,x\mapsto S,y\mapsto S}
      &=\{f~|~\forall a\in\den{\gB},fa\in\den r_{\phi,x\mapsto S,y\mapsto S,z\mapsto\den{\gB}}\}\\
      \textrm{(by IH)} &=\{f~|~\forall a\in\den{\gB},fa\in\den{(x/y)r}_{\phi,x\mapsto S,z\mapsto\den{\gB}}\}\\
      &=\den{\lambda z^{\gB}.(x/y)r}_{\phi,x\mapsto S}\\
      &=\den{(x/y)\lambda z^{\gB}.r}_{\phi,x\mapsto S}
    \end{align*}
  \item[Let $t=r\times s$.] Then,
    \begin{align*}
      \den{r\times s}_{\phi,x\mapsto S,y\mapsto S}
      &=\den r_{\phi,x\mapsto S,y\mapsto S}\times\den s_{\phi,x\mapsto S,y\mapsto S}\\
      \textrm{(by IH) }&=\den{(x/y)r}_{\phi,x\mapsto S}\times\den{(x/y)s}_{\phi,x\mapsto S}\\
      &=\den{(x/y)r\times(x/y)s}_{\phi,x\mapsto S}\\
      &=\den{(x/y)(r\times s)}_{\phi,x\mapsto S}
    \end{align*}
  \item[Let $t=\pair rs$.] Then,
    \begin{align*}
      \den{\pair rs}_{\phi,x\mapsto S,y\mapsto S}
      &=\{a+b~|~a\in\den r_{\phi,x\mapsto S,y\mapsto S},b\in\den s_{\phi,x\mapsto S,y\mapsto S}\}\\
      \textrm{(by IH) }&=\{a+b~|~a\in\den{(x/y)r}_{\phi,x\mapsto S},b\in\den{(x/y)s}_{\phi,x\mapsto S}\}\\
      &=\den{\pair{(x/y)r}{(x/y)s}}_{\phi,x\mapsto S}\\
      &=\den{(x/y)\pair rs}_{\phi,x\mapsto S}
    \end{align*}
  \item[Let $t=\alpha.r$.] Then,
    \begin{align*}
      \den{\alpha.r}_{\phi,x\mapsto S,y\mapsto S}
      &=\{\alpha.a~|~a\in\den r_{\phi,x\mapsto S,y\mapsto S}\}\\
      \textrm{(by IH) }&=\{\alpha.a~|~a\in\den{(x/y)r}_{\phi,x\mapsto S}\}\\
      &=\den{\alpha.(x/y)r}_{\phi,x\mapsto S}\\
      &=\den{(x/y)\alpha.r}_{\phi,x\mapsto S}
    \end{align*}
  \item[Let $t=rs$.] Then,
    \begin{align*}
      &\den{rs}_{\phi,x\mapsto S,y\mapsto S}\\
      &=
        \left\{\begin{array}{l}
                 \{\sum\limits_{i\in I}\alpha_i . g_i(a)~|~\sum\limits_{i\in I}\alpha_i . g_i\in\den r_{\phi,x\mapsto S,y\mapsto S},a\in\den s_{\phi,x\mapsto S,y\mapsto S}\}\\
                 \hspace{8cm}\textrm{ If } A=\gB\Rightarrow B\\
                 \{\sum\limits_{i\in I}\sum\limits_{j\in J}\alpha_i . \beta_j . g_i(c_j)~|~\sum\limits_{i\in I}\alpha_i . g_i\in\den r_{\phi,x\mapsto S,y\mapsto S},\sum\limits_{j\in J}\beta_j . c_j\in\den s_{\phi,x\mapsto S,y\mapsto S}\} \\
                 \hspace{8cm}\textrm{ If } A=S(\gB\Rightarrow B)
               \end{array}\right.\\
      &\textrm{(by IH)}\\
      &=
        \left\{\begin{array}{l}
                 \{\sum\limits_{i\in I}\alpha_i . g_i(a)~|~\sum\limits_{i\in I}\alpha_i . g_i\in\den{(x/y)r}_{\phi,x\mapsto S},a\in\den{(x/y)s}_{\phi,x\mapsto S}\}\\
                 \hspace{8cm}\textrm{ If } A=\gB\Rightarrow B\\
                 \{\sum\limits_{i\in I}\sum\limits_{j\in J}\alpha_i . \beta_j . g_i(c_j)~|~\sum\limits_{i\in I}\alpha_i . g_i\in\den{(x/y)r}_{\phi,x\mapsto S},\sum\limits_{j\in J}\beta_j . c_j\in\den{(x/y)s}_{\phi,x\mapsto S}\} \\
                 \hspace{8cm}\textrm{ If } A=S(\gB\Rightarrow B)
               \end{array}\right.\\
      &=\den{(x/y)r(x/y)s}_{\phi,x\mapsto S}\\
      &=\den{(x/y)(rs)}_{\phi,x\mapsto S}
    \end{align*}
  \item[Let $t=\pi_jr$.] Then,
    \[
      \den{\pi_jr}_{\phi,x\mapsto S,y\mapsto S}\! =\! \{
      \prod_{h=1}^jb_{hk}\times \sum_{i\in P}\! \left(\!
        \dfrac{\alpha_i}{\sqrt{\sum_{r\in P}|\alpha_r|^2}}
        \!\right)\prod_{h=j+1}^m b_{hi}~|~ \forall i\in P,\forall
      h,b_{hi}=b_{hk} \}
    \]
    with $\den r_{\phi,x\mapsto S,y\mapsto S} = \{
    \sum_{i=1}^n\may\prod_{h=1}^mb_{hi} \}$ where
    $b_{hi}=\vect 01$ or $\vect 10$.

    By the induction hypothesis, $\den{(x/y)r}_{\phi,x\mapsto S}=\den
    r_{\phi,x\mapsto S,y\mapsto S}$, hence,
    \[
      \den{(x/y)(\pi_jr)}_{\phi,x\mapsto S}=\den{\pi_j((x/y)r)}_{\phi,x\mapsto
        S}=\den{\pi_jr}_{\phi,x\mapsto S,y\mapsto S}
    \]
  \item[Let $t=\head~r$.] Then,
    \begin{align*}
      \den{\head~r}_{\phi,x\mapsto S,y\mapsto S}
      &=\{a_1~|~\prod_{i=1}^n a_i\in\den t_{\phi,x\mapsto S,y\mapsto S}, a_1\in\den{\B}\}\\
      \textrm{(by IH) }&=
                         \{a_1~|~\prod_{i=1}^n\in\den{(x/y)t}_{\phi,x\mapsto S}, a_1\in\den{\B}\}\\
      &=\den{\head~(x/y)r}_{\phi,x\mapsto S}\\
      &=\den{(x/y)(\head~r)}_{\phi,x\mapsto S}
    \end{align*}
  \item[Let $t=\tail~r$.] Then,
    \begin{align*}
      \den{\tail~r}_{\phi,x\mapsto S,y\mapsto S}
      &=\{\prod_{i=2}^n a_i~|~\prod_{i=1}^n a_i\in\den t_{\phi,x\mapsto S,y\mapsto S}, a_1\in\den{\B}\}\\
      \textrm{(by IH) }&=
                         \{\prod_{i=2}^n a_i~|~\prod_{i=1}^n a_i\in\den{(x/y)t}_{\phi,x\mapsto S}, a_1\in\den{\B}\}\\
      &=\den{\tail~(x/y)r}_{\phi,x\mapsto S}\\
      &=\den{(x/y)(\tail~r)}_{\phi,x\mapsto S}
    \end{align*}
  \item[Let $t=\Uparrow r$] Then,
    \begin{align*}
      \den{\Uparrow r}_{\phi,x\mapsto S,y\mapsto S}
      &=\den r_{\phi,x\mapsto S,y\mapsto S}\\
      \textrm{(by IH) }&=\den{(x/y)r}_{\phi,x\mapsto S}\\
      &=\den{\Uparrow (x/y)r}_{\phi,x\mapsto S}\\
      &=\den{(x/y)\Uparrow r}_{\phi,x\mapsto S}
    \end{align*} \qed
  \end{description}
\end{proof}

\recap{Theorem}{thm:soundness}{
  If $\Gamma\vdash t:A$, and $\phi$ is a $\Gamma$-valuation. Then $\den
  t_\phi\subseteq\den A$.
}
\begin{proof}
  We proceed by induction on the typing derivation.
  \begin{itemize}
  \item $\vcenter{\infer[\tax]{x:\gB\vdash x:\gB}{}}$.
    Then, $\den x_\phi=\phi x=\den{\gB}$.
  \item $\vcenter{\infer[\tax_{\vec 0}]{{\vdash \z:S(A)}}{}}$.
    Then, $\den{\z}_\phi = \{\vec 0\}\subset\gen\den A=\den{S(A)}$.
  \item $\vcenter{\infer[\tax_{\ket 0}]{{\vdash\ket 0:\B}}{}}$.
    Then, $\den{\ket 0}_\phi=\{\vect 10\}\subset\{\vect 10,\vect 01\}=\den{\B}$.
  \item $\vcenter{\infer[\tax_{\ket 1}]{{\vdash\ket 1:\B}}{}}$.
    Then, $\den{\ket 1}_\phi=\{\vect 01\}\subset\{\vect 10,\vect 01\}=\den{\B}$.
  \item $\vcenter{\infer[\preceq]{{\Gamma\vdash t:B}}{\Gamma\vdash t:A &
        A\preceq B}}$.

    Then, by the induction hypothesis, $\den t_\phi\subseteq\den A$ and by
    Lemma~\ref{lem:inc}, $\den A\subseteq\den B$, hence, $\den
    t_\phi\subseteq\den B$.
  \item $\vcenter{\infer[S_I^\alpha]{{\Gamma\vdash \alpha.t:S(A)}}{{\Gamma\vdash t:A}}}$.

    Then, by the induction hypothesis, $\den t_\phi\subseteq\den A$.
    Hence, $\den{\alpha.t}_\phi=\{\alpha.a~|~a\in\den
    t_\phi\}\subseteq\{\alpha.a~|~a\in\den A\}\subseteq\gen\den A=\den{S(A)}$.
  \item $\vcenter{\infer[\tif]{\Gamma\vdash\ite{}tr:\B\Rightarrow A}{\Gamma\vdash t:A
        & \Gamma\vdash r:A}}$.

    Then, since by the induction hypothesis, $\den t_\phi\in\den A$ and $\den
    r_\phi\in\den A$, we have
    \begin{align*}
      \den{\ite{}tr}_\phi & =\{f~|~\forall a\in\den{\B}, fa=\left\{\begin{array}{ll}
                                                                     \den t_\phi & \textrm{If }a=\vect 01\\
                                                                     \den r_\phi & \textrm{If }a=\vect 10
                                                                   \end{array}\right.\}\\
                          &\subset\den{\B}\Rightarrow\den{A}=\den{\B\Rightarrow A}.
    \end{align*}
  \item $\vcenter{\infer[\Rightarrow_I]{{\Gamma\vdash\lambda x^{\gB}.t:\gB\Rightarrow A}}{{\Gamma,x:\gB\vdash t:A}}}$.

    Then, by the induction hypothesis, $\den
    t_{\phi,x\mapsto\den{\gB}}\subseteq\den A$. Hence,
    \begin{align*}
      \den{\lambda x^{\gB}.t}_\phi
      & =\{f~|~\forall a\in\den{\gB}, fa\in\den t_{\phi,x\mapsto\den{\gB}}\}\\
      &\subseteq\{f~|~\forall a\in\den{\gB}, fa\in\den A\}\\
      &=\den{\gB}\Rightarrow\den A=\den{\gB\Rightarrow A}.
    \end{align*}
  \item $\vcenter{\infer[\Rightarrow_E]{{\Gamma,\Delta\vdash
          tu:A}}{{\Gamma\vdash t:\gB\Rightarrow A} & {\Delta\vdash u:\gB}}}$.

    Then, by the induction hypothesis $\den
    t_{\phi_\Gamma}\subseteq\den{\gB}\Rightarrow\den A$ and $\den
    u_{\phi_\Delta}\subseteq\den{\gB}$, where $\phi=\phi_\Gamma,\phi_\Delta$.
    Then,
    \begin{align*}
      \den{tu}_\phi
      &= \{\sum_{i\in I}\alpha_i . g_i(a)~|~\sum_{i\in I}\alpha_i . g_i\in\den t_\phi\textrm{ and }a\in\den u_\phi\}\\
      &\subseteq
        \{\sum_{i\in I}\alpha_i . g_i(a)~|~\sum_{i\in I}\alpha_i . g_i\in\den{\gB}\Rightarrow\den A\textrm{ and }a\in\den{\gB}\}.
    \end{align*}

    Since $\den{\gB}\Rightarrow\den A$ is a set of functions (and not a linear
    combination of them), $I$ is a singleton and so this set is equal to
    $\{fa~|~f\in\den{\gB}\Rightarrow\den A\textrm{ and
    }a\in\den{\gB}\}\subseteq\den A$.
  \item $\vcenter{\infer[\Rightarrow_{ES}]{{\Gamma,\Delta\vdash
          tu:S(A)}}{{\Gamma\vdash t:S(\gB\Rightarrow A)} & {\Delta\vdash u:S(\gB)}}}$.

    Then, by the induction hypothesis $\den
    t_{\phi_\Gamma}\subseteq\gen(\den{\gB}\Rightarrow\den A)$ and $\den
    u_{\phi_\Delta}\subseteq\gen\den{\gB}$. Then,

    \begin{align*}
      \den{tu}_\phi &= \{\sum_{i\in I}\sum_{j\in J}\alpha_i . \beta_j g_i(c_j)~|~\sum_{i\in I}\alpha_i . g_i\in\den t_\phi\textrm{ and }\sum_{j\in J}\beta_j . c_j\in\den u_\phi\}\\
                    & \subseteq \{\sum_{i\in I}\sum_{j\in J}\alpha_i . \beta_j . g_i(c_j)~|~\sum_{i\in I}\alpha_i . g_i\in\gen(\den{\gB}\Rightarrow\den A)\textrm{ and }\sum_{j\in J}\beta_j . c_j\in\gen\den{\gB}\}\\
                    &= \{\sum_{i\in I}\sum_{j\in J}\alpha_i . \beta_j . g_i(c_j)~|~g_i\in\den{\gB}\Rightarrow\den A\textrm{ and }c_j\in\den{\gB}\}\\
                    &\subseteq \gen{\den A}=\den{S(A)}
    \end{align*}
  \item $\vcenter{\infer[S_I^+]{{\Gamma,\Delta\vdash\pair
          tu:S(A)}}{{\Gamma\vdash t:A} & {\Delta\vdash u:A}}}$.

    Then, by the induction hypothesis $\den t_{\phi_\Gamma}\subseteq\den A$ and
    $\den u_{\phi_\Delta}\subseteq\den A$, with $\phi=\phi_\Gamma,\phi_\Delta$.
    Then,
    $\den{\pair tu}_\phi=\{a+b~|~a\in\den t_\phi\textrm{ and }b\in\den
    u_\phi\}\subseteq\{a+b~|~a,b\in\den A\}\subseteq\gen\den A=\den{S(A)}$.
  \item $\vcenter{\infer[S_E]{{\Gamma\vdash\pi_j t:\B^j\times
          S(\B^{n-j})}}{{\Gamma\vdash t:S(\B^n)}}}$.

    Then, by the induction hypothesis, $\den
    t_\phi\subseteq\den{S(\B^n)}=\mathbb C^{2^n}$. By definition,
    \[
      \den{\pi_jt}_\phi=\{\prod_{h=1}^jb_{hk}\times\sum_{i\in
        P}\left(\frac{\alpha_i}{\sqrt{\sum_{r\in
              P}|\alpha_r|^2}}\right)\prod_{h=j+1}^nb_{hi}~|~\forall i\in
      P,\forall h,b_{hi}=b_{hk}\}
    \]
    where $\den t_\phi=\{\sum_{i=1}^n \may\prod_{h=1}^n b_{hi}\}$, with
    $b_{hi}\in\den{\B}$, and $\forall i\in P, \forall h, b_{hi}=b_{hk}$. Then,
    $\den{\pi_jt}_\phi \subseteq\den{\B}^n\times\den{\B^{n-j}}
    =\den{\B^n\times\B^{n-j}}$.
  \item $\vcenter{\infer[W]{{\Gamma,x:\B^n\vdash t:A}}{{\Gamma\vdash t:A}}}$.

    Then, by the induction hypothesis, $\den t_\phi\subseteq\den A$, where
    $\phi$ is a $\Gamma$-valuation. Notice that any $\phi'$ that is a
    $(\Gamma,x:\B^n)$-valuation is also a $\Gamma$-valuation. Then, $\den
    t_{\phi'}\subseteq\den A$.
  \item $\vcenter{\infer[C]{{\Gamma,x:\B^n\vdash (x/y)t:A}}{{\Gamma,x:\B^n,y:\B^n\vdash t:A}}}$.

    Then, by the induction hypothesis, $\den t_\phi\subseteq\den A$, where
    $\phi$ is a $(\Gamma,x:\B^n,y:\B^n)$-valuation. Let $\phi'$ be a
    $(\Gamma,x:\B^n)$-valuation, then, by Lemma~\ref{lem:subsDen},
    $\den{(x/y)t}_{\phi'}\subseteq\den A$.
  \item $\vcenter{\infer[\times_I]{{\Gamma,\Delta\vdash t\times u:A\times
          B}}{\Gamma\vdash t:A & \Delta\vdash u:B}}$.

    Then, by the induction hypothesis, $\den t_{\phi_\Gamma}\subseteq\den A$ and
    $\den u_{\phi_\Delta}\subseteq\den B$. Then,
    $\den{t,u}_{\phi_\Gamma,\phi_\Delta}
    =\den t_{\phi_\Gamma,\phi_\Delta}\times\den u_{\phi_\Gamma,\phi_\Delta}
    =\den t_{\phi_\Gamma}\times\den u_{\phi_\Delta}
    \subseteq \den A\times\den B
    =\den{A\times B}$.
  \item $\vcenter{\infer[\times_{Er}]{{\Gamma\vdash \head~t:\B}}{{\Gamma\vdash t:\B^n}}}$.

    Then, by the induction hypothesis, $\den t_\phi\subseteq\den{\B^n}=
    \den{\B}\times\den{\B^{n-1}}= \{a\times
    b~|~a\in\den{\B},b\in\den{\B^{n-1}}\}$. So,
    $\den{\head~t}_\phi=\{a~|~a\times b\in\den
    t_\phi,a\in\den{\B}\}=\{a~|~a\times b\in\{a'\times
    b'~|~a'\in\den{\B},b'\in\den{\B^{n-1}}\}\}=\{a~|~a\in\den{\B}\}=\den{\B}$.
  \item $\vcenter{\infer[\times_{El}]{{\Gamma\vdash \tail~t:\B^{n-1}}}{{\Gamma\vdash t:\B^n}}}$.

    Then, by the induction
    hypothesis, $\den t_\phi\subseteq\den{\B^n}=\den{\B}\times\den{\B^{n-1}}=\{a\times b~|~a\in\den{\B},b\in\den{\B^{n-1}}\}$. So
    $\den{\tail~t}_\phi
    =\{\prod_{i=2}^n a_i~|~\prod_{i=1}^n a_i\in\den t_\phi,a_1\in\den{\B}\}
    =\{\prod_{i=2}^n a_i~|~\prod_{i=1}^n a_i\in\{a'\times b'~|~a'\in\den{\B},b'\in\den{\B^{n-1}}\}\}
    =\{b~|~b\in\den{\B^{n-1}}\}
    =\den{\B^{n-1}}$.
  \item $\vcenter{\infer[\Uparrow_r]{{\Gamma\vdash \Uparrow_r t:S(A\times
          B)}}{{\Gamma\vdash t:S(S(A)\times B)}}}$.

    Then, by the induction hypothesis, $\den{\Uparrow_r t}_\phi =\den
    t_\phi\subseteq\den{S(S(A)\times B)} =\gen(\gen{\den A}\times\den B)
    =\gen(\den A\times\den B) =\den{S(A\times B)}$.
  \item $\vcenter{\infer[\Uparrow_\ell]{{\Gamma\vdash \Uparrow_\ell t:S(A\times
          B)}}{{\Gamma\vdash t:S(A\times S(B))}}}$.
    This case is analogous to the previous one.
    \qed
  \end{itemize}
\end{proof}

\section{Trace and typing of the Deutsch algorithm}\label{ap:Deutsch}

We may use $\ket{q_1\cdots q_n}$ as a shorthand notation for $\ket{q_1}\times\cdots\times\ket{q_n}$.

The full trace of $\s{Deutsch}_{id}$ is given below.
\begin{align*}
  \omit\rlap{\parbox{\textwidth}{$\s{Deutsch}_{id}$}}\\
  =\quad &
           \pi_1~(\Uparrow_r \s H_1~(\s U_f~\Uparrow_\ell\Uparrow_r (\s H_{\textsl{both}}~\ket{01})))\\
  \red\rbetab&
               \pi_1\ (\Uparrow_r \s H_1\ (\s U_{id}\
               \Uparrow_\ell\Uparrow_r ((\s H(\head~\ket{01}))\times(\s H(\tail~\ket{01})))))
  \\
  \red\rhead&
              \pi_1\ (\Uparrow_r \s H_1\ (\s U_{id}\
              \Uparrow_\ell\Uparrow_r ((\s H\ket 0)\times(\s H(\tail~\ket{01})))))\\
  \red\rtail&
              \pi_1\ \Uparrow_r \s H_1\ (\s U_{id}\
              \Uparrow_\ell\Uparrow_r ((\s H\ket 0)\times(\s H\ket 1)))\\
  \red{\rbetab^2}&
                   \pi_1\ (\Uparrow_r \s H_1\ (\s U_{id}\
                   \Uparrow_\ell\Uparrow_r (
                   \frac 1{\sqrt 2}.\pair{\ket 0}{(\ite{\ket 0}{(-\ket 1)}{\ket 1})}
                   \times
                   \frac 1{\sqrt 2}.\pair{\ket 0}{(\ite{\ket 1}{(-\ket 1)}{\ket 1})}
                   )))
  \\
  \red\riffalse&
                 \pi_1\ (\Uparrow_r \s H_1\ (\s U_{id}\
                 \Uparrow_\ell\Uparrow_r (
                 \frac 1{\sqrt 2}.\pair{\ket 0}{\ket 1}
                 \times
                 \frac 1{\sqrt 2}.\pair{\ket 0}{(\ite{\ket 1}{(-\ket 1)}{\ket 1})}
                 )))
  \\
  \red\riftrue&
                \pi_1\ (\Uparrow_r \s H_1\ (\s U_{id}\
                \Uparrow_\ell\Uparrow_r (\frac 1{\sqrt 2}.\pair{\ket 0}{\ket 1}\times\frac 1{\sqrt 2}.\npair{\ket 0}{\ket 1}
                )))
  \\
  \red\rdistscalr&
                   \pi_1\ (\Uparrow_r \s H_1\ (\s U_{id}\
                   \Uparrow_\ell\frac 1{\sqrt 2}.\Uparrow_r(\pair{\ket 0}{\ket 1}\times\frac 1{\sqrt 2}.\npair{\ket 0}{\ket 1}
                   )))
  \\
  \red\rdistsumr&
                  \pi_1\ (\Uparrow_r \s H_1\ (\s U_{id}\
                  \Uparrow_\ell\frac 1{\sqrt 2}.
                  ( {\Uparrow_r
                  \ket 0\times\frac 1{\sqrt 2}.\npair{\ket 0}{\ket 1}}
                  + {\Uparrow_r \ket
                  1\times\frac 1{\sqrt 2}.\npair{\ket 0}{\ket 1}})))
  \\
  \red\rcaneutr&
                 \pi_1\ (\Uparrow_r \s H_1\ (\s U_{id}\
                 \Uparrow_\ell\frac 1{\sqrt 2}.
                 ( { \ket 0\times\frac 1{\sqrt 2}.\npair{\ket 0}{\ket 1} }
                 + { \ket 1\times\frac 1{\sqrt 2}.\npair{\ket 0}{\ket 1}
                 })))
  \\
  \red\rdistcascal&
                    \pi_1\ (\Uparrow_r \s H_1\ (\s U_{id}\
                    \frac 1{\sqrt 2}.
                    \Uparrow_\ell
                    ( { \ket 0\times\frac 1{\sqrt 2}.\npair{\ket 0}{\ket 1} }
                    + { \ket 1\times\frac 1{\sqrt 2}.\npair{\ket 0}{\ket 1}
                    })))
  \\
  \red\rdistcasum&
                   \pi_1\ (\Uparrow_r \s H_1\ (\s U_{id}\
                   \frac 1{\sqrt 2}.
                   ( \Uparrow_\ell (\ket 0\times\frac 1{\sqrt 2}.\npair{\ket 0}{\ket 1})
                   + \Uparrow_\ell (\ket 1\times\frac
                   1{\sqrt 2}.\npair{\ket 0}{\ket 1}))))
  \\
  \red{\rdistscall^2}&
                       \pi_1\ (\Uparrow_r \s H_1\ (\s U_{id}\
                       \frac 1{\sqrt 2}.
                       ( \frac 1{\sqrt 2}.\Uparrow_\ell (\ket 0\times\npair{\ket 0}{\ket 1})
                       + \frac 1{\sqrt 2}.\Uparrow_\ell
                       (\ket 1\times\npair{\ket 0}{\ket 1}))))
  \\
  \red\rdists&
               \pi_1\ (\Uparrow_r \s H_1\ (\s U_{id}\
               ( \frac 1{\sqrt 2}.  \frac 1{\sqrt 2}.\Uparrow_\ell (\ket 0\times\npair{\ket 0}{\ket 1})
               + \frac 1{\sqrt 2}. \frac 1{\sqrt 2}.\Uparrow_\ell (\ket 1\times\npair{\ket 0}{\ket 1}))))
  \\
  \red{\rprod}&
                \pi_1\ (\Uparrow_r \s H_1\ (\s U_{id}\
                ( \frac 12.\Uparrow_\ell (\ket 0\times\npair{\ket 0}{\ket 1})
                + \frac 12.\Uparrow_\ell (\ket
                1\times\npair{\ket 0}{\ket 1}))))
  \\
  \red{\rdistsuml^2}&
                      \pi_1\ (\Uparrow_r \s H_1\ (\s U_{id}\
                      ( \frac 12.  \pair{\Uparrow_\ell\ket{00}}{\Uparrow_\ell\ket 0\times(-\ket 1)}
                      + \frac 12. \pair{\Uparrow_\ell\ket{10}}{\Uparrow_\ell\ket 1\times(-\ket 1)})))
  \\
  \red{\rdistscall^2}&
                       \pi_1\ (\Uparrow_r \s H_1\ (\s U_{id}\
                       ( \frac 12.  \npair{\Uparrow_\ell\ket{00}}{\Uparrow_\ell\ket{01}}
                       + \frac 12. \npair{\Uparrow_\ell\ket{10}}{\Uparrow_\ell\ket{11}})))
  \\
  \red{\rcaneutl^4}&
                     \pi_1\ (\Uparrow_r \s H_1\ (\s U_{id}\
                     \pair { \frac 12.  \npair{\ket{00}}{\ket{01}} }
                     { \frac 12.  \npair{\ket{10}}{\ket{11}} })) \\
  \red{\rlinr}&
                \pi_1\ (\Uparrow_r \s H_1\
                \pair{\s U_{id}\ \frac 12.\npair{\ket{00}}{\ket{01}}}{\s U_{id}\ \frac 12.\npair{\ket{10}}{\ket{11}}}) \\
  \red{\rlinscalr^2}&
                      \pi_1\ (\Uparrow_r \s H_1\
                      \pair{\frac 12.\s U_{id}\npair{\ket{00}}{\ket{01}}}{\frac 12.\s U_{id}\npair{\ket{10}}{\ket{11}}}) \\
  \red{\rlinr^2}&
                  \pi_1\ (\Uparrow_r \s H_1\
                  \pair{\frac 12.\pair{\s U_{id}\ket{00}}{\s U_{id}(-\ket{01})}}{\frac 12.\pair{\s U_{id}\ket{10}}{\s U_{id}(-\ket{11})}}) \\
  \red{\rlinscalr^2}&
                      \pi_1\ (\Uparrow_r \s H_1\
                      \pair{\frac 12.\npair{\s U_{id}\ket{00}}{\s U_{id}\ket{01}}}{\frac 12.\npair{\s U_{id}\ket{10}}{\s U_{id}\ket{11}}}) \\
  \red\rbetab&
               \begin{aligned}[t]
                 \pi_1\ (\Uparrow_r\s H_1\ ( & \frac 12.(
                 (\head\ket{00})\times(\ite{(\tail\ket{00})}{(\s{not}(id(\head\ket{00})))}{(id(\head\ket{00}))})
                 \\
                 -& { \s U_{id}\ \ket{01} }) +{ \frac 12. \npair{ \s U_{id}\
                     \ket{10} }{ \s U_{id}\ \ket{11} } }))
               \end{aligned}
  \\
  \red{\rhead^3}&
                  \begin{aligned}[t]
                    \pi_1\ (\Uparrow_r\s H_1\ (
                    & \frac 12.( \ket 0\times(\ite{(\tail\ket{00})}{(\s{not}(id\ket 0)))}{(id\ket 0)}) \\
                    -& { \s U_{id}\ \ket{01} }) +({ \frac 12. \npair{ \s U_{id}\
                        \ket{10} }{ \s U_{id}\ \ket{11} } }))
                  \end{aligned}
  \\
  \red{\rtail}&
                \begin{aligned}[t]
                  \pi_1\ (\Uparrow_r\s H_1\ ( &
                  \frac 12.( \ket 0\times(\ite{\ket 0}{(\s{not}(id\ket 0)))}{(id\ket 0)}) \\
                  -& { \s U_{id}\ \ket{01} }) +({ \frac 12. \npair{ \s U_{id}\
                      \ket{10} }{ \s U_{id}\ \ket{11} } }))
                \end{aligned}
  \\
  \red\riffalse&
                 \pi_1\ (\Uparrow_r\s H_1\
                 \pair
                 { \frac 12.  \npair{\ket 0\times(id\ket 0)} { \s U_{id}\ \ket{01} } }
                 { \frac 12.  \npair{ \s U_{id}\ \ket{10} }{ \s U_{id}\ \ket{11} } }) \\
  \red\rbetab&
               \pi_1\ (\Uparrow_r\s H_1\
               \pair
               { \frac 12.  \npair{\ket{00}} { \s U_{id}\ \ket{01} } }
               { \frac 12.  \npair{ \s U_{id}\ \ket{10} }{ \s U_{id}\ \ket{11} } }) \\
  \lra^*&
          \pi_1\ (\Uparrow_r\s H_1\
          \pair
          { \frac 12.  \npair{\ket{00}}{\ket{01}} }
          { \frac 12.  \npair{\ket{11}}{\ket{10}} }) \\
  \red\rlinr&
              \pi_1\ (\Uparrow_r
              \pair
              { \s H_1\ \frac 12.  \npair{\ket{00}}{\ket{01}} }
              { \s H_1\ \frac 12.  \npair{\ket{11}}{\ket{10}} }) \\
  \red{\rlinscalr^2}&
                      \pi_1\ (\Uparrow_r
                      \pair
                      { \frac 12.  \s H_1 \npair{\ket{00}}{\ket{01}} }
                      { \frac 12.  \s H_1 \npair{\ket{11}}{\ket{10}} }) \\
  \red{\rlinr^2}&
                  \pi_1\ (\Uparrow_r
                  \pair
                  { \frac 12.  \pair{\s H_1\ket{00}}{\s H_1(-\ket{01})} }
                  { \frac 12.  \pair{\s H_1\ket{11}}{\s H_1(-\ket{10})} }) \\
  \red{\rlinscalr^2}&
                      \pi_1\ (\Uparrow_r
                      \pair
                      { \frac 12.  \npair{\s H_1\ket{00}}{\s H_1\ket{01}} }
                      { \frac 12.  \npair{\s H_1\ket{11}}{\s H_1\ket{10}} }) \\
  \red\rbetab&
               \pi_1\ (\Uparrow_r
               \pair
               { \frac 12.  \npair{ (\s H(\head\ket{00}))\times(\tail\ket{00}) }{\s H_1\ket{01}} }
               { \frac 12.  \npair{\s H_1\ket{11}}{\s H_1\ket{10}} }) \\
  \red\rhead&
              \pi_1\ (\Uparrow_r
              \pair
              { \frac 12.  \npair{ (\s H\ket 0)\times(\tail\ket{00}) }{\s H_1\ket{01}} }
              { \frac 12.  \npair{\s H_1\ket{11}}{\s H_1\ket{10}} }) \\
  \red\rtail&
              \pi_1\ (\Uparrow_r
              \pair
              { \frac 12.  \npair{ (\s H\ket 0)\times\ket 0 }{\s H_1\ket{01}} }
              { \frac 12.  \npair{\s H_1\ket{11}}{\s H_1\ket{10}} }) \\
  \red\rbetab&
               \pi_1\ (\Uparrow_r
               \pair
               { \frac 12.  \npair{ ( \frac 1{\sqrt 2}.  \pair{\ket 0}{\ite{\ket 0}{(-\ket 1)}{\ket 1}}) \times\ket 0 }{\s H_1\ket{01}} }
               { \frac 12.  \npair{\s H_1\ket{11}}{\s H_1\ket{10}} }) \\
  \red\riffalse&
                 \pi_1\ (\Uparrow_r
                 \pair
                 { \frac 12.  \npair{ ( \frac 1{\sqrt 2}.  \pair{\ket 0}{\ket 1}) \times\ket 0 }{\s H_1\ket{01}} }
                 { \frac 12.  \npair{\s H_1\ket{11}}{\s H_1\ket{10}} }) \\
  \red\rdistcasum&
                   \pi_1
                   \pair
                   { \Uparrow_r \frac 12.  \npair{ ( \frac 1{\sqrt 2}.\ket 0+\frac 1{\sqrt 2}.\ket 1) \times\ket 0 }{\s H_1\ket{01}} }
                   { \Uparrow_r \frac 12.  \npair{\s H_1\ket{11}}{\s H_1\ket{10}} } \\
  \red{\rdistcascal^2}&
                        \pi_1
                        \pair
                        { \frac 12.  \Uparrow_r \npair{ ( \frac 1{\sqrt 2}.\ket 0+\frac 1{\sqrt 2}.\ket 1) \times\ket 0 }{\s H_1\ket{01}} }
                        { \frac 12.  \Uparrow_r \npair{\s H_1\ket{11}}{\s H_1\ket{10}} } \\
  \red\rdistcasum&
                   \begin{aligned}[t]
                     \pi_1 ( & { \frac 12.  \pair{ \Uparrow_r ( ( \frac 1{\sqrt 2}.\ket 0+\frac 1{\sqrt 2}.\ket 1) \times\ket 0) }{ \Uparrow_r (-\s H_1\ket{01}) } }\\
                     +& { \frac 12. \Uparrow_r \npair{\s
                         H_1\ket{11}}{\s H_1\ket{10}} })
                   \end{aligned}
  \\
  \red\rdistcascal&
                    \begin{aligned}[t]
                      \pi_1 ( & { \frac 12.  \npair{ \Uparrow_r ( \pair{\frac 1{\sqrt 2}.\ket 0}{\frac 1{\sqrt 2}.\ket 1} \times\ket 0) }{ \Uparrow_r \s H_1\ket{01} } }\\
                      +& { \frac 12. \Uparrow_r \npair{\s
                          H_1\ket{11}}{\s H_1\ket{10}} })
                    \end{aligned}
  \\
  \red\rdistsumr&
                  \begin{aligned}[t]
                    \pi_1 ( & { \frac 12.  \npair{ \pair { \Uparrow_r (\frac 1{\sqrt 2}.\ket 0)\times\ket 0 }{ \Uparrow_r (\frac 1{\sqrt 2}.\ket 1)\times\ket 0 } }{ \Uparrow_r \s H_1\ket{01} } }\\
                    +& { \frac 12. \Uparrow_r \npair{\s
                        H_1\ket{11}}{\s H_1\ket{10}} })
                  \end{aligned}
  \\
  \red{\rdistscalr^2}&
                       \begin{aligned}[t]
                         \pi_1 ( & { \frac 12. \npair{ \pair { \frac 1{\sqrt
                                 2}.\Uparrow_r\ket{00} }{
                               \frac 1{\sqrt
                                 2}.\Uparrow_r\ket{10} }
                           }{ \Uparrow_r \s H_1\ket{01} }
                         }\\
                         +& { \frac 12. \Uparrow_r
                           \npair{\s H_1\ket{11}}{\s H_1\ket{10}} } )
                       \end{aligned}
  \\
  \red{\rcaneutr^2}&
                     \pi_1 ( { \frac 12.  \npair{ \pair { \frac 1{\sqrt 2}.\ket{00} }{ \frac 1{\sqrt 2}.\ket{10} } }{ \Uparrow_r \s H_1\ket{01} } }
                     +
                     { \frac 12.  \Uparrow_r \npair{\s H_1\ket{11}}{\s H_1\ket{10}} }) \\
  \lra^*&
          \begin{aligned}[t]
            \pi_1 (\frac 12. ( &{ \pair { \pair{\frac 1{\sqrt 2}.\ket{00}}{\frac
                  1{\sqrt 2}.\ket{10}} }{ \npair{-\frac 1{\sqrt
                    2}.\ket{01}}{\frac 1{\sqrt 2}.\ket{11}} }
            }\\
            +& { \pair{ \npair{\frac 1{\sqrt 2}.\ket{01}}{\frac 1{\sqrt
                    2}.\ket{11}} }{ \pair{-\frac 1{\sqrt 2}.\ket{00}}{\frac
                  1{\sqrt 2}.\ket{10}} } } ) )
          \end{aligned}
  \\
  =_{AC}\ &
            \begin{aligned}[t]
              \pi_1 (\frac 12. ( &{ \pair { \npair {\frac 1{\sqrt 2}.\ket{00}}
                  {\frac 1{\sqrt 2}.\ket{00}} }{ \pair {-\frac 1{\sqrt
                      2}.\ket{01}} {\frac 1{\sqrt 2}.\ket{01}} }
              }\\
              +& { \pair{ \npair {-\frac 1{\sqrt 2}.\ket{11}} {\frac 1{\sqrt
                      2}.\ket{11}} }{ \pair {\frac 1{\sqrt 2}.\ket{10}} {\frac
                    1{\sqrt 2}.\ket{10}} } } ) )
            \end{aligned}
  \\
  \red{\rfact^4} &
                   \begin{aligned}[t]
                     \pi_1 (\frac 12. ( &{ \pair { \npair {\frac 1{\sqrt 2}}
                         {\frac 1{\sqrt 2}} .\ket{00} }{ \pair {-\frac 1{\sqrt
                             2}} {\frac 1{\sqrt 2}} .\ket{01} }
                     }\\
                     +& { \pair{ \npair {-\frac 1{\sqrt 2}} {\frac 1{\sqrt
                             2}}.\ket{11} }{ \pair {\frac 1{\sqrt 2}} {\frac
                           1{\sqrt 2}}.\ket{10} } } ) )
                   \end{aligned}
  \\
  =\quad &
           \pi_1
           (\frac 12.
           (
           {
           \pair
           {
           0.\ket{00}
           }{
           0.\ket{01}
           }
           }
           +
           {
           \pair{
           {-\frac 2{\sqrt 2}}.\ket{11}
           }{
           {\frac 2{\sqrt 2}}.\ket{10}
           }
           }
           )
           )
  \\
  \red{\rzeros^2} &
                    \pi_1
                    (\frac 12.
                    (
                    {
                    \pair{\z[\B\times\B]}{\z[\B\times\B]}
                    }
                    +
                    {
                    \pair
                    {
                    {-\frac 2{\sqrt 2}}.\ket{11}
                    }
                    {
                    \frac 2{\sqrt 2}.\ket{10}
                    }
                    }
                    )
                    )
  \\
  \red{\rneut^2} &
                   \pi_1
                   (\frac 12.
                   \pair
                   {
                   {-\frac 2{\sqrt 2}}.\ket{11}
                   }
                   {
                   \frac 2{\sqrt 2}.\ket{10}
                   }
                   )
  \\
  \red\rdists &
                \pi_1
                \pair
                {
                {\frac 12.(-\frac 2{\sqrt 2}}).\ket{11}
                }
                {
                \frac 12.\frac 2{\sqrt 2}.\ket{10}
                }
  \\
  \red{\rprod^2} &
                   \pi_1
                   \pair
                   {
                   {-\frac 1{\sqrt 2}}.\ket{11}
                   }
                   {
                   \frac 1{\sqrt 2}.\ket{10}
                   }
  \\
  =_{AC}\ &
            \pi_1
            \npair
            {
            \frac 1{\sqrt 2}.\ket{10}
            }
            {
            \frac 1{\sqrt 2}.\ket{11}
            }
  \\
  \red\rproj &
               \ket 1\times\npair
               {
               \frac 1{\sqrt 2}.\ket{0}
               }
               {
               \frac 1{\sqrt 2}.\ket{1}
               }
\end{align*}

The typing of $\s{Deutsch}_f$, for any $\vdash f:\B\Rightarrow\B$, is given
below:
\begin{equation}
  \label{eq:H}
  \infer[\Rightarrow_I]
  {\vdash\s H:\B\Rightarrow S(\B)}
  {
    \infer[\preceq]
    {x:\B\vdash\frac 1{\sqrt 2}.(\ket 0+\ite{x}{(-\ket 1)}{\ket 1}):S(\B)}
    {
      \infer[S_I^\alpha]
      {x:\B\vdash\frac 1{\sqrt 2}.(\ket 0+\ite{x}{(-\ket 1)}{\ket 1}):S(S(S(\B)))}
      {
        \infer[S_I^+]
        {x:\B\vdash\ket 0+\ite{x}{(-\ket 1)}{\ket 1}:S(S(\B))}
        {
          \infer[\preceq]
          {\vdash\ket 0:S(\B)}
          {\infer[\tax_{\ket 0}]{\vdash\ket 0:\B}{}}
          &
          \infer[\Rightarrow_{E}]
          {x:\B\vdash\ite x{(-\ket 1)}{\ket 1}:S(\B)}
          {
            \infer[\tif]{\vdash\ite{}{(-\ket 1)}{(\ket 1)}:\B\Rightarrow S(\B)}
            {
              \infer[S_I^\alpha]{\vdash -\ket 1:S(\B)}
              {
                \infer[\tax_{\ket 1}]{\vdash\ket 1:\B}{}
              }
              &
              \infer[\preceq]{\vdash\ket 1:S(\B)}
              {
                \infer[\tax_{\ket 1}]{\vdash\ket 1:\B}{}
              }
            }
            &
            \infer[\tax]{x:\B\vdash x:\B}{}
          }
        }
      }
    }
  }
\end{equation}

\begin{equation}
  \label{eq:Hone}
  \infer[\Rightarrow_I]
  {\vdash \s H_1:\B^2\Rightarrow S(\B)\times\B}
  {
    \infer[C]
    {x:\B^2\vdash (\s H\ (\head\ x))\times(\tail\ x):S(\B)\times\B}
    {
      \infer[\times_I]
      {x:\B^2,y:\B^2\vdash (\s H\ (\head\ x))\times(\tail\ y):S(\B)\times\B}
      {
        \infer[\Rightarrow_E]
        {x:\B^2\vdash \s H\ (\head\ x):S(\B)}
        {
          \infer{\vdash \s H:\B\Rightarrow S(\B)}{\eqref{eq:H}}
          &
          \infer[E_r]
          {x:\B^2\vdash\head\ x:\B}
          {
            \infer[\tax]{x:\B^2\vdash x:\B^2}{}
          }
        }
        &
        \infer[E_l]
        {y:\B^2\vdash\tail\ y:\B}
        {
          \infer[\tax]{y:\B^2\vdash y:\B^2}{}
        }
      }
    }
  }
\end{equation}

\begin{equation}
  \label{eq:not}
  \infer[\Rightarrow_I]
  {\vdash\s{not}:\B\Rightarrow\B}
  {
    \infer[\Rightarrow_E]
    {x:\B\vdash\ite x{\ket 0}{\ket 1}:\B}
    {
      \infer[\tif]{\vdash\ite{}{\ket 0}{\ket 1}:\B\Rightarrow\B}
      {
        \infer[\tax_{\ket 0}]{\vdash\ket 0:\B}{}
        &
        \infer[\tax_{\ket 1}]{\vdash\ket 1:\B}{}
      }
      &
      \infer[\tax]{x:\B\vdash x:\B}{}
    }
  }
\end{equation}

\begin{equation}
  \hspace{-3.5cm}\label{eq:ite}
  {
    \infer[C]
    {y:\B^2\vdash\ite{(\tail\ y)}{(\s{not}\ (f\ (\head\ y)))}{(f\ (\head\ y))}:\B}
    {
      \infer[\Rightarrow_E]{x:\B^2,y:\B^2\vdash\ite{(\tail\ x)}{(\s{not}\ (f\ (\head\ y)))}{(f\ (\head\ y))}:\B}
      {
        \infer[\tif]
        {y:\B^2\vdash\ite{}{(\s{not}\ (f\ (\head\ y)))}{(f\ (\head\ y))}:\B\Rightarrow\B}
        {
          \infer[\Rightarrow_E]
          {y:\B^2\vdash \s{not}\ (f\ (\head\ y)):\B}
          {
            \infer
            {\vdash \s{not}:\B\Rightarrow\B}
            {\eqref{eq:not}}
            &
            \infer[\Rightarrow_E]
            {y:\B^2\vdash f\ (\head\ y):\B}
            {
              \infer{\vdash f:\B\Rightarrow\B}{}
              &
              \infer[\times_{E_r}]
              {y:\B^2\vdash\head\ y:\B}
              {\infer[\tax]{y:\B^2\vdash y:\B^2}{}}
            }
          }
          &
          \infer[\Rightarrow_E]
          {y:\B^2\vdash f\ (\head\ y):\B}
          {
            \infer{\vdash f:\B\Rightarrow\B}{}
            &
            \infer[\times_{E_r}]
            {y:\B^2\vdash\head\ y:\B}
            {\infer[\tax]{y:\B^2\vdash y:\B^2}{}}
          }
        }
        &
        \infer[\times_{E_l}]
        {x:\B^2\vdash\tail\ x:\B}
        {\infer[\tax]{x:\B^2\vdash x:\B^2}{}}
      }
    }
  }
\end{equation}

\begin{equation}
  \label{eq:Uf}
  \infer[\Rightarrow_I]
  {\vdash\s U_f:\B^2\Rightarrow\B^2}
  {
    \infer[C]
    {x:\B^2\vdash (\head\ x)\times\ite{(\tail\ x)}{(\s{not}\ (f\ (\head\ x)))}{f\ (\head\ x)}:\B^2}
    {
      \infer[\times_I]
      {x:\B^2,y:\B^2\vdash (\head\ x)\times\ite{(\tail\ y)}{(\s{not}\ (f\ (\head\ y)))}{f\ (\head\ y)}:\B^2}
      {
        \infer[\times_{E_r}]
        {x:\B^2\vdash\head\ x:\B}
        {\infer[\tax]{x:\B^2\vdash x:\B^2}{}}
        &
        \infer
        {y:\B^2\vdash\ite{(\tail\ y)}{(\s{not}\ (f\ (\head\ y)))}{(f\ (\head\ y))}:\B}
        {\eqref{eq:ite}}
      }
    }
  }
\end{equation}

\begin{equation}
  \label{eq:Hboth}
  \hspace{-2.8cm}\infer[\Rightarrow_E]
  {\vdash \s H_{\textsl{both}}\ket{01}: S(\B)\times S(\B)}
  {
    \infer[\Rightarrow_I]
    {\vdash \s H_{\textsl{both}}:\B^2\Rightarrow S(\B)\times S(\B)}
    {
      \infer[C]
      {x:\B^2\vdash(\s H\ (\head\ x))\times(\s H\ (\tail\ x)):S(\B)\times S(\B)}
      {  
        \infer[\times_I]
        {x:\B^2,y:\B^2\vdash(\s H\ (\head\ x))\times(\s H\ (\tail\ y)):S(\B)\times S(\B)}
        {
          \infer[\Rightarrow_E]
          {x:\B^2\vdash \s H\ (\head\ x):S(\B)}
          {
            \infer{\vdash \s H:\B\Rightarrow S(\B)}{\eqref{eq:H}}
            &
            \infer[\times_{E_r}]
            {x:\B^2\vdash\head\ x:\B}
            {\infer[\tax]{x:\B^2\vdash x:\B^2}{}}
          }
          &
          \infer[\Rightarrow_E]
          {y:\B^2\vdash \s H\ (\tail\ y):S(\B)}
          {
            \infer{\vdash \s H:\B\Rightarrow S(\B)}{\eqref{eq:H}}
            &
            \infer[\times_{E_l}]
            {y:\B^2\vdash\tail\ y:\B}
            {\infer[\tax]{y:\B^2\vdash y:\B^2}{}}
          }
        }
      }
    }
    &
    \hspace{-1cm}\infer[\times_I]
    {\vdash\ket{01}:\B^2}
    {
      \infer[\tax_{\ket 0}]
      {\vdash\ket 0:\B}{}
      &
      \infer[\tax_{\ket 1}]
      {\vdash\ket 1:\B}{}
    }
  }
\end{equation}

\[
  \infer[S_E]
  {\vdash \s{Deutsch}_f:\B\times S(\B)}
  {
    \infer[\Uparrow_r]
    {\vdash \Uparrow_r\s H_1\ (U_f\ \Uparrow_\ell\Uparrow_r \s
      H_{\textsl{both}}\ket{01}): S(\B^2)}
    {
      {
        \infer[\Rightarrow_{ES}]
        {\vdash \s H_1\ (U_f\ \Uparrow_\ell\Uparrow_r \s
          H_{\textsl{both}}\ket{01}): S(S(\B)\times\B)}
        {
          \infer[\preceq]
          {\vdash \s H_1:S(\B^2\Rightarrow S(\B)\times\B)}
          {
            \infer{\vdash \s H_1:\B^2\Rightarrow S(\B)\times\B}{\eqref{eq:Hone}}
          }
          &
          \infer[\Rightarrow_{ES}]
          {\vdash U_f\ (\Uparrow_\ell\Uparrow_r \s H_{\textsl{both}}\ket{01}):
            S(\B^2)}
          {
            \infer[\preceq]
            {\vdash U_f:S(\B^2\Rightarrow\B^2)}
            {
              \infer{\vdash U_f:\B^2\Rightarrow\B^2}{\eqref{eq:Uf}}
            }
            &
            \infer[\Uparrow_\ell]
            {\vdash\Uparrow_\ell\Uparrow_r \s H_{\textsl{both}}\ket{01}:S({\B^2})}
            {
              \infer[\Uparrow_r]
              {\vdash\Uparrow_r \s
                H_{\textsl{both}}\ket{01}:S({\B\times S(\B)})}
              {
                \infer[\preceq]
                {\vdash \s H_{\textsl{both}}\ket{01}:S(S(\B)\times S(\B))}
                {
                  \infer
                  {\vdash \s H_{\textsl{both}}\ket{01}:S(\B)\times S(\B)}
                  {\eqref{eq:Hboth}}
                }
              }
            }
          }
        }
      }
    }
  }
\]

\section{Trace and typing of the Teleportation algorithm}\label{ap:telep}
The full trace of $\s{Teleportation}\ \pair{\alpha.\ket 0}{\beta.\ket 1}$ is
given below.

\begin{align*}
  \omit\rlap{\parbox{\textwidth}{$\s{Teleportation}\ \pair{\alpha.\ket 0}{\beta.\ket 1}$}}
  \\
  =\quad &
           (\lambda q^{S(\B)}.(\s{Bob}~\Uparrow_\ell(\s{Alice}~(q\times\beta_{00}))))\ \pair{\alpha.\ket 0}{\beta.\ket 1}\\
  \red{\rbetan} & 
                  \s{Bob}\ \Uparrow_\ell(\s{Alice}\ (\pair{\alpha.\ket 0}{\beta.\ket 1}\times\beta_{00}))
  \\
  \red{\rbetan} &
                  \s{Bob}\ \Uparrow_\ell(\pi_2(\Uparrow_r \s H_1^3(\s{cnot}_{12}^3\ \Uparrow_\ell \Uparrow_r \pair{\alpha.\ket 0}{\beta.\ket 1}\times\beta_{00})))
  \\
  \red{\rdistsumr} &
                     \s{Bob}\ \Uparrow_\ell(\pi_2(\Uparrow_r \s H_1^3(\s{cnot}_{12}^3\ \Uparrow_\ell \pair{\Uparrow_r (\alpha.\ket 0\times\beta_{00)}}{\Uparrow_r (\beta.\ket 1\times\beta_{00})})))
  \\
  \red{\rdistcasum} &
                      \s{Bob}\ \Uparrow_\ell(\pi_2(\Uparrow_r \s H_1^3(\s{cnot}_{12}^3\ \pair{\Uparrow_\ell \Uparrow_r (\alpha.\ket 0\times\beta_{00})}{\Uparrow_\ell \Uparrow_r (\beta.\ket 1\times\beta_{00})})))
  \\
  \red{\rdistscalr^2} &
                        \s{Bob}\ \Uparrow_\ell(\pi_2(\Uparrow_r \s H_1^3(\s{cnot}_{12}^3\ \pair{\Uparrow_\ell \alpha.\Uparrow_r (\ket 0\times\beta_{00})}{\Uparrow_\ell \beta.\Uparrow_r (\ket 1\times\beta_{00})})))
  \\
  \red{\rcaneutr^2} &
                      \s{Bob}\ \Uparrow_\ell(\pi_2(\Uparrow_r \s H_1^3(\s{cnot}_{12}^3\ \pair{\Uparrow_\ell \alpha.(\ket 0\times\beta_{00})}{\Uparrow_\ell \beta.(\ket 1\times\beta_{00})})))
  \\
  \red{\rdistcascal^2} &
                         \s{Bob}\ \Uparrow_\ell(\pi_2(\Uparrow_r \s H_1^3(\s{cnot}_{12}^3\ \pair{\alpha.\Uparrow_\ell (\ket 0\times\beta_{00})}{\beta.\Uparrow_\ell (\ket 1\times\beta_{00})})))
  \\
  \red{\rdistsuml^2} &
                       \begin{aligned}[t]
                         \s{Bob}\ \Uparrow_\ell(\pi_2(\Uparrow_r \s
                         H_1^3(\s{cnot}_{12}^3\ ( \alpha.& (\Uparrow_\ell (\ket
                         0\times(\frac 1{\sqrt 2}.\ket{00}))
                         \\
                         +& (\Uparrow_\ell (\ket 0\times\frac 1{\sqrt
                           2}.\ket{11})))
                         \\
                         + \beta. & (\Uparrow_\ell (\ket 1\times(\frac 1{\sqrt
                           2}.\ket{00}))
                         \\
                         +& (\Uparrow_\ell (\ket 1\times\frac 1{\sqrt
                           2}.\ket{11}))) ) )))
                       \end{aligned}
  \\
  \red{\rdistscall^4} &
                        \begin{aligned}[t]
                          \s{Bob}\ \Uparrow_\ell(\pi_2(\Uparrow_r \s
                          H_1^3(\s{cnot}_{12}^3\ ( \alpha.& \pair {\frac 1{\sqrt
                              2}.\Uparrow_\ell \ket{000}} {\frac 1{\sqrt
                              2}.\Uparrow_\ell \ket{011}}
                          \\
                          + \beta. & \pair {\frac 1{\sqrt
                              2}.\Uparrow_\ell \ket{100}} {\frac 1{\sqrt
                              2}.\Uparrow_\ell \ket{111}} ))))
                        \end{aligned}
  \\
  \red{\rcaneutl^4} &
                      \begin{aligned}[t]
                        \s{Bob}\ \Uparrow_\ell(\pi_2(\Uparrow_r \s
                        H_1^3(\s{cnot}_{12}^3\ ( \alpha.& \pair {\frac 1{\sqrt
                            2}.\ket{000}} {\frac 1{\sqrt 2}.\ket{011}}
                        \\
                        + \beta. & \pair {\frac 1{\sqrt 2}.\ket{100}} {\frac
                          1{\sqrt 2}.\ket{111}} ))))
                      \end{aligned}
  \\
  \red{\rdists^2} &
                    \begin{aligned}[t]
                      \s{Bob}\ \Uparrow_\ell(\pi_2(\Uparrow_r \s
                      H_1^3(\s{cnot}_{12}^3\ ( & \pair {\alpha.\frac 1{\sqrt
                          2}.\ket{000}} {\alpha.\frac 1{\sqrt 2}.\ket{011}}
                      \\
                      + & \pair {\beta.\frac 1{\sqrt 2}.\ket{100}} {\beta.\frac
                        1{\sqrt 2}.\ket{111}} ))))
                    \end{aligned}
  \\
  \red{\rprod^4} &
                   \begin{aligned}[t]
                     \s{Bob}\ \Uparrow_\ell(\pi_2(\Uparrow_r \s H_1^3(\s{cnot}_{12}^3\
                     ( & \pair {\frac{\alpha}{\sqrt 2}.\ket{000}}
                     {\frac{\alpha}{\sqrt 2}.\ket{011}}
                     \\
                     + & \pair {\frac{\beta}{\sqrt 2}.\ket{100}}
                     {\frac{\beta}{\sqrt 2}.\ket{111}} ))))
                   \end{aligned}
  \\
  \red{\rlinr^3} &
                   \begin{aligned}[t]
                     \s{Bob}\ \Uparrow_\ell(\pi_2(\Uparrow_r \s H_1^3 ( & \pair
                     {\s{cnot}_{12}^3\ \frac{\alpha}{\sqrt 2}.\ket{000}}
                     {\s{cnot}_{12}^3\ \frac{\alpha}{\sqrt 2}.\ket{011}}
                     \\
                     + & \pair {\s{cnot}_{12}^3\ \frac{\beta}{\sqrt
                         2}.\ket{100}} {\s{cnot}_{12}^3\ \frac{\beta}{\sqrt
                         2}.\ket{111}} )))
                   \end{aligned}
  \\
  \red{\rlinscalr^4} &
                       \begin{aligned}[t]
                         \s{Bob}\ \Uparrow_\ell(\pi_2(\Uparrow_r \s H_1^3 ( & \pair
                         {\frac{\alpha}{\sqrt 2}.\s{cnot}_{12}^3\ket{000}}
                         {\frac{\alpha}{\sqrt 2}.\s{cnot}_{12}^3\ket{011}}
                         \\
                         + & \pair {\frac{\beta}{\sqrt
                             2}.\s{cnot}_{12}^3\ket{100}} {\frac{\beta}{\sqrt
                             2}.\s{cnot}_{12}^3\ket{111}} )))
                       \end{aligned}
  \\
  \red{\rbetab^4} &
                    \begin{aligned}[t]
                      \s{Bob}\ \Uparrow_\ell(\pi_2(\Uparrow_r \s H_1^3 ( ( & {
                        \frac{\alpha}{\sqrt 2}. ((\s{cnot}\ (head\
                        \ket{000}\times(head\ tail\ \ket{000})))\times(tail\
                        tail\ \ket{000})) }
                      \\
                      + & { \frac{\alpha}{\sqrt 2}. ((\s{cnot}\ (head\
                        \ket{011}\times(head\ tail\ \ket{011})))\times(tail\
                        tail\ \ket{011})) } )
                      \\
                      + ( & { \frac{\beta}{\sqrt 2}. ((\s{cnot}\ (head\
                        \ket{100}\times(head\ tail\ \ket{100})))\times(tail\
                        tail\ \ket{100})) }
                      \\
                      + & { \frac{\beta}{\sqrt 2}. ((\s{cnot}\ (head\
                        \ket{111}\times(head\ tail\ \ket{111})))\times(tail\
                        tail\ \ket{111})) } ) )))
                    \end{aligned}
  \\
  \red{\rtail^{12}} &
                      \begin{aligned}[t]
                        \s{Bob}\ \Uparrow_\ell(\pi_2(\Uparrow_r \s H_1^3 ( ( & {
                          \frac{\alpha}{\sqrt 2}. ((\s{cnot}\ (head\
                          \ket{000}\times(head\ \ket{00})))\times(\ket{0})) }
                        \\
                        + & { \frac{\alpha}{\sqrt 2}. ((\s{cnot}\ (head\
                          \ket{011}\times(head\ \ket{11})))\times(\ket{1})) }
                        )
                        \\
                        + ( & { \frac{\beta}{\sqrt 2}. ((\s{cnot}\ (head\
                          \ket{100}\times(head\ \ket{00})))\times(\ket{0})) }
                        \\
                        + & { \frac{\beta}{\sqrt 2}. ((\s{cnot}\ (head\
                          \ket{111}\times(head\ \ket{11})))\times(\ket{1})) }
                        ) )))
                      \end{aligned}
  \\
  \red{\rhead^8} &
                   \begin{aligned}[t]
                     \s{Bob}\ \Uparrow_\ell(\pi_2(\Uparrow_r \s H_1^3 ( ( & {
                       \frac{\alpha}{\sqrt 2}. ((\s{cnot}
                       \ket{00})\times(\ket{0})) }
                     \\
                     + & { \frac{\alpha}{\sqrt 2}.
                       ((\s{cnot}\ket{01})\times(\ket{1})) } )
                     \\
                     + ( & { \frac{\beta}{\sqrt 2}.
                       ((\s{cnot}\ket{10})\times(\ket{0})) }
                     \\
                     + & { \frac{\beta}{\sqrt 2}.
                       ((\s{cnot}\ket{11})\times(\ket{1})) } ) )))
                   \end{aligned}
  \\
  \red{\rbetab^4} &
                    \begin{aligned}[t]
                      \s{Bob}\ \Uparrow_\ell(\pi_2(\Uparrow_r \s H_1^3( ( ( &
                      {\frac{\alpha}{\sqrt 2}.( ( (head \ket{00}) \times ( \ite
                        {(head \ket{00})} {(\s{not}(tail \ket{00}))} {(tail
                          \ket{00})} ) )
                        \times\ket 0)}\\
                      +& {\frac{\alpha}{\sqrt 2}.( ( (head \ket{01}) \times (
                        \ite {(head \ket{01})} {(\s{not}(tail \ket{01}))} {(tail
                          \ket{01})} ) ) \times\ket 1)}
                      )\\
                      + (& {\frac{\beta}{\sqrt 2}.( ( (head \ket{10}) \times (
                        \ite {(head \ket{10})} {(\s{not}(tail \ket{10}))} {(tail
                          \ket{10})} ) )
                        \times\ket 0)}\\
                      +& {\frac{\beta}{\sqrt 2}.( ( (head \ket{11}) \times (
                        \ite {(head \ket{11})} {(\s{not}(tail \ket{11}))} {(tail
                          \ket{11})} ) ) \times\ket 1)} ) ) )))
                    \end{aligned}
  \\
  \red{\rhead^8} &
                   \begin{aligned}[t]
                     \s{Bob}\ \Uparrow_\ell(\pi_2(\Uparrow_r \s H_1^3( ( ( &
                     {\frac{\alpha}{\sqrt 2}.( ( \ket 0 \times ( \ite {\ket 0}
                       {(\s{not}(tail \ket{00}))} {(tail \ket{00})} ) )
                       \times\ket 0)}\\
                     +&
                     {\frac{\alpha}{\sqrt 2}.( ( \ket 0 \times ( \ite {\ket 0} {(\s{not}(tail \ket{01}))} {(tail \ket{01})})) \times\ket 1)})\\
                     +
                     (& {\frac{\beta}{\sqrt 2}.( ( \ket 1 \times ( \ite {\ket 1} {(\s{not}(tail \ket{10}))} {(tail \ket{10})})) \times\ket 0)}\\
                     +& {\frac{\beta}{\sqrt 2}.(( \ket 1 \times ( \ite {\ket 1}
                       {(\s{not}(tail \ket{11}))} {(tail \ket{11})}))
                       \times\ket 1)})))))
                   \end{aligned}
  \\
  \red{\riftrue^2} &
                     \begin{aligned}[t]
                       \s{Bob}\ \Uparrow_\ell(\pi_2(\Uparrow_r \s H_1^3( ( (
                       &    {\frac{\alpha}{\sqrt 2}.( ( \ket 0 \times ( \ite {\ket 0} {(\s{not}(tail \ket{00}))} {(tail \ket{00})})) \times\ket 0)}\\
                       +& {\frac{\alpha}{\sqrt 2}.( ( \ket 0 \times ( \ite {\ket 0} {(\s{not}(tail \ket{01}))} {(tail \ket{01})})) \times\ket 1)})\\
                       + (& {\frac{\beta}{\sqrt 2}.((\ket 1 \times (\s{not}(tail \ket{10}))) \times\ket 0)}\\
                       +& {\frac{\beta}{\sqrt 2}.((\ket 1 \times (\s{not}(tail
                         \ket{11}))) \times\ket 1)})))))
                     \end{aligned}
  \\
  \red{\riffalse^2} &
                      \begin{aligned}[t]
                        \s{Bob}\ \Uparrow_\ell(\pi_2(\Uparrow_r \s H_1^3( ((
                        &    {\frac{\alpha}{\sqrt 2}.(( \ket 0 \times (tail \ket{00})) \times\ket 0)}\\
                        +& {\frac{\alpha}{\sqrt 2}.( ( \ket 0 \times {(tail \ket{01})}) \times\ket 1)})\\
                        + (& {\frac{\beta}{\sqrt 2}.((\ket 1 \times (\s{not}(tail \ket{10}))) \times\ket 0)}\\
                        +& {\frac{\beta}{\sqrt 2}.((\ket 1 \times (\s{not}(tail
                          \ket{11}))) \times\ket 1)})))))
                      \end{aligned}
  \\
  \red{\rtail^4} &
                   \begin{aligned}[t]
                     \s{Bob}\ \Uparrow_\ell(\pi_2(\Uparrow_r \s H_1^3( ((
                     &    {\frac{\alpha}{\sqrt 2}.(( \ket 0 \times \ket 0) \times\ket 0)}\\
                     +& {\frac{\alpha}{\sqrt 2}.( ( \ket 0 \times {\ket 1}) \times\ket 1)})\\
                     + (& {\frac{\beta}{\sqrt 2}.((\ket 1 \times (\s{not}\ket 0)) \times\ket 0)}\\
                     +& {\frac{\beta}{\sqrt 2}.((\ket 1 \times (\s{not}\ket 1))
                       \times\ket 1)})))))
                   \end{aligned}
  \\
  \red{\rbetab^2} &
                    \begin{aligned}[t]
                      \s{Bob}\ \Uparrow_\ell(\pi_2(\Uparrow_r \s H_1^3( ((
                      &    \frac{\alpha}{\sqrt 2}.( \ket 0 \times \ket 0 \times \ket 0)\\
                      +& \frac{\alpha}{\sqrt 2}.( \ket 0 \times \ket 1 \times \ket 1))\\
                      + (& \frac{\beta}{\sqrt 2}.( \ket 1 \times (\ite{\ket 0}{\ket 0}{\ket 1}) \times \ket 0)\\
                      +& \frac{\beta}{\sqrt 2}.( \ket 1 \times (\ite{\ket
                        1}{\ket 0}{\ket 1}) \times \ket 1))))))
                    \end{aligned}
  \\
  \red{\riffalse} &
                    \begin{aligned}[t]
                      \s{Bob}\ \Uparrow_\ell(\pi_2(\Uparrow_r \s H_1^3( ((
                      &    \frac{\alpha}{\sqrt 2}.( \ket 0 \times \ket 0 \times \ket 0)\\
                      +& \frac{\alpha}{\sqrt 2}.( \ket 0 \times \ket 1 \times \ket 1))\\
                      + (& \frac{\beta}{\sqrt 2}.( \ket 1 \times \ket 1 \times \ket 1)\\
                      +& \frac{\beta}{\sqrt 2}.( \ket 1 \times (\ite{\ket
                        1}{\ket 0}{\ket 1}) \times \ket 1))))))
                    \end{aligned}
  \\
  \red{\riftrue} &
                   \s{Bob}\ \Uparrow_\ell(\pi_2(\Uparrow_r \s H_1^3(
                   \pair
                   {\pair{\frac{\alpha}{\sqrt 2}.\ket{000}}{\frac{\alpha}{\sqrt 2}.\ket{011}}}
                   {\pair{\frac{\beta}{\sqrt 2}.\ket{110}}{\frac{\beta}{\sqrt 2}.\ket{101}}}
                   )))
  \\
  \red{\rlinr^3} &
                   \s{Bob}\ \Uparrow_\ell(\pi_2(\Uparrow_r
                   \pair
                   {\pair{\s H_1^3(\frac{\alpha}{\sqrt 2}.\ket{000})}{\s H_1^3(\frac{\alpha}{\sqrt 2}.\ket{011})}}
                   {\pair{\s H_1^3(\frac{\beta}{\sqrt 2}.\ket{110})}{\s H_1^3(\frac{\beta}{\sqrt 2}.\ket{101})}}
                   ))
  \\
  \red{\rlinscalr^4} &
                       \s{Bob}\ \Uparrow_\ell(\pi_2(\Uparrow_r
                       \pair
                       {\pair{\frac{\alpha}{\sqrt 2}.\s H_1^3\ket{000}}{\frac{\alpha}{\sqrt 2}.\s H_1^3\ket{011}}}
                       {\pair{\frac{\beta}{\sqrt 2}.\s H_1^3\ket{110}}{\frac{\beta}{\sqrt 2}.\s H_1^3\ket{101}}}
                       ))
  \\
  \red{\rbetab^4} &
                    \begin{aligned}[t]
                      \s{Bob}\ \Uparrow_\ell(\pi_2(\Uparrow_r  ( &\pair
                      {\frac{\alpha}{\sqrt 2}.((\s H(head
                        \ket{000}))\times(tail \ket{000}))}
                      {\frac{\alpha}{\sqrt 2}.((\s H(head
                        \ket{011}))\times(tail \ket{011}))}
                      \\
                      + &\pair {\frac{\beta}{\sqrt 2}.((\s H(head
                        \ket{110}))\times(tail \ket{110}))} {\frac{\beta}{\sqrt
                          2}.((\s H(head \ket{101}))\times(tail \ket{101}))} )
                      ))
                    \end{aligned}
  \\
  \red{\rhead^4} &
                   \begin{aligned}[t]
                     \s{Bob}\ \Uparrow_\ell(\pi_2(\Uparrow_r  ( &\pair
                     {\frac{\alpha}{\sqrt 2}.((\s H\ket 0)\times(tail
                       \ket{000}))} {\frac{\alpha}{\sqrt 2}.((\s H\ket
                       0)\times(tail \ket{011}))}
                     \\
                     + &\pair {\frac{\beta}{\sqrt 2}.((\s H\ket 1)\times(tail
                       \ket{110}))} {\frac{\beta}{\sqrt 2}.((\s H\ket
                       1)\times(tail \ket{101}))} ) ))
                   \end{aligned}
  \\
  \red{\rtail^4} &
                   \begin{aligned}[t]
                     \s{Bob}\ \Uparrow_\ell(\pi_2(\Uparrow_r  ( &\pair
                     {\frac{\alpha}{\sqrt 2}.((\s H\ket 0)\times\ket{00})}
                     {\frac{\alpha}{\sqrt 2}.((\s H\ket 0)\times\ket{11})}
                     \\
                     + &\pair {\frac{\beta}{\sqrt 2}.((\s H\ket
                       1)\times\ket{10})} {\frac{\beta}{\sqrt 2}.((\s H\ket
                       1)\times\ket{01})} ) ) )
                   \end{aligned}
  \\
  \red{\rbetab^4} &
                    \begin{aligned}[t]
                      \s{Bob}\ \Uparrow_\ell(\pi_2(\Uparrow_r  ( & {
                        \frac{\alpha}{\sqrt 2}.(\frac 1{\sqrt 2}.( \ket
                        0+(\ite{\ket 0}{(-\ket 1)}{\ket 1}) )\times\ket{00}) }
                      \\
                      + & { \frac{\alpha}{\sqrt 2}.(\frac 1{\sqrt 2}.( \ket
                        0+(\ite{\ket 0}{(-\ket 1)}{\ket 1}) )\times\ket{11}) }
                      \\
                      + ( & { \frac{\beta}{\sqrt 2}.(\frac 1{\sqrt 2}.( \ket
                        0+(\ite{\ket 1}{(-\ket 1)}{\ket 1}) )\times\ket{10}) }
                      \\
                      + & { \frac{\beta}{\sqrt 2}.(\frac 1{\sqrt 2}.( \ket
                        0+(\ite{\ket 1}{(-\ket 1)}{\ket 1}) )\times\ket{01}) }
                      ) ) ) )
                    \end{aligned}
  \\
  \red{\riffalse^2} &
                      \begin{aligned}[t]
                        \s{Bob}\ \Uparrow_\ell(\pi_2(\Uparrow_r  ( (
                        & {\frac{\alpha}{\sqrt 2}.(\frac 1{\sqrt 2}.( \ket 0+\ket 1)\times\ket{00}) }\\
                        +& {\frac{\alpha}{\sqrt 2}.(\frac 1{\sqrt 2}.( \ket 0+\ket 1)\times\ket{11}) }) \\
                        + ( & {\frac{\beta}{\sqrt 2}.(\frac 1{\sqrt 2}.( \ket 0+(\ite{\ket 1}{(-\ket 1)}{\ket 1}))\times\ket{10}) }\\
                        +& {\frac{\beta}{\sqrt 2}.(\frac 1{\sqrt 2}.( \ket
                          0+(\ite{\ket 1}{(-\ket 1)}{\ket 1}))\times\ket{01})
                        }))))
                      \end{aligned}
  \\
  \red{\riftrue^2} &
                     \begin{aligned}[t]
                       \s{Bob}\ \Uparrow_\ell(\pi_2(\Uparrow_r  ( &\pair
                       {\frac{\alpha}{\sqrt 2}.(\frac 1{\sqrt 2}.\pair{\ket
                           0}{\ket 1}\times\ket{00})} {\frac{\alpha}{\sqrt
                           2}.(\frac 1{\sqrt 2}.\pair{\ket 0}{\ket
                           1}\times\ket{11})}
                       \\
                       + &\pair {\frac{\beta}{\sqrt 2}.(\frac 1{\sqrt
                           2}.\npair{\ket 0}{\ket 1}\times\ket{10})}
                       {\frac{\beta}{\sqrt 2}.(\frac 1{\sqrt 2}.\npair{\ket
                           0}{\ket 1}\times\ket{01})} ) ) )
                     \end{aligned}
  \\
  \red{\rdistcasum^3} &
                        \begin{aligned}[t]
                          \s{Bob}\ \Uparrow_\ell(\pi_2 ( &\pair {\Uparrow_r \frac{\alpha}{\sqrt 2}.(\frac 1{\sqrt 2}.\pair{\ket 0}{\ket 1}\times\ket{00})} {\Uparrow_r \frac{\alpha}{\sqrt 2}.(\frac 1{\sqrt
                              2}.\pair{\ket 0}{\ket 1}\times\ket{11})}
                          \\
                          + &\pair {\Uparrow_r \frac{\beta}{\sqrt 2}.(\frac
                            1{\sqrt 2}.\npair{\ket 0}{\ket 1}\times\ket{10})}
                          {\Uparrow_r \frac{\beta}{\sqrt 2}.(\frac 1{\sqrt
                              2}.\npair{\ket 0}{\ket 1}\times\ket{01})} ) )
                        \end{aligned}
  \\
  \red{\rdistcascal^4} &
                         \begin{aligned}[t]
                           \s{Bob}\ \Uparrow_\ell(\pi_2 ( &\pair
                           {\frac{\alpha}{\sqrt 2}.\Uparrow_r (\frac 1{\sqrt
                               2}.\pair{\ket 0}{\ket 1}\times\ket{00})}
                           {\frac{\alpha}{\sqrt 2}.\Uparrow_r (\frac 1{\sqrt
                               2}.\pair{\ket 0}{\ket 1}\times\ket{11})}
                           \\
                           + &\pair {\frac{\beta}{\sqrt 2}.\Uparrow_r (\frac
                             1{\sqrt 2}.\npair{\ket 0}{\ket 1}\times\ket{10})}
                           {\frac{\beta}{\sqrt 2}.\Uparrow_r (\frac 1{\sqrt
                               2}.\npair{\ket 0}{\ket 1}\times\ket{01})} ) )
                         \end{aligned}
  \\
  \red{\rdistscalr^4} &
                        \begin{aligned}[t]
                          \s{Bob}\ \Uparrow_\ell(\pi_2 ( &\pair {\frac{\alpha}{\sqrt
                              2}.\frac 1{\sqrt 2}.\Uparrow_r \pair{\ket 0}{\ket
                              1}\times\ket{00}} {\frac{\alpha}{\sqrt 2}.\frac
                            1{\sqrt 2}.\Uparrow_r \pair{\ket 0}{\ket
                              1}\times\ket{11}}
                          \\
                          + &\pair {\frac{\beta}{\sqrt 2}.\frac 1{\sqrt 2}.\Uparrow_r \npair{\ket 0}{\ket 1}\times\ket{10}}
                          {\frac{\beta}{\sqrt 2}.\frac 1{\sqrt 2}.\Uparrow_r \npair{\ket 0}{\ket 1}\times\ket{01}} ) )
                        \end{aligned}
  \\
  \red{\rprod^4} &
                   \begin{aligned}[t]
                     \s{Bob}\ \Uparrow_\ell(\pi_2 ( &\pair {\frac{\alpha}2.\Uparrow_r \pair{\ket 0}{\ket 1}\times\ket{00}}
                     {\frac{\alpha}2.\Uparrow_r \pair{\ket 0}{\ket
                         1}\times\ket{11}}
                     \\
                     + &\pair {\frac{\beta}2.\Uparrow_r \npair{\ket 0}{\ket
                         1}\times\ket{10}} {\frac{\beta}2.\Uparrow_r \npair{\ket
                         0}{\ket 1}\times\ket{01}} ) )
                   \end{aligned}
  \\
  \red{\rdistsumr^4} &
                       \begin{aligned}[t]
                         \s{Bob}\ \Uparrow_\ell(\pi_2 ( &\pair
                         {\frac{\alpha}2.\pair{\Uparrow_r \ket{000}}{\Uparrow_r \ket{100}}} {\frac{\alpha}2.\pair{\Uparrow_r \ket{011}}{\Uparrow_r \ket{111}}}
                         \\
                         + &\pair {\frac{\beta}2.\pair{\Uparrow_r \ket{010}}{\Uparrow_r (-\ket{110})}} {\frac{\beta}2.\pair{\Uparrow_r \ket{001}}{\Uparrow_r (-\ket{101})}} ) )
                       \end{aligned}
  \\
  \red{\rdistscalr^2} &
                        \begin{aligned}[t]
                          \s{Bob}\ \Uparrow_\ell(\pi_2 ( &\pair
                          {\frac{\alpha}2.\pair{\Uparrow_r \ket{000}}{\Uparrow_r \ket{100}}} {\frac{\alpha}2.\pair{\Uparrow_r \ket{011}}{\Uparrow_r \ket{111}}}
                          \\
                          + &\pair {\frac{\beta}2.\npair{\Uparrow_r \ket{010}}{\Uparrow_r \ket{110}}}
                          {\frac{\beta}2.\npair{\Uparrow_r \ket{001}}{\Uparrow_r \ket{101}}} ) )
                        \end{aligned}
  \\
  \red{\rcaneutr^8} &
                      \begin{aligned}[t]
                        \s{Bob}\ \Uparrow_\ell(\pi_2 ( &\pair
                        {\frac{\alpha}2.\pair{\ket{000}}{\ket{100}}}
                        {\frac{\alpha}2.\pair{\ket{011}}{\ket{111}}}
                        \\
                        + &\pair {\frac{\beta}2.\npair{\ket{010}}{\ket{110}}}
                        {\frac{\beta}2.\npair{\ket{001}}{\ket{101}}} ) )
                      \end{aligned}
  \\
  \red{\rdists^4} &
                    \begin{aligned}[t]
                      \s{Bob}\ \Uparrow_\ell(\pi_2 ( &\pair
                      {\pair{\frac{\alpha}2.\ket{000}}{\frac{\alpha}2.\ket{100}}}
                      {\pair{\frac{\alpha}2.\ket{011}}{\frac{\alpha}2.\ket{111}}}
                      \\
                      + &\pair
                      {\pair{\frac{\beta}2.\ket{010}}{\frac{\beta}2.(-\ket{110})}}
                      {\pair{\frac{\beta}2.\ket{001}}{\frac{\beta}2.(-\ket{101})}}
                      ) )
                    \end{aligned}
  \\
  \red{\rprod^2} &
                   \begin{aligned}[t]
                     \s{Bob}\ \Uparrow_\ell(\pi_2 ( &\pair
                     {\pair{\frac{\alpha}2.\ket{000}}{\frac{\alpha}2.\ket{100}}}
                     {\pair{\frac{\alpha}2.\ket{011}}{\frac{\alpha}2.\ket{111}}}
                     \\
                     + &\pair
                     {\npair{\frac{\beta}2.\ket{010}}{\frac{\beta}2.\ket{110}}}
                     {\npair{\frac{\beta}2.\ket{001}}{\frac{\beta}2.\ket{101}}}
                     ) )
                   \end{aligned}
  \\
  =_{AC} &
           \begin{aligned}[t]
             \s{Bob}\ \Uparrow_\ell(\pi_2 (
             &\pair {\pair{\frac{\alpha}2.\ket{000}}{\frac{\beta}2.\ket{001}}} {\pair{\frac{\alpha}2.\ket{011}}{\frac{\beta}2.\ket{010}}} \\
             + &\pair
             {\npair{\frac{\alpha}2.\ket{100}}{\frac{\beta}2.\ket{101}}}
             {\npair{\frac{\alpha}2.\ket{111}}{\frac{\beta}2.\ket{110}}}))
           \end{aligned}
\end{align*}

The next rewrite step following rule $\rproj$, may produce one of the following
four results probability $\frac 14$ each:
\begin{description}
\item[(00)] $\s{Bob}\ \Uparrow_\ell\ket{00}\times\pair{\alpha.\ket 0}{\beta.\ket
    1}$
\item[(01)] $\s{Bob}\ \Uparrow_\ell\ket{01}\times\pair{\alpha.\ket 1}{\beta.\ket
    0}$
\item[(10)] $\s{Bob}\ \Uparrow_\ell\ket{01}\times\npair{\alpha.\ket 0}{\beta.\ket
    1}$
\item[(11)] $\s{Bob}\ \Uparrow_\ell\ket{11}\times\npair{\alpha.\ket 1}{\beta.\ket
    0}$
\end{description}

So, in general, $\s{Bob}\ \Uparrow_\ell\ket{xy}\times\pair{\alpha.\ket
  z}{[-]\beta.\ket w}$. Then,
\begin{align*}
  \omit\rlap{$\s{Bob}\ \Uparrow_\ell\ket{xy}\times\pair{\alpha.\ket z}{[-]\beta.\ket w}$}\\
  \red{\rdistsuml} &
                     \s{Bob}\pair{\Uparrow_\ell\ket{xy}\times\alpha.\ket z}{\Uparrow_\ell\ket{xy}\times[-]\beta.\ket w}
  \\
  \red{\rdistscall^2} &
                        \s{Bob}\ \pair{\alpha.\Uparrow_\ell\ket{xyz}}{[-]\beta.\Uparrow_\ell\ket{xyw}}
  \\
  \red{\rcaneutl^2} &
                      \s{Bob}\ \pair{\alpha.\ket{xyz}}{[-]\beta.\ket{xyw}}
  \\
  \red{\rlinr} &
                 \pair{\s{Bob}\ \alpha.\ket{xyz}}{\s{Bob}\ [-]\beta.\ket{xyw}}
  \\
  \red{\rlinscalr} &
                     \pair{\alpha.\s{Bob}\ \ket{xyz}}{[-]\beta.\s{Bob}\ \ket{xyw}}
  \\
  \red{\rbetab^2} &
                    ({\alpha.\s Z^{head\ket{xyz}}(\s{not}^{head\ tail\ket{xyz}}(tail\ tail\ \ket{xyz}))}\\
                   &+[-]\beta.\s Z^{head\ket{xyw}}(\s{not}^{head\ tail\ket{xyw}}(tail\ tail\ \ket{xyw})))
  \\
  \red{\rtail^6} &
                   \pair
                   {\alpha.\s Z^{head\ket{xyz}}(\s{not}^{head\ket{yz}}(\ket z))}
                   {[-]\beta.\s Z^{head\ket{xyw}}(\s{not}^{head\ket{yw}}(\ket w))}
  \\
  \red{\rhead^4} &
                   \pair
                   {\alpha.\s Z^{\ket x}(\s{not}^{\ket y}\ket z)}
                   {[-]\beta.\s Z^{\ket x}(\s{not}^{\ket y}\ket w)}
  \\
  \red{\rbetab^4} &
                    \pair
                    {\alpha.\s Z^{\ket x}(\ite{\ket y}{\s{not}\ket z}{\ket z})}
                    {[-]\beta.\s Z^{\ket x}(\ite{\ket y}{\s{not}\ket w}{\ket w})}
\end{align*}
Cases:
\begin{description}
\item[$(00)$] $\begin{aligned}[t] \omit\rlap{\parbox{\textwidth}{$ \pair
        {\alpha.\s Z^{\ket 0}(\ite{\ket 0}{\s{not}\ket 0}{\ket 0})} {\beta.\s
          Z^{\ket 0}(\ite{\ket 0}{\s{not}\ket 1}{\ket 1})}
        $}}\\
    \red{\riffalse^2} & \pair {\alpha.\s Z^{\ket 0}\ket 0} {\beta.\s Z^{\ket
        0}\ket 1}
    \\
    \red{\rbetab^4} & \pair {\alpha.\ite{\ket 0}{(\s Z\ket 0)}{\ket 0}}
    {\beta.\ite{\ket 0}{(\s Z\ket 1)}{\ket 1}}
    \\
    \red{\riffalse^2} & \pair {\alpha.\ket 0} {\beta.\ket 1}
    \\[5mm]
  \end{aligned}$

\item[$(01)$] $\begin{aligned}[t] \omit\rlap{\parbox{\textwidth}{$ \pair
        {\alpha.\s Z^{\ket 0}(\ite{\ket 1}{\s{not}\ket 1}{\ket 1})} {\beta.\s
          Z^{\ket 0}(\ite{\ket 1}{\s{not}\ket 0}{\ket 0})}
        $}}\\
    \red{\riftrue^2} & \pair {\alpha.\s Z^{\ket 0}(\s{not} \ket 1)} {\beta.\s
      Z^{\ket 0}(\s{not} \ket 0)}
    \\
    \red{\rbetab^2} & \pair {\alpha.\s Z^{\ket 0}(\ite{\ket 1}{\ket 0}{\ket 1})}
    {\beta.\s Z^{\ket 0}(\ite{\ket 0}{\ket 0}{\ket 1})}
    \\
    \red{\riftrue} & \pair {\alpha.\s Z^{\ket 0}\ket 0} {\beta.\s Z^{\ket
        0}(\ite{\ket 0}{\ket 0}{\ket 1})}
    \\
    \red{\riffalse} & \pair {\alpha.\s Z^{\ket 0}\ket 0} {\beta.\s Z^{\ket
        0}\ket 1}
    \\
    \red{\rbetab^4} & \pair {\alpha.\ite{\ket 0}{(\s Z\ket 0)}{\ket 0}}
    {\beta.\ite{\ket 0}{(\s Z\ket 1)}{\ket 1}}
    \\
    \red{\riffalse^2} & \pair {\alpha.\ket 0} {\beta.\ket 1}
  \end{aligned}$

\item[$(10)$] $\begin{aligned}[t] \omit\rlap{\parbox{\textwidth}{$ \npair
        {\alpha.\s Z^{\ket 1}(\ite{\ket 0}{\s{not}\ket 0}{\ket 0})} {\beta.\s
          Z^{\ket 1}(\ite{\ket 0}{\s{not}\ket 1}{\ket 1})}
        $}}\\
    \red{\riffalse^2} & \npair {\alpha.\s Z^{\ket 1}\ket 0} {\beta.\s Z^{\ket
        1}\ket 1}
    \\
    \red{\rbetab^4} & \npair {\alpha.\ite{\ket 1}{(\s Z\ket 0)}{\ket 0}}
    {\beta.\ite{\ket 1}{(\s Z\ket 1)}{\ket 1}}
    \\
    \red{\riftrue^2} & \npair {\alpha.\s Z\ket 0} {\beta.\s Z\ket 1}
    \\
    \red{\rbetab^2} & \npair {\alpha.\ite{\ket 0}{(-\ket 1)}{\ket 0}}
    {\beta.\ite{\ket 1}{(-\ket 1)}{\ket 0}}
    \\
    \red{\riffalse} & \npair {\alpha.\ket 0} {\beta.\ite{\ket 1}{(-\ket 1)}{\ket
        0}}
    \\
    \red{\riftrue} & \npair {\alpha.\ket 0} {\beta.(-\ket 1)}
    \\
    \red{\rprod} & \pair {\alpha.\ket 0} {\beta.\ket 1}
  \end{aligned}$

\item[$(11)$] $\begin{aligned}[t] \omit\rlap{\parbox{\textwidth}{$ \npair
        {\alpha.\s Z^{\ket 1}(\ite{\ket 1}{\s{not}\ket 1}{\ket 1})} {\beta.\s
          Z^{\ket 1}(\ite{\ket 1}{\s{not}\ket 0}{\ket 0})}
        $}}\\
    \red{\riftrue^2} & \npair {\alpha.\s Z^{\ket 1}(\s{not} \ket 1)} {\beta.\s
      Z^{\ket 1}(\s{not} \ket 0)}
    \\
    \red{\rbetab^2} & \npair {\alpha.\s Z^{\ket 1}(\ite{\ket 1}{\ket 0}{\ket
        1})} {\beta.\s Z^{\ket 1}(\ite{\ket 0}{\ket 0}{\ket 1})}
    \\
    \red{\riftrue} & \npair {\alpha.\s Z^{\ket 1}\ket 0} {\beta.\s Z^{\ket
        1}(\ite{\ket 0}{\ket 0}{\ket 1})}
    \\
    \red{\riffalse} & \npair {\alpha.\s Z^{\ket 1}\ket 0} {\beta.\s Z^{\ket
        1}\ket 1}
    \\
    \red{\rbetab^4} & \npair {\alpha.\ite{\ket 1}{(\s Z\ket 0)}{\ket 0}}
    {\beta.\ite{\ket 1}{(\s Z\ket 1)}{\ket 1}}
    \\
    \red{\riftrue^2} & \pair {\alpha.\s Z\ket 0} {\beta.\s Z\ket 1}
    \\
    \red{\rbetab^2} & \npair {\alpha.\ite{\ket 0}{(-\ket 1)}{\ket 0}}
    {\beta.\ite{\ket 1}{(-\ket 1)}{\ket 0}}
    \\
    \red{\riffalse} & \npair {\alpha.\ket 0} {\beta.\ite{\ket 1}{(-\ket 1)}{\ket
        0}}
    \\
    \red{\riftrue} & \npair {\alpha.\ket 0} {\beta.(-\ket 1)}
    \\
    \red{\rprod} & \pair {\alpha.\ket 0} {\beta.\ket 1}
  \end{aligned}$
\end{description}
Hence, in every case, $\s{Teleportation}\ (\alpha.\ket 0+\beta.\ket
1)\lra^*(\alpha.\ket 0+\beta.\ket 1)$ as expected. \medskip

The typing of $\s{Teleportation}$ is given below:

\begin{equation}
  \label{eq:Z}
  \infer[\Rightarrow_E]
  {x:\B\vdash \s Zx:S(\B)}
  {
    \infer[\Rightarrow_I]
    {\vdash \s Z:\B\Rightarrow S(\B)}
    {
      \infer[\Rightarrow_E]
      {y:\B\vdash\ite y{(-\ket 1)}{\ket 0}: S(\B)}
      {
        \infer[\tif]{\vdash\ite{}{(-\ket 1)}{\ket 0}:\B\Rightarrow S(\B)}
        {
          \infer[S_I^\alpha]{\vdash -\ket 1:S(\B)}
          {
            \infer[\tax_{\ket 1}]{\vdash\ket 1:\B}{}
          }
          &
          \infer[\preceq]{\vdash\ket 0:S(\B)}
          {
            \infer[\tax_{\ket 0}]{\vdash\ket 0:\B}{}
          }
        }
        &
        \infer[\tax]{y:\B\vdash y:\B}{}
      }
    }
    &
    \infer[\tax]{x:\B\vdash x:\B}{}
  }
\end{equation}

\begin{equation}
  \label{eq:Zpot}
  \infer[\Rightarrow_E]
  {x:\B^3\vdash \s Z^{(\head\ x)}:\B\Rightarrow S(\B)}
  {
    \infer[\Rightarrow_I]
    {\vdash\lambda y^\B.\lambda w^\B.\ite y{\s Zw}w:\B\Rightarrow\B\Rightarrow S(\B)}
    {
      \infer[\Rightarrow_I]
      {y:\B\vdash\lambda w^\B.\ite y{\s Zw}w:\B\Rightarrow S(\B)}
      {
        \infer[\Rightarrow_E]{y:\B,w:\B\vdash\ite y{\s Zw}w:S(\B)}
        {
          \infer[\tif]{w:\B\vdash\ite{}{\s Zw}w:\B\Rightarrow S(\B)}
          {
            \infer
            {w:\B\vdash \s Zw:S(\B)}
            {\eqref{eq:Z}}
            &
            \infer[\preceq]
            {w:\B\vdash w:S(\B)}
            {
              \infer[\tax]{w:\B\vdash w:\B}{}
            }
          }
          &
          \infer[\tax]{y:\B\vdash y:\B}{}
        }
      }
    }
    &
    \infer[\times_{E_r}]
    {x:\B^3\vdash\head\ x:\B}
    {
      \infer[\tax]
      {x:\B^3\vdash x:\B^3}
      {}
    }
  }
\end{equation}

\begin{equation}
  \label{eq:Notpot}
  \hspace{-1cm}\infer[\Rightarrow_E]
  {x:\B^3\vdash \s{not}^{(\head\ \tail\ x)}:\B\Rightarrow\B}
  {
    \infer[\Rightarrow_I]
    {\vdash\lambda y^\B.\lambda z^\B.\ite y{\s{not}\ z}z:\B\Rightarrow\B\Rightarrow\B}
    {
      \infer[\Rightarrow_I]
      {y:\B\vdash\lambda z^\B.\ite y{\s{not}\ z}z:\B\Rightarrow\B}
      {
        \infer[\Rightarrow_E]
        {y:\B,z:\B\vdash\ite y{\s{not}\ z}z:\B}
        {
          \infer[\tif]
          {z:\B\vdash\ite{}{\s{not}\ z}z:\B\Rightarrow\B}
          {
            \infer[\Rightarrow_E]
            {z:\B\vdash \s{not}\ z:\B}
            {
              \infer
              {\vdash \s{not}:\B\Rightarrow\B}
              {\eqref{eq:not}}
              &
              \infer[\tax]
              {z:\B\vdash z:\B}
              {}
            }
            &
            \infer[\tax]
            {z:\B\vdash z:\B}
            {}
          }
          &
          \infer[\tax]{y:\B\vdash y:\B}{}
        }
      }
    }
    &
    \infer[\times_{E_r}]
    {x:\B^3\vdash\head\ \tail\ x:\B}
    {
      \infer[\times_{E_l}]
      {x:\B^3\vdash\tail\ x:\B^2}
      {\infer[\tax]{x:\B^3\vdash x:\B^3}{}}
    }
  }
\end{equation}

\begin{equation}
  \label{eq:Bob}
  \infer[\Rightarrow_I]
  {\vdash \s{Bob}:\B^3\Rightarrow S(\B)}
  {
    \infer[C]
    {x:\B^3\vdash \s Z^{(\head\ x)}(\s{not}^{(\head\ \tail\ x)}(\tail\ \tail\ x)):S(\B)}
    {
      \infer[\Rightarrow_E]
      {x:\B^3,y:\B^3\vdash \s Z^{(\head\ y)}(\s{not}^{(\head\ \tail\ x)}(\tail\ \tail\ x)):S(\B)}
      {
        \infer
        {y:\B^3\vdash \s Z^{(\head\ y)}:\B\Rightarrow S(\B)}
        {\eqref{eq:Zpot}}
        &
        \hspace{-5mm}\infer[C]
        {x:\B^3\vdash \s{not}^{(\tail\ x)}(\tail\ \tail\ x):\B}
        {
          \infer[\Rightarrow_E]
          {x:\B^3,y:\B^3\vdash \s{not}^{(\tail\ y)}(\tail\ \tail\ x):\B}
          {
            \infer
            {y:\B^3\vdash \s{not}^{(\tail\ y)}:\B\Rightarrow\B}
            {\eqref{eq:Notpot}}
            &
            \infer[\times_{E_l}]
            {x:\B^3\vdash\tail\ \tail\ x:\B}
            {
              \infer[\times_{E_l}]
              {x:\B^3\vdash\tail\ x:\B\times\B}
              {
                \infer[\tax]
                {x:\B^3\vdash x:\B^3}
                {}
              }
            }
          }
        }
      }
    }
  }
\end{equation}

\begin{equation}
  \label{eq:HoneThree}
  \infer[\Rightarrow_I]
  {\vdash \s H_1^3:\B^3\Rightarrow S(\B)\times\B^2}
  {
    \infer[C]
    {x:\B^3\vdash (\s H(\head\ x))\times(\tail\ x):S(\B)\times\B^2}
    {
      \infer[\times_I]
      {x:\B^3,y:\B^3\vdash (\s H(\head\ x))\times(\tail\ y):S(\B)\times\B^2}
      {
        \infer[\Rightarrow_E]
        {x:\B^3\vdash \s H(\head\ x):S(\B)}
        {
          \infer
          {\vdash \s H:\B\Rightarrow S(\B)}
          {\eqref{eq:H}}
          &
          \infer[\times_{E_r}]
          {x:\B^3\vdash\head\ x:\B}
          {
            \infer[\tax]{x:\B^3\vdash x:\B^3}{}
          }
        }
        &
        \infer[\times_{E_l}]
        {y:\B^3\vdash\tail\ y:\B^2}
        {
          \infer[\tax]{y:\B^3\vdash y:\B^3}{}
        }
      }
    }
  }
\end{equation}

\begin{equation}
  \label{eq:cnot}
  \hspace{-3cm}
  \infer[\Rightarrow_I]
  {\vdash \s{cnot}:\B^2\Rightarrow\B^2}
  {
    \infer[C]
    {x:\B^2\vdash (\head\ x)\times(\ite{(\head\ x)}{\s{not}(\tail\ x)}{(\tail\ x)}):\B^2}
    {
      \infer[\times_I]
      {x:\B^2,y:\B^2\vdash (\head\ y)\times(\ite{(\head\ x)}{\s{not}(\tail\ x)}{(\tail\ x)}):\B^2}
      {
        \infer[\times_{E_r}]
        {y:\B^2\vdash\head\ y:\B}
        {
          \infer[\tax]
          {y:\B^2\vdash y:\B^2}
          {}
        }
        &
        \infer[C]
        {x:\B^2\vdash\ite{(\head\ x)}{\s{not}(\tail\ x)}{(\tail\ x)}:\B}
        {
          \infer[\Rightarrow_E]
          {x:\B^2,y:\B^2\vdash\ite{(\head\ y)}{\s{not}(\tail\ x)}{(\tail\ x)}:\B}
          {
            \infer[\tif]{x:\B^2\vdash\ite{}{\s{not}(\tail\ x)}{(\tail\ x)}:\B\Rightarrow\B}
            {
              \infer[\Rightarrow_E]
              {x:\B^2\vdash \s{not}(\tail\ x):\B}
              {
                \infer
                {\vdash \s{not}:\B\Rightarrow\B}
                {\eqref{eq:not}}
                &
                \infer[\times_{E_l}]
                {x:\B^2\vdash\tail\ x:\B}
                {
                  \infer[\tax]
                  {x:\B^2\vdash x:\B^2}
                  {}
                }
              }
              &
              \infer[\times_{E_l}]
              {x:\B^2\vdash\tail\ x:\B}
              {
                \infer[\tax]
                {x:\B^2\vdash x:\B^2}
                {}
              }
            }
            &
            \infer[\times_{E_r}]
            {y:\B^2\vdash\head\ y:\B}
            {
              \infer[\tax]
              {y:\B^2\vdash y:\B^2}
              {}
            }
          }
        }
      }
    }
  }
\end{equation}

\begin{equation}
  \label{eq:cnotOneTwo}
  {
    \hspace{-1.5cm}
    \infer[\Rightarrow_I]
    {\vdash \s{cnot}_{12}^3:\B^3\Rightarrow\B^3}
    {
      \infer[C]
      {x:\B^3\vdash (\s{cnot}((\head\ x)\times(\head\ \tail\ x)))\times(\tail\ \tail\ x):\B^3}
      {
        \infer[\times_I]
        {x:\B^3,y:\B^3\vdash (\s{cnot}((\head\ x)\times(\head\ \tail\ x)))\times(\tail\ \tail\ x):\B^3}
        {
          \infer[\Rightarrow_E]
          {x:\B^3\vdash \s{cnot}((\head\ x)\times(\head\ \tail\ x)):\B^2}
          {
            \infer
            {\vdash \s{cnot}:\B^2\Rightarrow\B^2}
            {\eqref{eq:cnot}}
            &
            \infer[C]
            {x:\B^3\vdash(\head\ x)\times(\head\ \tail\ x):\B^2}
            {
              \infer[\times_I]
              {x:\B^3,y:\B^3\vdash(\head\ x)\times(\head\ \tail\ y):\B^2}
              {
                \infer[\times_{E_r}]
                {x:\B^3\vdash\head\ x:\B}
                {\infer[\tax]{x:\B^3\vdash x:\B^3}{}}
                &
                \infer[\times_{E_r}]
                {y:\B^3\vdash\head\ \tail\ y:\B}
                {
                  \infer[\times_{E_l}]
                  {y:\B^3\vdash\tail\ y:\B^2}
                  {
                    \infer[\tax]
                    {y:\B^3\vdash y:\B^3}
                    {}
                  }
                }
              }
            }
          }
          &
          \infer[\times_{E_l}]
          {y:\B^3\vdash \tail\ \tail\ y:\B}
          {
            \infer[\times_{E_l}]
            {y:\B^3\vdash \tail\ y:\B^2}
            {
              \infer[\tax]
              {y:\B^3\vdash y:\B^3}
              {}
            }
          }
        }
      }
    }
  }
\end{equation}

\begin{equation}
  \label{cnotOneTwoX}
  \infer[\Rightarrow_{ES}]
  {x:S(\B)\times S(\B^2)\vdash \s{cnot}_{12}^3\Uparrow_\ell \Uparrow_r x:S(\B^3)}
  {
    \infer[\preceq]
    {\vdash \s{cnot}_{12}^3: S(\B^3\Rightarrow\B^3)}
    {
      \infer
      {\vdash \s{cnot}_{12}^3:\B^3\Rightarrow\B^3}
      {\eqref{eq:cnotOneTwo}}
    }
    &
    \infer[\Uparrow_\ell]
    {x:S(\B)\times S(\B^2)\vdash\Uparrow_\ell \Uparrow_r  x:S(\B^3)}
    {
      \infer[\Uparrow_r]
      {x:S(\B)\times S(\B^2)\vdash\Uparrow_r  x:S(\B\times S(\B^2))}
      {
        \infer[\preceq]
        {x:S(\B)\times S(\B^2)\vdash x:S(S(\B)\times S(\B^2))}
        {
          \infer[\tax]
          {x:S(\B)\times S(\B^2)\vdash x:S(\B)\times S(\B^2)}
          {}
        }
      }
    }
  }
\end{equation}

\begin{equation}
  \label{eq:Alice}
  \infer[\Rightarrow_I]
  {\vdash\s{Alice}:S(\B)\times S(\B^2)\Rightarrow\B^2\times S(\B)}
  {
    \infer[S_E]
    {x:S(\B)\times S(\B^2)\vdash\pi_2(\Uparrow_r \s H_1^3(\s{cnot}_{12}^3\ \Uparrow_\ell \Uparrow_r x)):\B^2\times S(\B)}
    {
      \infer[\Uparrow_r]
      {x:S(\B)\times S(\B^2)\vdash \Uparrow_r \s H_1^3(\s{cnot}_{12}^3\ \Uparrow_\ell \Uparrow_r x):S(\B^3)}
      {
        \infer[\Rightarrow_{ES}]
        {x:S(\B)\times S(\B^2)\vdash \s H_1^3(\s{cnot}_{12}^3\ \Uparrow_\ell \Uparrow_r x):S(S(\B)\times\B^2)}
        {
          \infer[\preceq]
          {\vdash \s H_1^3:S(\B^3\Rightarrow S(\B)\times\B^2)}
          {
            \infer
            {\vdash \s H_1^3:\B^3\Rightarrow S(\B)\times\B^2}
            {\eqref{eq:HoneThree}}
          }
          &
          \infer[]
          {x:S(\B)\times S(\B^2)\vdash \s{cnot}_{12}^3\ \Uparrow_\ell \Uparrow_r x:S(\B^3)}
          {\eqref{cnotOneTwoX}}
        }
      }
    }
  }
\end{equation}

\begin{equation}
  \label{eq:Bell}
  \infer[\preceq]
  {\vdash\beta_{00}:S(\B^2)}
  {
    \infer[S_I^+]
    {\vdash\beta_{00}:S(S(\B^2))}
    {
      \infer[S_I^\alpha]
      {\vdash\frac 1{\sqrt 2}.\ket{00}:S(\B^2)}
      {
        \infer[\times_I]
        {\vdash\ket{00}:\B^2}
        {
          \infer[\tax_{\ket 0}]{\vdash\ket 0:\B}{}
          &
          \infer[\tax_{\ket 0}]{\vdash\ket 0:\B}{}
        }
      }
      &
      \infer[S_I^\alpha]
      {\vdash\frac 1{\sqrt 2}.\ket{11}:S(\B^2)}
      {
        \infer[\times_I]
        {\vdash\ket{11}:\B^2}
        {
          \infer[\tax_{\ket 1}]{\vdash\ket 1:\B}{}
          &
          \infer[\tax_{\ket 1}]{\vdash\ket 1:\B}{}
        }
      }
    }
  }
\end{equation}

\[
  \hspace{-2cm}\infer[\Rightarrow_I]
  {\vdash Teleportation: S(\B)\Rightarrow S(\B)}
  {
    \infer[\preceq]
    {q:S(\B)\vdash \s {Bob} \ (\Uparrow_\ell \s {Alice} \ (q\times\beta_ {00})):S(\B)}
    {
      \infer[\Rightarrow_ {ES}]
      {q:S(\B)\vdash \s {Bob} \ (\Uparrow_\ell \s {Alice} \ (q\times\beta_ {00})):S(S(\B))}
      {
        \infer[\preceq]
        {\vdash \s {Bob} :S(\B^3\Rightarrow S(\B))}
        {
          \infer {\vdash \s {Bob} :\B^3\Rightarrow S(\B)}
          {\eqref {eq:Bob}}
        }
        &
        \infer[\Uparrow_\ell]
        {q:S(\B)\vdash \Uparrow_\ell  \s {Alice} \ (q\times\beta_ {00}): S(\B^3)}
        {
          \infer[\preceq]
          {q:S(\B)\vdash \s {Alice} \ (q\times\beta_ {00}): S(\B^2\times S(\B))}
          {
            \infer[\Rightarrow_E]
            {q:S(\B)\vdash \s {Alice} \ (q\times\beta_ {00}): \B^2\times S(\B)}
            {
              \infer {\vdash \s {Alice} :S(\B)\times S(\B^2)\Rightarrow\B^2\times S(\B)}
              {\eqref {eq:Alice}}
              &
              \infer[\times_I]
              {q:S(\B)\vdash q\times\beta_ {00} :S(\B)\times S(\B^2)}
              {
                \infer[\tax]
                {q:S(\B)\vdash q:S(\B)}
                {}
                &
                \infer {\vdash\beta_ {00} :S(\B^2)}
                {\eqref {eq:Bell}}
              }
            }
          }
        }
      }
    }
  }
\]

\end{document}